\documentclass[12pt]{article}

\usepackage{authblk}

\usepackage{amsmath,amssymb,amsfonts,amsthm,mathtools}
\usepackage{stmaryrd}

\usepackage[T1]{fontenc}     
\usepackage[utf8]{inputenc}  
\usepackage{newtxtext}       
\usepackage{newtxmath}       

\usepackage{dsfont}          
\usepackage{bbm}             
\usepackage{mathrsfs}        
\usepackage{graphicx}
\usepackage{xcolor}          
\usepackage{wrapfig}
\usepackage{booktabs}
\usepackage{subcaption}      
\usepackage{accents}
\usepackage{enumitem}
\usepackage{datetime}
\usepackage{titlesec}

\usepackage{hyperref}
\newcommand{\email}[1]{%
  \begingroup\urlstyle{same}%
  \href{mailto:#1}{\nolinkurl{#1}}%
  \endgroup
}
\usepackage{cleveref}        
\hypersetup{
    linktoc=all,
    linkcolor=blue,          
    citecolor=blue,        
    filecolor=blue,      
    urlcolor=blue,
    colorlinks=true           
}

\crefname{section}{Section}{Sections}
\crefname{subsection}{Subsection}{Subsections}
\crefname{appendix}{Appendix}{Appendices}

\crefname{theorem}{Theorem}{Theorems}
\crefname{proposition}{Proposition}{Propositions}
\crefname{corollary}{Corollary}{Corollaries}
\crefname{lemma}{Lemma}{Lemmas}
\crefname{definition}{Definition}{Definitions}
\crefname{table}{Table}{Tables}
\crefname{assumption}{Assumption}{Assumptions}
\crefname{remark}{Remark}{Remarks}

\crefname{theorem}{Theorem}{Theorems}

\theoremstyle{thmstyleone}%
\newtheorem{theorem}{Theorem}[section]
\newtheorem{proposition}[theorem]{Proposition}%
\newtheorem{lemma}[theorem]{Lemma}%

\theoremstyle{thmstyletwo}%

\theoremstyle{definition}
\newtheorem{remark}[theorem]{Remark}

\theoremstyle{thmstylethree}%
\newtheorem{corollary}[theorem]{Corollary}
\theoremstyle{definition}
\newtheorem{definition}[theorem]{Definition}%
\newtheorem{example}{Example}%
\numberwithin{equation}{section}
\newcommand{\R}{\mathbb{R}}
\newcommand{\mR}{\mathcal{R}}
\newcommand{\M}{\mathcal{M}}

\renewcommand{\L}{\mathscr{L}}

\newcommand{\E}{\mathbb{E}}
\renewcommand{\P}{\mathbb{P}}

\newcommand{\ps}[1]{\langle #1 \rangle}

\newcommand{\bi}{\begin{itemize}}
\newcommand{\ei}{\end{itemize}}

\newcommand{\law}{\stackrel{\mathrm{law}}{=}}
\newcommand{\w}{\omega}

\renewcommand{\S}{\mathcal{S}}

\newcommand{\pp}[2]{\left(\!\left(#1 , #2\right)\!\right)}

\newcommand{\norm}[1]{\| #1 \|}

\def\eps{\varepsilon}

\newcommand{\introsec}[1]{%
  \par                                             
  \refstepcounter{subsection}%
  \addcontentsline{toc}{subsection}%
    {\protect\numberline{\thesubsection}#1}%
  \bigskip\noindent\textit{\thesubsection.\ #1.}\quad%
  \ignorespaces
}

\begin{document}
\title{Renormalization Group on Wiener Space: \\[0.2cm]
Spectral Theory and Universality}
\author[1,2*]{André L. P. Considera}
\author[1]{Alexei A. Mailybaev}
\affil[1]{Instituto de Matemática Pura e Aplicada (IMPA), Rio de Janeiro, Brazil}
\affil[2]{Université Côte d'Azur, Nice, France}
\date{}

%
%
%

\maketitle

\begin{center}
  \raisebox{1em}[0pt][0pt]{\large\Affilfont
    *Corresponding author(s). E-mail(s): \email{andre.luis@impa.br};}
\end{center}

\begin{abstract}
  We develop a rigorous renormalization group (RG)
  formalism on Wiener space. We introduce an RG operator $\mathcal R$ 
  acting on
  probability measures over continuous paths, 
  whose fixed point is
  the Wiener measure. 
  We linearize $\mathcal{R}$ around the Brownian fixed point and
  analyze the spectral structure of the resulting operator $\mathscr{L}$. 
  Its eigenvectors are too singular to be realized as honest
  measures on path space, 
  and we therefore develop a generalized spectral theory
  within the framework of white noise analysis (Hida calculus).
  The eigenvectors of $\mathscr{L}$ are realized as Hida
  distributions, satisfying $\L U=\lambda U$ in the weak sense.
  This analysis yields structural results such as a spectral gap and
  a diagonal-concentration property of the eigenvectors.
  Moreover, we identify a family $\{\mathfrak{U}_n\}_{n \geq 0}$ of
  eigenvectors with eigenvalues $\lambda_n = 2^{1-n/2}$, consisting of Wick polynomials of white noise
  formally expressed as $\mathfrak{U}_n = \int_0^1
  {:}\dot W(t)^n{:}\, dt$, and rigorously constructed as Hida
  distributions. 
  We then use this spectral structure to uncover a finer, second
  layer of universality in the Donsker invariance
  principle: not only is the convergence of random walks towards the
  Brownian scaling limit universal, but the entire hierarchy of
  leading corrections to the Brownian limit is universal as well,
  governed by the eigenpairs $(\lambda_n, \mathfrak{U}_n)$.
  We also show, at a formal level, that the framework developed here
  extends to Gibbs-type perturbations of the Brownian fixed point, in
  the spirit of self-interacting quantum field theories. 
  In particular, we recover the standard irrelevant/marginal/relevant 
  classification, usually obtained in the physics literature by power-counting, 
  solely from the spectral analysis of $\mathscr{L}$.
\end{abstract}
\newpage

\tableofcontents

\section{Introduction}\label{sec:intro}

\introsec{Rigorous renormalization group}
Renormalization group (RG) is a powerful conceptual and computational 
framework, originally developed in quantum field theory (QFT) and statistical 
mechanics, to analyze the behavior of physical systems across different 
scales.
In its essence, the RG formalism describes how the effective description of a 
system evolves when small-scale degrees of freedom are integrated out, 
revealing large-scale, often universal, structures. 
Iterating this scale-by-scale coarse-graining defines
an RG flow whose fixed points correspond to scale-invariant
theories and provide effective descriptions of the system's large-scale behavior.

While one of the most influential ideas in theoretical physics,
RG methods have historically lacked a rigorous mathematical
foundation. This lack of rigor stems from 
the very nature of the formalism: RG
involves the manipulation of ill-defined objects such as divergent
power series and formal path integrals on infinite-dimensional spaces. 
Typically, RG procedures operate at
a heuristic level, relying on formal perturbative expansions and scaling
arguments that evade direct mathematical control. 
As emphasized by Kupiainen \cite{Kupiainen2023}, 
these methods were developed in the service of physical intuition and
computational success, rather than deductive precision.

The development of mathematically rigorous formulations of RG
remains an active and challenging area of research.
Considerable effort has been devoted to placing RG calculations on a rigorous
mathematical footing, and this effort has often led to new
mathematics of independent interest 
\cite{bleher1973investigation,glimm1987quantum,gawedzki1980rigorous,gawedzki1981renormalization,gawedzki1982renormalization,gawedzki1983block,gawedzki1985massless,brydges1990grad,sinai2014theory, bauerschmidt2019introduction, Kupiainen2016, AizenmanDuminilCopin2021}. 
In the context of stochastic partial differential equations (SPDEs), 
perhaps the most celebrated recent example is
Hairer's theory of regularity structures \cite{Hairer2014}, which
turns RG ideas and intuition into a rigorous
mathematical framework for analyzing singular SPDEs.
This shows that, while RG remains non-rigorous in its traditional form,
it can serve as a conceptual blueprint for constructing
mathematically robust theories in infinite-dimensional analysis.

In this work, we contribute to this program by developing
a rigorous RG formalism on Wiener space and uncovering
a finer, second layer of universality in the classical Donsker
invariance principle: not only is the convergence of random walks
towards the Brownian scaling limit universal, but the entire hierarchy
of leading corrections to the Brownian limit is universal as well,
governed by the spectral data of the linearized RG operator.

Let us briefly describe our approach.
We introduce an RG operator acting on probability
measures over continuous paths, naturally induced by the dyadic
cascade structure of the Lévy--Ciesielski construction of Brownian
motion. 
The Wiener measure then emerges as the fixed point of the
resulting RG dynamics.
At this stage, we follow the standard
dynamical-systems program:
we linearize around the Brownian fixed point and analyze
the spectrum of the linearized operator in order to describe the
asymptotic behavior of nearby orbits.
In this light, Donsker's theorem is interpreted as a stability
statement: random walks are stable perturbations of the
Brownian fixed point, shrinking under the RG flow.
More specifically, we show that, asymptotically, random walks
approach the Brownian fixed point through perturbative directions
given by Wick polynomials of white noise. These directions are
eigenvectors of the linearized RG operator with eigenvalues of modulus
less than one, and are therefore stable.

The Wick-polynomial eigenmodes turn out to be too singular to be realized as
honest measures on Wiener space, forcing us into a larger
distributional setting. For this reason, we work within the framework of white
noise analysis (Hida calculus), a distributional theory on Wiener
space designed precisely to handle Wick-renormalized polynomials of
white noise and other singular Brownian functionals \cite{kuo2018white,obata1994white}. It is the
natural context in which the spectral theory developed in this
paper takes place.

With the correct functional-analytic context in place,
we then show that the cumulant structure of the Donsker
approximations is organized by the same spectral data. More precisely,
a perturbation in the direction of the $m$-th Wick-polynomial
eigenmode excites precisely the $m$-th cumulant of the walk, while the
corresponding eigenvalue determines the decay rate of this contribution
under the RG flow. These cumulant asymptotics translate directly into
correction terms around the Brownian limit. Thus the universality of
Donsker's theorem extends beyond convergence to the Brownian limit:
the hierarchy of correction terms is universal as well.

Let us mention in passing that we also analyze, at a formal level,
the stability of Gibbs-type perturbations of the Brownian fixed point.
These perturbations are obtained by tilting the Brownian measure with
Wick-polynomial interactions, in the same spirit as self-interacting
QFTs. They are classified by the same spectral stability criterion as
the random-walk perturbations. In this case, we recover, formally and
by purely dynamical-systems arguments, the familiar
irrelevant/marginal/relevant classification obtained by standard
power-counting analysis.

Before entering the path-space setting and stating the main
results, it is worth recalling the RG formulation of the Central
Limit Theorem (CLT). It provides a finite-dimensional prototype in which
nearly all the ingredients of the dynamical-systems picture
described above can be seen explicitly. In this case, the RG machinery is
fully rigorous, free from the technical complications intrinsic to
the infinite-dimensional context. It also serves as
a guiding analogy throughout the paper, since many of its basic
objects have direct counterparts in the infinite-dimensional
case developed in the sequel.

\introsec{RG formulation of the central limit theorem}
Following Sinai \cite[pp.~131--132]{sinai1991probability}, we sketch the RG 
formulation of the CLT; see also \cite{jona2001renormalization}.
Let $\{X_n\}_{n\geq 1}$ be a sequence of centered 
i.i.d. random variables on $\R$ with unit variance. 
Their partial sums are given by 
\begin{equation*}
    S_k = \frac{X_1 + \dots + X_{2^k}}{2^{k/2}},
\end{equation*}
and we have
\begin{equation}\label{eq:rg_intro}
    S_{k+1}= 2^{-1/2}\left(S'_k + S''_k\right),
\end{equation}
where $S'_k$ and $S''_k$ are i.i.d. random variables given by 
$S'_k=2^{-k/2}\sum_{n=1}^{2^k}X_n$ and $S''_k=2^{-k/2}\sum_{n=2^k+1}^{2^{k+1}}X_n$.
Denote by $\rho_k$ the density of $S_k$. 
Then, at the level of probability densities, \eqref{eq:rg_intro} becomes
\begin{equation*}
\rho_{k+1}(x)= \mR[\rho_k](x):=\sqrt{2}\,\rho_k * \rho_k\bigl(\sqrt{2}x\bigr),
\end{equation*}
where $\mR$ is the RG operator acting on 
probability densities.
One can easily check that the standard Gaussian density $\rho_G$
satisfies the fixed point equation $\mR[\rho_G]=\rho_G$.
We can thus perform a linearization around $\rho_G$ 
by considering small perturbations of the form $\rho=\rho_G +\eps \eta$.
The linearized RG operator may be written as
\begin{equation*}
    \L\eta(x)=2\sqrt{2}\, \eta*\rho_G\left(\sqrt{2}x\right). 
\end{equation*}
Its eigenfunctions are the Hermite functions $\phi_n(x)=H_n(x) e^{-x^2/2}$, 
where $H_n$ is the $n$-th Hermite polynomial.
The corresponding eigenvalues are $\lambda_n=2^{1-n/2}$, $n=0,1,2,\dots$.
As a result, the operator $\L$ is unstable along the directions 
$\phi_0$ and $\phi_1$, marginally stable along $\phi_2$, and stable in 
all remaining directions.
Consequently, to ensure spectral stability of the Gaussian fixed point, 
one must consider perturbations $\eta$ which are orthogonal to 
$\phi_0$, $\phi_1$ and $\phi_2$. 
This means that the Gaussian fixed point is stable under perturbations that 
preserve the total mass --- so that it remains a probability density --- as well as the first and second moments.
We may then conclude that if $\rho$ is a probability density function 
with zero mean and unit variance, and is sufficiently close 
to $\rho_G$, then $\mR^n[\rho] \to \rho_G$ as $n \to \infty$.

Anticipating some of what is to follow, we briefly indicate how
this pattern carries over to the infinite-dimensional setting
developed in the sequel. The Gaussian density is replaced by the
Wiener measure, and the convolution-rescaling operator
$\mathcal{R}$ on densities by an RG operator on path measures,
induced by the dyadic Lévy--Ciesielski cascade, also involving
convolution and path rescaling. The linearization $\mathscr{L}$ retains
the same geometric sequence of eigenvalues $\lambda_n = 2^{1-n/2}$,
with the role of the Hermite functions $\phi_n$ now played by the
Wick polynomials of white noise, interpreted as Hida distributions. 
The unstable/marginal/stable trichotomy carries over unchanged. In this
correspondence, Donsker's invariance principle plays the role of
the CLT, with random walks taking the place of sums of i.i.d.\
random variables. The second layer of universality discussed
previously then emerges as the path-space counterpart of this
Sinai-type RG analysis.

We now turn to the Wiener space setting and describe our main results.

\introsec{Main results}\label{sec:main-results}

Let $\Omega$ be the space of continuous paths on $[0,1]$ 
vanishing at the origin.
We denote the Wiener measure on $\Omega$
by $\mu_\infty$. The associated Cameron--Martin space is the
Hilbert space $\mathfrak{H}$ of absolutely continuous paths with
square-integrable derivative. The triple
$(\Omega, \mathfrak{H}, \mu_\infty)$ 
is the classical Wiener space associated with Brownian motion.

The construction of the RG operator is based on the
Lévy--Ciesielski representation of Brownian motion, which
realizes Brownian paths as the sum of two parts: a single
large-scale Gaussian mode, and a fluctuating component
consisting of a Brownian bridge built from the dyadic
Faber--Schauder wavelets (tent functions);
see \cref{sec:RG operator}. 
This representation gives rise to a sequence of probability
laws $(\mu_N)_{N\geq1}$
converging to the Wiener measure $\mu_\infty$
and obeying the cascade relation
\begin{equation}\label{eq:cascade-intro}
    \mu_{N+1}
    =
    (F_{\#}\mu_N)*(G_{\#}\mu_N).
\end{equation}
Here, the functions $F$ and $G$ are suitable 
path-rescaling maps defined in \eqref{eq:FandG},
compressing trajectories onto the left and right halves of
the unit interval, respectively. 
The notation $*$ denotes convolution of probability measures, 
and $\#$ denotes pushforward.

The cascade relation above 
suggests introducing the following 
RG operator acting on probability measures over $\Omega$,
\[
\mathcal R[\mu] = (F_{\#}\mu)*(G_{\#}\mu).
\]
Observe that $\mathcal R$ 
is defined on the full space of probability measures
over $\Omega$ and not just on the sequence $(\mu_{N})$.
The cascade dynamics is recovered as a special case: 
advancing in $N$ is simply an application of $\mathcal{R}$
and one has 
$\mu_{N+1}=\mathcal R[\mu_N]=\mathcal R^N[\mu_1]$.
More generally, $\mathcal R$ defines a discrete 
dynamical system on the space of 
probability measures over $\Omega$,
rendering meaningful 
the classical notions of dynamical systems theory such as
orbits, fixed point attractors, basins of attraction, 
and stability.
This reformulation of the cascade as 
a discrete dynamical system also provides a 
clear program to follow, namely,
identify the fixed points of $\mathcal{R}$, 
linearize around them, 
and analyze the spectrum of the 
linearized operator to 
describe the asymptotic behavior of nearby orbits.
Universality statements --- a central theme in statistical physics 
and the theory of critical phenomena --- 
are then reinterpreted as stability of fixed points, 
with different universality classes corresponding to distinct 
fixed points and their respective basins of attraction.
In what follows, we carry out this program rigorously in 
the infinite-dimensional setting of Wiener space 
and apply it to derive a newfound layer of universality 
in the classical Donsker invariance principle.
\bigskip

We now turn our attention to the Brownian fixed point 
and its linearization.
Passing the cascade relation to the limit, one finds that
the Wiener measure satisfies the fixed-point
relation
\[
\mathcal R[\mu_\infty]=\mu_\infty.
\]
We linearize $\mathcal R$ around $\mu_\infty$ by writing
$\mu=\mu_\infty+\varepsilon\nu$ and retaining terms of
order $\varepsilon$. This yields the linearized RG operator
\begin{equation}\label{eq:L-intro}
    \L\nu
    =
    (F_{\#}\nu)*(G_{\#}\mu_\infty)
    +
    (F_{\#}\mu_\infty)*(G_{\#}\nu),
\end{equation}
governing the first-order dynamics of perturbations
around the fixed point. 
To study the stability of
$\mu_\infty$, one is naturally led to analyze the spectral
properties of $\L$ as an operator acting on signed measures,
seeking measure-valued solutions to the eigenvalue problem. 
However, as the spectral analysis conducted 
in Sects.~\ref{sec:gen-spec} and \ref{sec:wick-power}
reveals, the operator $\L$
admits a large family of physically relevant eigenmodes, 
most of which turn out to be too singular to be 
represented as honest measures on $\Omega$.
As a result, one is forced to work in a larger distributional space
of generalized Brownian functionals,
broad enough to accommodate these singular objects.

The framework of white noise analysis (Hida calculus) 
presents itself as the ideal setting to handle this issue.
Initiated by Hida in the 1970s \cite{hida1975analysis}
and subsequently extended by Potthoff, Streit, Kondratiev and others 
\cite{potthoff1991characterization,kondratiev1993spaces,kondratiev1996generalized}, 
its development was originally motivated by 
the need to give rigorous meaning to path integrals 
and singular stochastic processes. 
In other words, the theory was designed precisely to handle 
objects that are too irregular to be functions or measures 
in the classical sense. 
In the present context, 
it provides the correct setting for the generalized spectral theory 
of $\mathscr{L}$ and a natural description of its generalized eigenmodes.
The classical development of the theory can be found in 
\cite{hida1993white,kuo2018white,obata1994white}; 
\cite{dinunno2009malliavin,holden2010stochastic} 
offer modern introductory accounts emphasizing connections 
to Malliavin calculus.

To put it simply, 
white noise analysis provides us with
the space $(\S)$ of Hida stochastic test functions, 
together with its topological dual $(\S)^*$ --- the space of
Hida stochastic distributions. 
They are arranged in the Gelfand 
triple structure
\[
(\S) \subset L^2(\Omega,\mu_\infty)\subset (\S)^*,
\]
the infinite-dimensional stochastic analogue of 
the Schwartz--Gelfand triple in classical distribution theory, 
here for Brownian functionals $\Phi(\omega)$ defined on $\Omega$. 
More specifically, the test space $(\S)$ 
is a nuclear countably Hilbert (Fréchet) space, 
obtained as the projective limit of a family of Hilbert spaces, 
in direct analogy with the construction of the Schwartz 
space via the harmonic oscillator Hamiltonian.
On the other hand, its dual $(\S)^*$, 
the space of Hida distributions, 
is represented as the corresponding
inductive limit; its elements are generalized 
Brownian functionals, not necessarily square-integrable,
and thus need not belong to $L^2(\Omega,\mu_\infty)$ in general.
The construction of these tailor-made functional spaces
revolves around the Wiener chaos expansion.
Test functionals 
$\Phi \in (\mathcal{S})$ and distributions
$U \in (\mathcal{S})^*$ admit Wiener chaos 
decompositions of the form
\begin{equation*}
    \Phi = \sum_{n \geq 0} I_n(f_n), \qquad 
    U = \sum_{n \geq 0} I_n(u_n),
\end{equation*}
where the kernels $f_n$ are tensors belonging to 
appropriate deterministic test spaces on $[0,1]^n$, 
and the $u_n$ are generalized tensors belonging to 
the associated deterministic distribution spaces.
This decomposition is the main engine of analysis 
and permeates the entire edifice of the Hida white-noise framework;
we refer the reader to \cref{sec:hida} for a detailed construction 
of the relevant functional spaces.
\bigskip

Coming back to $\L$ and its spectral theory,
the framework described above is rich enough to
formulate the eigenvalue problem in a natural way.
\begin{definition}[Generalized eigenvectors and spectrum of $\L$]\label{def:gen-spec-intro}
    Let $\pp{\cdot}{\cdot}$ denote the dual pairing between 
    $(\S)^*$ and $(\S)$. 
    We say that $U\in (\S)^*\setminus \{0\}$ is a \emph{generalized 
    eigenvector} of $\L$ with eigenvalue $\lambda \in \mathbb{C}$ if
    \begin{equation}\label{eq:eigenvector-intro}
        \pp{\L U}{\Phi} = \lambda \pp{U}{\Phi},
        \quad \text{for all }  \Phi \in (\S).
    \end{equation}
    The \emph{generalized spectrum} of $\L$, denoted by $\sigma(\L)$, 
    is the set of all $\lambda \in \mathbb{C}$ for which 
    \eqref{eq:eigenvector-intro} holds for some $U \in (\S)^* \setminus \{0\}$.
\end{definition}
Let us take a moment and comment on \cref{def:gen-spec-intro}.
At face value, the eigenvalue problem 
$\mathscr{L}\nu = \lambda\nu$ is posed on signed
measures, which are themselves distributional objects --- 
dual elements of continuous functions on $\Omega$. 
In this sense, the problem lives in a distribution space 
from the very beginning, albeit a small one.
Hence, by working with \cref{def:gen-spec-intro},
one is not introducing distributional methods 
in an ad-hoc manner, but rather
enlarging the solution space of the eigenvalue problem 
from measures to Hida distributions ---
an enlargement within the same distributional paradigm.
This procedure is classical in the theory of 
partial differential equations:
just as distributional solutions extend
classical solutions by transferring differential operators 
onto smooth test functions via adjoint duality,
here Hida-distributional solutions extend measure-valued ones 
by transferring $\mathscr{L}$ onto Hida test functionals 
through its dual $\mathscr{L}^*$; see \cref{sec:dual-operator}.

The following theorem, one of the main results of this work, 
characterizes the generalized 
spectral structure of $\L$ in the sense above.
We state it in a simplified form to avoid getting bogged down in technicalities; 
the precise statements can be found 
in \cref{thm:spectral_gap,thm:diag_support_eigenkernel,thm:eigenvector}.

\begin{theorem}[Generalized spectral structure, simplified form]\label{thm:intro-spec}
    The linearized RG operator $\mathscr{L}$ has the following 
    generalized spectral structure:
    \begin{enumerate}[label=(\roman*)]
    \item \textbf{Spectral gap.} 
        $\sigma(\L) \subseteq \{z \in \mathbb{C} : |z| \leq \sqrt{2}\} \cup \{2\}$, 
        and $\lambda_0 = 2$ is a simple eigenvalue with eigenspace spanned by $\mu_\infty$
        (see \cref{thm:spectral_gap}).

    \item \textbf{Diagonal concentration of eigenvectors.} 
        Let $U = \sum_{n \geq 0} I_n(u_n)$ be a generalized 
        eigenvector of $\mathscr{L}$ with nonzero eigenvalue. 
        Then, for every $n \geq 2$, the kernel $u_n$ is 
        supported on the time diagonal 
        $\{t_1 = \cdots = t_n\} \subset [0,1]^n$
        (see \cref{thm:diag_support_eigenkernel}).

    \item \textbf{Wick powers as eigenvectors.} 
        For each $n \geq 0$, let $\mathfrak{U}_n \in (\mathcal{S})^*$ 
        be the Hida distribution given formally by the time integral 
        of the $n$-th Wick power of white noise,
        \begin{equation}\label{eq:un-intro}
            \mathfrak{U}_n = \int_0^1 {:}\dot{W}(t)^n{:}\,dt,
        \end{equation}
        where $W$ is a Brownian motion.
        Then, for each $n \geq 0$, $\mathfrak{U}_n$ is a generalized 
        eigenvector of $\mathscr{L}$ with eigenvalue
        \begin{equation*}
            \lambda_n = 2^{1 - n/2},
        \end{equation*}
        (see \cref{thm:eigenvector}).

    \end{enumerate}
\end{theorem}

We now discuss the content of \cref{thm:intro-spec} and 
the overarching methodology.
Unlike the self-adjoint or normal operators of classical spectral 
theory in Hilbert spaces, $\mathscr{L}$ 
does not enjoy any type of spectral theorem
at this level of generality,
and no general theory supplies 
us with a complete description of its spectral data. 
Items (i) and (ii) of \cref{thm:intro-spec}
are the most general structural results we can extract in this 
setting, obtained without invoking any form of spectral theorem.
This absence of general structure also shapes item (iii).
Without a general characterization or completeness theorem 
to rely on, 
we are forced to proceed by constructing explicit eigenmodes, 
and the family $\{\mathfrak{U}_n\}$ 
appearing in (iii) is representative rather than exhaustive.

Item (i) is a 
spectral-gap statement. Apart from the leading eigenvalue 
$\lambda_0 = 2$, the entire spectrum of $\mathscr{L}$, 
in the sense of \cref{def:gen-spec-intro},
is confined 
to the closed disk of radius $\sqrt{2}$. 
We observe that the
gap is sharp, since $\lambda_1 = \sqrt{2}$ lies exactly on the 
boundary of the disk, so the gap size cannot be increased. 
Moreover, the separation of $\lambda_0 = 2$ from the rest of the spectrum 
identifies $\mu_\infty$ as the unique maximally unstable direction, 
isolated from the rest of the dynamics. 
In other words, under iterations of $\L$, 
a perturbation in the direction of $\mu_{\infty}$ 
is amplified by a factor of $2$ per step,
and there is no faster-growing direction. 
Item (ii), on the other hand, is a rigidity result about the shape of 
eigenvectors, roughly saying that generalized eigenvectors 
of $\mathscr{L}$ must be concentrated on the time diagonal.

Part (iii) of \cref{thm:intro-spec} identifies a particular 
family $\{\mathfrak{U}_n\}$ of Hida distributions  
as generalized eigenvectors of 
$\mathscr{L}$, namely, time integrals of Wick powers of 
white noise. Heuristically, $\mathfrak{U}_n$ corresponds 
to the formal integral
$\int_0^1 {:}\dot{W}(t)^n{:}\,dt$,
where $\dot{W}$ is the distributional time derivative of 
Brownian motion and ${:}\,\,{:}$ denotes Wick normal ordering. 
Although this expression is only formal, the eigenvectors 
$\mathfrak{U}_n$ are rigorously constructed as elements of 
$(\mathcal{S})^*$ in \cref{sec:wick-power}.
These are familiar objects in Quantum Field Theory (QFT),
where Wick renormalized polynomials 
appear as interaction terms in the action 
functional of self-interacting field theories
\cite{glimm1987quantum}.

In the present context, 
the integral $\int {:}\dot{W}(t)^n{:}\,dt$,
which formally defines \(\mathfrak U_n\), 
appears as the interaction
term of a one-dimensional field theory 
with derivative coupling. The
corresponding Hamiltonian is of gradient 
type, depending on the field
only through its derivative, a feature 
typical of interface models. In
the small coupling regime, its equilibrium 
states are described
formally by the Gibbs measure
\begin{equation}\label{eq:tilt-intro}
    \mu^{(\varepsilon, n)}(dW) \propto 
    e^{-\varepsilon \int_0^1 
    {:}\dot{W}(t)^n{:}\,dt}\,
    \mu_\infty(dW),
\end{equation}
which may be viewed as a continuum 
analogue of the
Ginzburg--Landau \(\nabla\phi\) interface 
model with polynomial
gradient interaction of order \(n\) 
\cite{gawedzki1980rigorous,gawedzki1981renormalization,gawedzki1982renormalization,gawedzki1983block,funaki1997motion,miller2011fluctuations}.
In \cref{sec:qft}, we formally apply the RG theory developed in this 
paper to recover the classical 
relevant/marginal/irrelevant classification familiar from QFT power counting, 
solely from the spectrum of the linearized RG operator 
around the Brownian fixed point.

From a dynamical perspective,
the eigenpairs $(\lambda_n, \mathfrak{U}_n)$ admit a transparent 
interpretation, paralleling the classical
RG analysis of the CLT sketched above. 
There, the linearized RG operator at the Gaussian fixed point 
has eigenvalues $\lambda_n = 2^{1-n/2}$ with eigenvectors given by the 
Hermite functions.
In our path-space setting, the same geometric sequence of 
eigenvalues reappears, with the Wick polynomials 
$\mathfrak{U}_n$ taking the role 
of the Hermite functions. 
The stability of the Brownian fixed point follows the same route. 
The leading eigenvalue $\lambda_0 = 2$ is unstable, with 
associated eigenvector $\mathfrak{U}_0 = \mu_\infty$, representing 
the Brownian fixed point itself. The next eigenvalue 
$\lambda_1 = \sqrt{2}$ is also unstable. 
The eigenvalue $\lambda_2 = 1$ is marginal, while all 
$\lambda_n < 1$ for $n \geq 3$ are stable, with stability 
becoming progressively stronger as $n$ grows.

The eigenmodes also admit a clean probabilistic interpretation,
which is fleshed out in \cref{sec:cumulant-interpretation}; here we summarize the 
main points.
A perturbation $\mu \approx \mu_\infty + \varepsilon \mathfrak{U}_n$ 
of the Wiener measure modifies precisely its $n$-th 
cumulant at order $\varepsilon$, leaving all other cumulants 
unchanged at leading order; a consequence of the localization of 
$\mathfrak{U}_n$ in the $n$-th Wiener chaos.
More specifically, 
the unstable modes $\mathfrak{U}_0, \mathfrak{U}_1$ and the 
marginal mode $\mathfrak{U}_2$ correspond, respectively, to 
perturbations of the total mass, the mean, and the covariance 
of the Wiener measure.
To obtain linearly stable perturbations, and thereby 
convergence towards the fixed point, 
one must remove the spectrally unstable and marginal modes.
Equivalently, stable 
perturbations are those that preserve total mass, mean and 
covariance, deviating from the Wiener measure only in cumulants 
of order $n \geq 3$. 
The spectrum $\{\lambda_n\}_{n \geq 3}$ 
describes how each such cumulant contracts under iteration, with 
higher cumulants decaying faster.
\bigskip

Item (iii) of \cref{thm:intro-spec} provides a rather abstract 
characterization of the basin of attraction of $\mu_\infty$, 
identifying linearly stable perturbations with linear 
combinations of the eigenvectors $\{\mathfrak{U}_n\}_{n \geq 3}$. 
To complement this general picture, one is naturally led to 
ask for a more concrete description of the basin, and the 
question then becomes: which specific measures actually belong 
to it?
We find the answer in two distinct classes of probability measures.
The first corresponds to the family of Gibbs measures in
\eqref{eq:tilt-intro} for $n\geq 3$,
which are treated formally in \cref{sec:qft}.
The second consists of random walk approximations of the
Brownian fixed point.
In what follows, we focus on the random walk
approximations and state the corresponding results.

Let $\zeta$ be a random variable with $\mathbb{E}[\zeta] = 0$ 
and $\mathbb{E}[\zeta^2] = 1$, and let 
$\{\zeta_j\}_{j=1}^{2^N}$ be a sequence of i.i.d. copies of 
$\zeta$. Define the rescaled partial sums
\begin{equation*}
    S_k^{(N)} = 2^{-N/2} \sum_{j=1}^k \zeta_j, 
    \qquad k = 0, 1, \ldots, 2^N,
\end{equation*}
and let $W_N : [0,1] \to \mathbb{R}$ be the continuous path 
obtained by linearly interpolating the values $S_k^{(N)}$ 
at the dyadic grid points.
We denote its law on $\Omega$ by $\widetilde\mu_N$.
By the classical Donsker invariance principle, 
$\widetilde\mu_N$ converges weakly to the 
Wiener measure $\mu_\infty$ as $N \to \infty$.
Donsker's theorem is traditionally interpreted
as a functional CLT, 
lifting the convergence of normalized 
sums of random variables to that of random continuous paths.

This naturally invites an RG interpretation,  
where the sequence $(\widetilde\mu_N)_{N \geq 1}$ 
represents an orbit converging to the fixed point attractor 
$\mu_\infty$. In this case, universality is the 
statement that this convergence is largely insensitive to the 
law of $\zeta$, that is, any centered unit-variance distribution yields 
the same large-scale Brownian limit $\mu_\infty$, 
with $\zeta$ playing the role of a microscopic degree of 
freedom whose particular details are washed away in the scaling 
limit.

Our next result not only substantiates this vision but refines 
it by uncovering a second, finer layer of universality hidden 
in the invariance principle: not only is the convergence towards 
the fixed point universal, but the entire hierarchy of leading 
corrections to the Brownian limit is also universal, governed 
by the eigenmodes $(\lambda_n, \mathfrak{U}_n)$ of $\mathscr{L}$ 
featured in \cref{thm:intro-spec}.

\begin{theorem}[Universality of corrections, simplified form]
    \label{thm:intro-universality}
    Assume that $\zeta$ is a centered random variable with unit variance
    satisfying the sub-Gaussian condition
    \begin{equation*}
        \mathbb{E}\left[e^{t\zeta}\right]
        \leq e^{\sigma^2 t^2/2},
        \qquad t \in \mathbb{R},
      \end{equation*}
    for some constant $\sigma > 0$.
    Then the following holds:
    \begin{enumerate}[label=(\roman*)]
        \item \textbf{Cumulant expansion.} 
        For suitable test functions $h$,
        let 
        $\mathcal{E}(h) = \exp\!\left(\int_0^1 h(t)\,dW_t - \tfrac{1}{2}|h|_{L^2}^2\right)$
        denote the Wick exponential, where $W$ is a 
        standard Brownian motion.
        Then, for every integer $M\ge 3$,
        \begin{equation*}
            \log \mathbb{E}_{\widetilde\mu_N}\!\left[\mathcal{E}(h)\right]
            = \sum_{m=3}^{M} \frac{\kappa_m(\zeta)}{m!}\,
            \lambda_m^N\,\pp{\mathfrak{U}_m}{\mathcal{E}(h)} 
            + o(\lambda_M^N),
            \qquad N \to \infty,
        \end{equation*}
        where $\kappa_m(\zeta)$ is the $m$-th cumulant of $\zeta$ 
        (see \cref{thm:cumulant-expansion}).

        \item \textbf{Leading correction.} 
        Fix $m \geq 3$, and suppose that the moments of $\zeta$
        match those of a standard Gaussian up to order $m-1$,
        that is,
        \begin{equation}\label{eq:moment-match-intro}
          \mathbb{E}[\zeta^r] = \mathbb{E}[Z^r],
          \qquad r = 1, \ldots, m-1,
        \end{equation}
        where $Z \sim \mathcal{N}(0,1)$.
        Then, for every stochastic test functional $\Phi \in (\S)$,

        \begin{equation*}
            \mathbb{E}_{\widetilde\mu_N}[\Phi] 
            - \mathbb{E}_{\mu_\infty}[\Phi]
            = \frac{\kappa_m(\zeta)}{m!}\,\lambda_m^N\,
            \pp{\mathfrak{U}_m}{\Phi} + o(\lambda_m^N),
            \qquad N \to \infty,
        \end{equation*}
        (see \cref{thm:leading-correction}).
    \end{enumerate}
\end{theorem}

In the infinite-dimensional context, the Wick exponential 
$\mathcal{E}(h)$ plays the role of the standard exponential, 
and the pairing $\mathbb{E}_{\widetilde\mu_N}[\mathcal{E}(h)]$ 
is the analogue of a moment generating function. Its logarithm, 
therefore, plays the role of a cumulant generating functional 
for the rescaled random walk law $\widetilde\mu_N$ on the space 
of continuous paths. The asymptotic cumulant expansion in (i) thus reveals a 
remarkable connection between the Donsker approximations and 
the spectral structure of the linearized RG operator: 
the cumulant generating functional of $\widetilde\mu_N$ 
is organized into a hierarchy of contributions arising
from the eigenpairs $(\lambda_n, \mathfrak{U}_n)$, 
with each 
eigenmode $\mathfrak{U}_m$ exciting the corresponding cumulant 
and the associated
eigenvalue $\lambda_m$ dictating its decay rate.
Moreover, the particular cumulant structure of $\zeta$ only
enters as a prefactor 
independent of $N$.
This is the second layer of universality expressed 
in cumulant terms: the non-Gaussian cumulants of 
$\widetilde\mu_N$ vanish at universal rates, prescribed entirely 
by the eigenvalues of $\mathscr{L}$, with the eigenmodes 
$\mathfrak{U}_m$ providing the corresponding universal shapes 
of the corrections.

In part (ii), this universality is investigated directly 
through the difference between expectations of test 
functionals under $\widetilde\mu_N$ and under $\mu_\infty$. 
More specifically, we identify which eigenmode controls the 
leading correction: under the moment-matching condition 
\eqref{eq:moment-match-intro}, the lower eigenmodes 
$\mathfrak{U}_3, \ldots, \mathfrak{U}_{m-1}$ are not excited, 
and the leading correction is determined by $\mathfrak{U}_m$ 
at rate $\lambda_m^N$. The generic case is $m = 3$, meaning that any
centered, unit-variance $\zeta$ with $\kappa_3(\zeta) \neq 0$ 
produces a leading correction governed by the eigenvalue 
$\lambda_3 = 1/\sqrt{2}$ and the eigenmode $\mathfrak{U}_3$. 
Moment matching to higher orders systematically switches off 
lower-order eigenmodes and accelerates the convergence rate 
from $\lambda_3^N$ to $\lambda_m^N$, climbing the eigenmode 
hierarchy.
The role of the microscopic dynamics, 
described by the random variable $\zeta$, is again restricted 
to a scalar prefactor, not affecting the convergence rate.

\introsec{Extensions and Applications}
The RG framework developed here can serve as a foundation for 
several lines of further investigation. 
Below, we comment on a select few.

\paragraph{SDEs driven by Brownian motion.} 
A natural first extension of the present framework is to study 
not only the driving noise itself but the dynamics it generates. 
Given a stochastic differential equation
\begin{equation}\label{eq:intro-sde}
    dX_t = b(X_t)\,dt + \sigma(X_t)\,dW_t,
\end{equation}
the solution $X_t$ is a path-valued functional of the driving 
Brownian motion $W_t$. The classical Wong--Zakai theorem 
\cite{wong1965convergence,wong1965relation}
asserts that, under appropriate assumptions, 
smooth approximations of $W$ produce solutions converging to 
those of the Stratonovich SDE. 
The RG analysis developed here applies directly to the driver 
$W$ and allows one to recast Wong--Zakai-type results into a 
dynamical-systems perspective: smooth approximations of $W$ 
define a family of path-space measures lying in the basin of 
attraction of the Stratonovich fixed point, all flowing to it 
under coarse-graining.
This viewpoint becomes even more illuminating 
when Wong--Zakai fails. In this case, approximations of $W$ may produce,
under coarse-graining, solutions $X$ with a nontrivial It\^o correction.
As such, they flow to a different fixed point 
--- an It\^o fixed point --- with its own basin of attraction. 
The RG framework thus offers a systematic tool to classify
the resulting landscape of fixed points, 
providing a unified universality picture for SDE solutions.

\paragraph{Fractional Gaussian fields.} 
A second direction concerns the extension of the driving noise 
itself. The Brownian fixed point of $\mathcal{R}$ is a particular 
instance of a self-similar Gaussian process --- a fractional 
Gaussian field (FGF) with Hurst index $H = 1/2$ \cite{lodhia2016fractional}.
Generalizing it to an arbitrary FGF in spatial dimension $d$ and Hurst 
index $H \in (0, 1)$ yields a one-parameter family of RG operators 
$\mathcal R^{(s)}$ parametrized by $s=H-d/2$, each with its own Gaussian fixed point.
This extension accommodates many random fields of interest in 
mathematical physics, most notably fractional Brownian motion 
($d = 1$, $H \in (0, 1)$) and its multidimensional extensions, 
Gaussian free fields ($d \geq 1$, $H = (2-d)/2$), and 
log-correlated Gaussian fields ($d \geq 1$, $H = 0$).
Such an extension would deliver a unified universality framework 
covering a broad range of physical models whose scaling limits are 
self-similar Gaussian fields.

\paragraph{Spontaneous stochasticity.} 
A more physically motivated direction concerns the phenomenon 
of spontaneous stochasticity in turbulent flows. 
Spontaneous stochasticity is the phenomenon in which 
deterministic dynamical systems exhibit genuinely random 
behavior in the limit of vanishing regularization, in the 
absence of any external noise. 
In the particular context of SDEs, one considers 
\eqref{eq:intro-sde} with $b$
non-Lipschitz (only Hölder continuous) and with 
noise term of the form $\varepsilon\sigma$ 
where $\varepsilon\to0$.
The system corresponding to $\varepsilon=0$ 
has non-unique solutions, while the limiting procedure 
$\varepsilon\to 0$ selects a stochastic process supported 
on such solutions \cite{barlet2026spontaneous}.

Spontaneous stochasticity has been observed in
turbulent flows modelled by vectorial fractional Gaussian 
fields, where Lagrangian trajectories fail to be unique, with 
the resulting probabilistic spread persisting in the inviscid 
limit \cite{gawedzki2000phase,le2002integration,considera2023spontaneous,considera2026transport}. 
More recently, RG methods have been applied to investigate this phenomenon in 
simplified models, providing a dynamical-systems perspective 
on the emergence of spontaneous stochasticity 
as 
a chaotic attractor on the space of Markov kernels
\cite{mailybaev2022spontaneous,mailybaev2026renormalization,mailybaev2026renormalization2}. 
Verifying this scenario rigorously remains a challenging 
problem, deeply connected to
the well-posedness of singular SDEs and SPDEs. 
We conjecture that a suitable extension of the framework 
developed here will allow for chaotic attractors similar to 
those in \cite{mailybaev2022spontaneous,mailybaev2026renormalization,mailybaev2026renormalization2}, 
paving the way to a rigorous mathematical theory of spontaneous stochasticity 
with ramifications for the solution theory of singular 
differential equations.

\introsec{Organization of the paper}\label{sec:intro-org}

The paper is organized as follows. 
\Cref{sec:preliminaries} collects the necessary background on 
Wiener space, Wiener chaos, and Hida white noise analysis, 
including the construction of the Gelfand triple 
$(\mathcal{S}) \subset L^2(\Omega,\mu_\infty) \subset (\mathcal{S})^*$ and the 
$S$-transform. \Cref{sec:RG operator} introduces the RG 
operator $\mathcal{R}$ at the level of probability measures 
on Wiener space, derives the cascade relation, and identifies 
the Brownian fixed point and its linearization $\mathscr{L}$. 
\Cref{sec:dual-operator} develops the dual operator 
$\mathscr{L}^*$ acting on Hida test functionals, including its 
second-quantization formula, and uses it to extend $\mathscr{L}$ 
itself to Hida distributions. \Cref{sec:gen-spec} establishes 
the generalized spectral theory of $\mathscr{L}$: the spectral 
gap of \cref{thm:intro-spec}~(i) and the diagonal-concentration 
result of \cref{thm:intro-spec}~(ii). \Cref{sec:wick-power} 
constructs the Wick-power eigenvectors $\mathfrak{U}_n$ and 
proves \cref{thm:intro-spec}~(iii). 
\Cref{sec:cumulant-interpretation} develops the probabilistic 
interpretation of the eigenmodes in terms of cumulants and 
Wick perturbations of the Brownian fixed point.
\Cref{sec:donsker} is 
devoted to the Donsker approximations: the random walks $W_N$ 
are realized as Hida distributions, and the cumulant expansion 
of \cref{thm:intro-universality}~(i) and the leading-correction 
result of \cref{thm:intro-universality}~(ii) are established. 
\Cref{sec:qft} carries out the formal application to 
Quantum Field Theory, recovering the 
relevant/marginal/irrelevant classification from the 
spectral data of $\mathscr{L}$.

\section{Preliminaries}\label{sec:preliminaries}
\subsection{Wiener space}\label{sec:Wiener_space}
Let $\Omega := C_{0}([0,1])$
be the Banach space of real-valued continuous functions 
\(\omega : [0,1] \to \mathbb{R}\) satisfying
\(\omega(0) = 0\), 
equipped with the supremum norm
$
|\omega|_\infty := \sup_{t \in [0,1]} |\omega(t)|.
$
Denote by \(\mathscr{B}\) the Borel \(\sigma\)-algebra generated by the norm topology on \(\Omega\).
Let \(\mu_\infty\) be the Wiener measure on \((\Omega, \mathscr{B})\); that is, the law of a standard Brownian motion on \([0,1]\).
The triple
$(\Omega, \mathscr{B}, \mu_\infty)$
thus defines the canonical probability space of Brownian motion.
The coordinate  process $W_t(\omega) := \omega(t)$ 
is then a standard Brownian motion under $\mu_\infty$.

Let $\frak H$ be the Cameron--Martin space associated with the Wiener measure 
$\mu_\infty$, namely
\begin{equation}
    \frak H = \left\{ h \in \Omega : h \text{ is absolutely continuous and } \dot{h} \in L^2([0,1]) \right\}.
\end{equation}
Here, $\dot{h}$ is the time derivative of $h$.
The space $\frak H$ is a Hilbert space equipped with the inner product
\begin{equation}
    \langle h_1, h_2 \rangle_{\frak H} = \int_0^1 \dot{h}_1(t) \dot{h}_2(t) \,dt.
\end{equation}
Observe that with this choice of inner product, $\frak H$ 
is isometrically isomorphic to the space $L^2([0,1])$
of square integrable functions (see \cite[pp.~32]{nualart2006malliavin}), that is 
\begin{equation*}
   \frak H\simeq L^2= L^2([0,1]).    
\end{equation*}
Indeed, the map $h\in L^2([0,1])\mapsto \int_0^\cdot h \in \frak H$ 
is a linear isometry. It is also surjective 
by the fundamental theorem of calculus.
Its inverse is again an isometry (and therefore continuous), 
and hence it is a Hilbert-space isomorphism.
In view of this correspondence, we shall identify $\frak H$ with 
$L^2([0,1])$ in what follows.
Under this identification, 
the triple $(\Omega,L^2([0,1]),\mu_\infty)$ 
with the continuous embedding 
$L^2([0,1])\hookrightarrow \Omega$
is the classical Wiener space associated with
Brownian motion.

\subsection{Wiener chaos expansion}
Consider the $L^2$-indexed stochastic process $\{\mathcal W(h) : h\in L^2([0,1])\}$ 
defined on $(\Omega, \mathscr{B}, \mu_\infty)$ 
by the Wiener--Itô stochastic integral
\begin{equation}\label{eq:isonormal}
    \mathcal{W}(h)\coloneqq  \int_{0}^{1} h(t)\,dW_t.
\end{equation}
For each $h\in L^2([0,1])$, the random variable $\mathcal{W}(h)$ is 
centered Gaussian, and the process $\mathcal W$ has covariance
$\E_{\mu_\infty}\!\!\left(\mathcal{W}(h)\mathcal{W}(u)\right)=\ps{h,u}_{L^2}$,
for all $h,u\in L^2([0,1])$.
Consequently, $\mathcal W$ is an \emph{isonormal Gaussian process} 
on $L^2$ (see \cite[Definition~1.1.1]{nualart2006malliavin}).

Let $\mathcal{F}_\mathcal{W}=\sigma\bigl(\mathcal{W}(h):h\in L^2([0,1])\bigr)$
and $(L^2)= L^2\left(\Omega, \mathcal{F}_\mathcal{W},\mu_\infty\right)$.
Then every $\Phi\in (L^2)$ 
admits a unique Wiener chaos expansion \cite[Theorem~1.1.2]{nualart2006malliavin}
\begin{equation}\label{eq:chaos}
    \Phi=\sum_{n=0}^{\infty} I_n(f_n)\quad \text{in } (L^2),
\end{equation}
where, for $n\ge1$, $f_n\in L^2([0,1])^{\hat\otimes n}$ are symmetric kernels
uniquely determined by $\Phi$, with $f_0=\E_{\mu_\infty}[\Phi]$.
For $n\geq1$,
$I_n$ denotes the $n$-fold stochastic integral with respect to 
the isonormal process $\mathcal{W}$, and $I_0$ is the identity map on $\R$.
Here $\hat{\otimes}$ denotes the symmetric Hilbert tensor product.

\subsection{Hida calculus}\label{sec:hida}

In this subsection we construct the test and
distribution spaces relevant to white noise analysis, 
following the general framework of 
\cite[Chapter 4.2]{kuo2018white}.
Let us comment on two features 
distinguishing our construction from the
classical one.
Traditionally, white noise analysis is formulated on the
canonical white noise probability space, namely 
the Schwartz space of tempered distributions
endowed with the white noise measure
constructed via the Bochner--Minlos theorem
\cite{hida1993white,obata1994white}.
Here, however, we work on the Wiener space
$(\Omega, L^2, \mu_\infty)$.
The framework of \cite[Chapter~4.2]{kuo2018white}
is general enough that the classical construction carries over
upon replacing the Schwartz nuclear space with a
nuclear subspace of $L^2$, obtained as
the projective limit of the domains of a suitable
operator $A$.

The other distinguishing feature lies precisely in
the choice of $A$. In the standard setting, one
takes $A$ to be the harmonic oscillator Hamiltonian,
whose eigenfunctions (the Hermite functions)
determine the nuclear topology of the Schwartz space
(see \cite[Chapter 3]{kuo2018white}).
In the present work, 
we instead define $A$ spectrally,
by declaring the Haar wavelets to be its eigenfunctions, 
and the associated eigenvalues to grow
dyadically in the resolution scale.
This particular choice is motivated by the structure
of the RG operator introduced in
\cref{sec:RG operator}, which acts by splitting and
rescaling paths on the left and right halves of
the unit interval --- precisely the dyadic
decomposition encoded by the Haar system.
This compatibility between $A$ and the Haar system
gives rise to natural functional spaces 
in which
many of the analytical estimates in later sections
become rather straightforward.

We also introduce the $S$-transform, one of the
main tools in Hida calculus, whose definition and
basic properties are readily adapted to our setting.
The main theorems of the classical $S$-transform
theory
--- notably the characterization theorem
(\cref{thm:characterization}) and the associated 
convergence criteria
(\cref{thm:S-convergence}) --- remain at our disposal,
and are put to essential use in \cref{sec:donsker}.

\subsubsection{Deterministic test and distribution spaces} \label{sec:deterministic}

Let $e_{0,0}$ denote the mother Haar wavelet, defined by
\begin{equation}\label{eq:mother-haar}
  e_{0,0}(t)\coloneqq
  \mathbf{1}_{[0,1/2)}(t)-\mathbf{1}_{[1/2,1)}(t),
\end{equation}
and for each $n\ge 0$ and $0\le k\le 2^n-1$ define
\[
  e_{n,k}(t)\coloneqq 2^{n/2}\,e_{0,0}(2^n t-k),
  \qquad t\in[0,1].
\]
Let $\mathbf{1}$ denote the indicator function of the interval $[0,1]$.
Then, the family
$\{\mathbf{1}\}\cup\{e_{n,k}\}_{n\ge 0,\,0\le k\le 2^n-1}$
forms a complete orthonormal basis of $L^2([0,1])$,
and any $f\in L^2([0,1])$ can be represented as 
\[
  f = \ps{f,\, \mathbf{1}}_{L^2}\mathbf{1} + \sum_{n=0}^{\infty}\sum_{k=0}^{2^n-1} \ps{f,\, e_{n,k}}_{L^2}\, e_{n,k}.
\]
Define the operator $A$ on $L^2([0,1])$ by
\[
  Af \coloneqq 
  \alpha_{-1}\,\ps{f,\, \mathbf{1}}_{L^2}\mathbf{1} 
  +
  \sum_{n=0}^{\infty}
  \sum_{k=0}^{2^n-1} \alpha_n\,
  \ps{f,\, e_{n,k}}_{L^2}\, e_{n,k},
  \qquad
  \alpha_n \coloneqq \sqrt{1+2^{2n}},\quad n\ge -1,
\]
with domain
\[
  \operatorname{Dom}(A) = \Bigl\{
  f \in L^2([0,1]) :
  \alpha_{-1}^2\,|\ps{f,\, \mathbf{1}}_{L^2}|^2 +
  \sum_{n=0}^{\infty}\sum_{k=0}^{2^n-1}
  \alpha_n^2\,
  |\ps{f,\, e_{n,k}}_{L^2}|^2
  < \infty \Bigr\}.
\]
By construction, $A$ is a positive
operator with eigenfunctions 
$\{\mathbf{1}\}\cup\{e_{n,k}\}_{n\ge 0,\,0\le k\le 2^n-1}$ and eigenvalues $\{\alpha_n\}_{n\ge -1}$ 
satisfying $1<\alpha_{-1}<\alpha_0<\alpha_1<\cdots$.
For $p\in\mathbb{R}$, the powers $A^p$ 
are defined via spectral multiplication by $\alpha_n^p$ 
in a similar way.
In particular, the Hilbert--Schmidt norm of $A^{-1}$ 
is given by
\[
  \norm{A^{-1}}^2_{HS}
  =\alpha_{-1}^{-2}+\sum_{n=0}^\infty\sum_{k=0}^{2^n-1} \alpha_n^{-2}
  < \infty,
\]
so that $A^{-1}$ is a Hilbert--Schmidt operator on $L^2([0,1])$.


For each $p\ge 0$, define the norm
$$
  |f|_p \coloneqq |A^p f|_{L^2}
   = \left(\alpha_{-1}^{2p}\,|\ps{f,\, \mathbf{1}}_{L^2}|^2 +\sum_{n=0}^\infty\sum_{k=0}^{2^n-1} \alpha_n^{2p}\,|\ps{f,e_{n,k}}_{L^2}|^2\right)^{1/2},
$$
and set 
$$
\S_p\coloneq\operatorname{Dom}(A^p) = \bigl\{f\in L^2([0,1]) : |f|_p <\infty\bigr\}.
$$
For each $p\ge 0$, $(\S_p,|\cdot|_p)$ 
is a separable Hilbert space, 
and  $\S_p\subset \S_q$ whenever $p\geq q\geq0$.
Moreover, the inclusion maps
$$
  \S_{p+1} \hookrightarrow \S_p, \qquad p\ge 0,
$$
are Hilbert--Schmidt operators.
The projective limit
\begin{equation*}
  \mathcal{S} \coloneqq \bigcap_{p=0}^\infty \S_p,
\end{equation*}
is a nuclear countably Hilbert space 
(in particular, a Fr\'echet space)
endowed with the projective limit topology 
induced by the family of norms \(\{ | \cdot |_{p} \}_{p \geq0}\). 
That is, convergence in $\S$ 
means convergence in every $|\cdot|_p$ norm:
\[
f_k \to f \quad \text{in}\;\; \S \quad \Longleftrightarrow \quad | f_k - f |_{p} \to 0 \quad \text{for all } p \geq 0.
\]

The topological (continuous) duals of $\S$ and $\S_p$ are 
denoted by $\S'$ and $\S_p'$, respectively. 
It follows that $\S'$ is given by the inductive limit (see \cite[Chapter 2]{kuo2018white})
$$
\S'= \bigcup_{p=0}^\infty \S_p'.
$$
Moreover, if $p\geq q\geq 0$, any continuous linear functional
over $\S_q$ can be restricted to $\S_p$, leading to the inclusions
$\S'_q\subset\S'_p\subset \S'$ . 
The space $L^2([0,1])$ can be continuously embedded in $\S'_p$
via the continuous injection 
\[
\begin{aligned}
\iota_p : L^2([0,1]) &\hookrightarrow \S_p',\\
f &\mapsto \bigl(\varphi \mapsto \ps{f,\varphi}_{L^2}\bigr),
\qquad \varphi\in \S_p.
\end{aligned}
\]
Similarly, restricting $\iota_p(f)$ to $\S\subset \S_p$ yields a continuous
embedding $L^2([0,1])\hookrightarrow \S'$.
As a result, we obtain the Gelfand triple
$$
\S \subset L^2([0,1]) \subset \S',
$$
with the continuous injections 
$$
\S\subset\S_p  \subset L^2([0,1])\subset\S_p' \subset \S',
$$
for each $p\geq 0$.

Observe that we do not identify $\S_p'$ with $\S_p$, but rather with 
the negative-index space $\S_{-p}$.
More precisely, $\S_{-p}$ is defined as the Hilbert space completion of $L^2([0,1])$ 
with respect to the norm 
$$
|f|_{-p} \coloneqq |A^{-p} f|_{L^2},
$$
and the map
\begin{equation}\label{eq:iso}
    \begin{aligned}
        \Lambda_p:\S_{-p} &\to \S_p',\\
        f &\mapsto \bigl(\varphi \mapsto \ps{A^{-p}f,\,A^{p}\varphi}_{L^2}\bigr),
        \qquad \varphi\in \S_p,
    \end{aligned}
\end{equation}
is an isometric isomorphism. 
As a consequence of this correspondence
we identify $\S_p' \simeq \S_{-p}$.

The preceding development is sufficiently 
abstract to apply to any operator $A$ 
with suitable spectral properties.
We now collect a couple of results 
that rely on the particular structure 
of the Haar wavelets and our choice of $A$.
Specifically, we show that elements 
of $\S_p$, for $p>1/2$, are bounded functions 
defined everywhere on $[0,1]$ and 
not merely $L^2$-equivalence classes.

\begin{proposition}\label{prop:Sp-Linfty}
    Let $p>1/2$. Then for every $f\in\S_p$,
    the Haar series
    \[ \,\ps{f,\mathbf{1}}_{L^2}\mathbf{1}(t) +
    \sum_{n=0}^\infty 
    \sum_{k=0}^{2^n-1} 
    \ps{f,e_{n,k}}_{L^2}\, e_{n,k}(t)
    \]
    converges absolutely for every 
    $t\in[0,1]$.
    Moreover, there exists a finite constant $C_p$ 
    such that
    \begin{equation}\label{eq:Sp-Linfty}
        |f|_{L^\infty} 
        \leq C_p\, |f|_p.
    \end{equation}
    That is, $\S_p$ embeds continuously 
    into $L^\infty([0,1])$.
\end{proposition}
\begin{proof}
    Fix $t\in[0,1]$.
    Multiplying and dividing each term 
    of the Haar expansion by $\alpha_n^p$ 
    and applying the Cauchy--Schwarz 
    inequality gives
    \begin{equation*}
        \begin{split}
            \alpha_{-1}^p|\ps{f,\mathbf{1}}_{L^2}|\,
            \alpha_{-1}^{-p}|\mathbf{1}(t)|
            &+\sum_{n=0}^\infty 
            \sum_{k=0}^{2^n-1}\alpha_n^p|\ps{f,e_{n,k}}_{L^2}| 
            \alpha_n^{-p} |e_{n,k}(t)|
            \\
            &\leq 
            \left(\alpha_{-1}^{2p}|\ps{f,\mathbf{1}}_{L^2}|^2
            +\sum_{n=0}^\infty \sum_{k=0}^{2^n-1}\alpha_n^{2p}|\ps{f,e_{n,k}}_{L^2}|^2\right)^{1/2}
            \\
            &\qquad\times
            \left(\alpha_{-1}^{-2p}\mathbf{1}(t)^2
            +\sum_{n=0}^\infty \sum_{k=0}^{2^n-1} \alpha_n^{-2p} e_{n,k}(t)^2 \right)^{1/2}
            \\
            &=|f|_p 
            \left(\alpha_{-1}^{-2p}\mathbf{1}(t)^2
            +\sum_{n=0}^\infty \sum_{k=0}^{2^n-1} \alpha_n^{-2p} e_{n,k}(t)^2 \right)^{1/2}.
        \end{split}
    \end{equation*}
    At each resolution level $n$, the Haar 
    wavelets $\{e_{n,k}\}_{k=0}^{2^n-1}$ have 
    mutually disjoint supports, 
    so at most one term in the inner sum 
    is nonzero.
    Since $|e_{n,k}(t)|\leq 2^{n/2}$ 
    and 
    $\alpha_n^{-2p}\leq 2^{-2p n}$, 
    we obtain 
    \[
        \alpha_{-1}^{-2p}
        +\sum_{n=0}^\infty
        \sum_{k=0}^{2^n-1} 
        \alpha_n^{-2p}\, e_{n,k}(t)^2 
        \leq \alpha_{-1}^{-2p}
        +\sum_{n=0}^\infty 
        2^{-2pn}\cdot 2^n 
        = \alpha_{-1}^{-2p}
        +\frac{1}{1-2^{1-2p}}
        < \infty,
    \]
    for $p>1/2$, independent of $t$.
    This establishes absolute
    convergence, as well as 
    the bound \eqref{eq:Sp-Linfty}.
\end{proof}

The pointwise convergence of the Haar 
series in \cref{prop:Sp-Linfty} allows 
us to pass from $L^2$-equivalence classes 
to canonical, everywhere-defined representatives.

\begin{corollary}\label{cor:pointwise-rep}
    For $p>1/2$, every $f\in\S_p$ admits 
    a canonical bounded representative, 
    defined at every point of $[0,1]$ 
    by the pointwise limit of its Haar series.
\end{corollary}

\begin{remark}
  Throughout the rest of the paper, whenever $p>1/2$, we shall 
  identify each $f\in\S_p$ with its canonical representative 
  and regard $f$ as an 
  everywhere-defined bounded function.
  In particular, this applies to every 
  $f\in\S$.
\end{remark}

\subsubsection{Stochastic test and distribution spaces}\label{sec:stochastic}

In what follows, we fix a single-index enumeration
of the family 
$\{\mathbf{1}\}\cup\{e_{n,k}\}_{n \geq 0,\, 0 \leq k \leq 2^n-1}$
and denote the resulting
orthonormal basis by $\{e_k\}_{k \geq 1}$,
ordered so that the sequence of eigenvalues
$\{\alpha_k\}_{k \geq 1}$ remains nondecreasing
and strictly greater than 1.
For instance, this is achieved by 
setting $e_1=\mathbf{1}$ and then
enumerating the
wavelets level by level, according to their
resolution scale.
We employ this convention for the sake of
simplicity, reducing the notational weight of
what is to follow.
In addition, this matches the convention of the
standard references
\cite{kuo2018white,obata1994white},
facilitating comparison with the existing
literature.

With this convention in place, we now turn our
attention to the stochastic test and distribution spaces. 
They can be constructed by lifting the previous construction 
from $L^2([0,1])$ to $(L^2)=L^2(\Omega,\mathcal{F}_\mathcal{W},\mu_\infty)$
using the \emph{second quantization operator} $\Gamma(A)$ 
and the Wiener chaos expansion \eqref{eq:chaos}.
For $\Phi\in (L^2)$ with chaos expansion 
$$
    \Phi=\sum_{n=0}^{\infty} I_n(f_n),
$$
such that
$$
\sum_{n=0}^\infty n!\bigl|A^{\otimes n}f_n\bigr|_{L^2([0,1]^n)}^2 < \infty,
$$
we define $\Gamma(A)\Phi \in (L^2)$ by 
$$
\Gamma(A)\Phi=\sum_{n=0}^{\infty} I_n(A^{\otimes n}f_n).
$$

The operator $\Gamma(A)$ is a densely defined operator on $(L^2)$ 
with many properties similar to $A$. For instance,
we can lift the basis family $\{e_k\}_{k\geq 1}\subset L^2([0,1])$ 
to $(L^2)$ in the following manner.
Let $\mathbf n=(n_1,n_2,\dots)$ be a multiindex 
such that all but finitely many values vanish,
and write
$$
|\mathbf n|\coloneqq \sum_{k\geq1} n_k, \quad 
\mathbf n!\coloneqq \prod_{k\geq 1} n_k! \;.
$$
Define the symmetric kernel
\[
f_{\mathbf n}
\;\coloneqq\;
\mathop{\bigotimes}\limits_{k\ge 1}^{\smash{\raisebox{-1.7ex}[0pt][0pt]{$\widehat{\phantom{X}}$}}}
\, e_k^{\otimes n_k}
\;\in\; L^2([0,1])^{\hat\otimes |\mathbf n|},
\]
and set
$$
\Psi_{\mathbf n}
\coloneqq
\frac{1}{\sqrt{\mathbf n!}}\; I_{|\mathbf n|}(f_{\mathbf n}),
$$
with the convention $\Psi_{\mathbf 0}=I_0(1)=1$.
Then the family $\bigl\{\Psi_{\mathbf n}:|\mathbf n|=n ,\, n=0,1,2\dots \bigr\}$ 
is an orthonormal basis of
$(L^2)$, and \(\Gamma(A)\) is diagonal in this basis: 
\[
\Gamma(A)\Psi_{\mathbf n}
=
\frac{1}{\sqrt{\mathbf n!}}\; I_{|\mathbf n|}\!\Big(A^{\otimes |\mathbf n|} f_{\mathbf n}\Big)
=
\Big(\prod_{k\ge 1} \alpha_k^{\,n_k}\Big)\Psi_{\mathbf n}.
\]
In other words, $\Psi_{\mathbf n}$ 
is an eigenfunction of $\Gamma(A)$ with eigenvalue
$\alpha_{\mathbf n}\coloneqq \displaystyle \prod_{k\ge 1} \alpha_k^{n_k}$.
One can easily check that the family $\{\alpha_{\mathbf n}\}_{\mathbf{n}}$ is 
bounded below by 1 and unbounded above.
Consequently, $\Gamma(A)^{-1}$ is a Hilbert-Schmidt operator with
\begin{equation*}
        \|\Gamma(A)^{-1}\|_{HS}^2
        =\sum_{\mathbf n}\prod_{k\ge 1}\alpha_k^{-2n_k}
        = \prod_{k= 1}^\infty\sum_{m=0}^\infty\alpha_k^{-2m} 
        =\prod_{k= 1}^\infty \frac{1}{1-\alpha_k^{-2}} < \infty.
\end{equation*}

These structural similarities allow us to construct the associated 
stochastic Gelfand triple following the same approach 
as in \cref{sec:deterministic}.
For each $p\geq0$, we define the Hida test norms
\begin{equation*}
    \norm{\Phi}_p\coloneqq \norm{\Gamma(A)^p\Phi}_{(L^2)}=
    \left(\sum_{n=0}^\infty n!\bigl|(A^p)^{\otimes n}f_n\bigr|_{L^2([0,1]^n)}^2\right)^{1/2},
\end{equation*}
and let 
\begin{equation*}
    (\S_p)\coloneqq\operatorname{Dom}(\Gamma(A)^p)= \bigl\{\Phi \in (L^2): \norm{\Phi}_p <\infty\bigr\}.
\end{equation*}
Similar to the deterministic case, 
we have that 
for each $p\ge 0$, $\bigl((\S_p),\norm{\cdot}_p\bigr)$ 
is a separable Hilbert space, 
and  $(\S_p)\subset (\S_q)$ whenever $p\geq q\geq0$.
Similarly, the inclusion maps
$$
  (\S_{p+1}) \hookrightarrow (\S_p), \qquad p\ge 0,
$$
are Hilbert--Schmidt operators.
The \emph{Hida test space} $(\mathcal S)$ is then obtained as the projective limit
$$
(\mathcal S) \coloneqq \bigcap_{p\ge0} (\mathcal S_p).
$$
It is a nuclear countably Hilbert space (hence a Fréchet space),
equipped with the projective limit topology 
induced by the family of norms 
\(\{ \norm{\cdot}_{p} \}_{p \geq0}\):
\[
\Phi_k \to \Phi \quad \text{in}\;\; (\S) \quad \Longleftrightarrow \quad  \norm{\Phi_k - \Phi}_{p} \to 0 \quad \text{for all } p \geq 0.
\]

We denote by $(\S)^*$ and $(\S_p)^*$
the topological duals of $(\S)$ and $(\S_p)$, respectively,
with $(\S)^*$ represented by the inductive limit 
\begin{equation*}\label{eq:hida distributions}
    (\mathcal S)^* = \bigcup_{p\ge0} (\mathcal S_p)^*.
\end{equation*}
The space $(\S)^*$ is called \emph{Hida distribution space}. 
Its elements are called \emph{Hida distributions}
or \emph{stochastic generalized functionals}.
By restricting linear functionals we have the 
inclusions $(\S_q)^*\subset(\S_p)^* \subset(\S)^*$
if $p\geq q \geq 0$.
Moreover, by using the $(L^2)$-inner product 
we have the continuous embeddings $(L^2)\hookrightarrow (\S_p)^*$ 
and $(L^2)\hookrightarrow (\S)^*$,
in a similar fashion as \cref{sec:deterministic}.
This yields the stochastic Gelfand triple
$$
(\S) \subset (L^2) \subset (\S)^*,
$$
and the continuous injections 
$$
(\S)\subset(\S_p)  \subset (L^2)\subset(\S_p)^* \subset (\S)^*,
$$
for each $p\geq 0$.

Denote by $(\S_{-p})$ the Hilbert space completion of $(L^2)$ 
with respect to the norm 
\begin{equation*}
    \norm{\Phi}_{-p}\coloneqq \norm{\Gamma(A)^{-p}\Phi}_{(L^2)}.
\end{equation*}
Then, one can construct an isometric isomorphism from 
$(\S_{-p})$ to $(\S_{p})^*$ by replacing $A$ with $\Gamma(A)$ and 
the $L^2$-inner product with the $(L^2)$-inner product 
in expression \eqref{eq:iso}. Thus we make the identification
$(\S_{p})^*\simeq (\S_{-p})$.

We end this subsection with two results characterizing
Hida test functionals and Hida distributions 
in terms of their chaos expansions. 
Roughly speaking, the first proposition shows that 
the regularity of a test functional
is encoded in the regularity of its chaos kernels, 
whereas the second provides a
generalized chaos expansion for Hida distributions.
Both statements are standard in Hida calculus; 
see, e.g.,
\cite[Theorem~3.1.5]{obata1994white} and 
\cite[Theorem~3.1.6]{obata1994white}.

\medskip
\noindent\textbf{Notation.}
Throughout the paper, we use the notation $(\cdot,\cdot)$ 
to denote the dual pairing
between $\S'$ and $\S$, as well as the induced 
pairings between the corresponding
symmetric tensor powers $\S'^{\hat\otimes n}$ 
and $\S^{\hat\otimes n}$, $n\ge1$.
Similarly, the notation $\pp{\cdot}{\cdot}$ 
denotes the dual pairing between
$(\S)^*$ and $(\S)$.

\begin{proposition}\label{prop:hida-test-char}
    Let $\Phi \in (L^2)$ have chaos expansion
    \[
        \Phi=\sum_{n=0}^{\infty} I_n(f_n).
    \]
    Then $\Phi\in (\S)$ if and only if $f_n\in \S^{\hat\otimes n}$ for all $n\geq0$ and
    \[
        \sum_{n=0}^\infty n!\bigl|(A^p)^{\otimes n}f_n\bigr|_{L^2([0,1]^n)}^2 < \infty,
    \]
    for all  $p\geq0$.
\end{proposition}

\begin{proposition}[Generalized chaos expansion]\label{thm:generalized chaos}
Let $U\in (\S)^*$.
Then there exists an index $p\geq0$ and a unique sequence
$\{u_n\}_{n\geq0}$ with $u_n\in \S_{-p}^{\hat\otimes n}$ such that for every
$\Phi\in (\S)$ with chaos expansion $\Phi=\sum_{n=0}^{\infty} I_n(f_n)$ one has
\begin{equation}\label{eq:generalized chaos}
    \pp{U}{\Phi}=\sum_{n=0}^\infty n!(u_n,f_n),
\end{equation}
where the series converges absolutely.
In this case, we write formally
\begin{equation*}
    U= \sum_{n=0}^{\infty} I_n(u_n),
\end{equation*}
in analogy with the chaos expansion of square-integrable functionals.
In addition, it follows that
\begin{equation}\label{eq:hida1}
    \norm{U}_{-p}^2= \sum_{n=0}^\infty n!\bigl|(A^{-p})^{\otimes n}u_n\bigr|_{L^2([0,1]^n)}^2 <\infty.
\end{equation}

Conversely, assume we are given a sequence of kernels $\{u_n\}_{n\geq0}$ with 
$u_n\in \S_{-p}^{\hat\otimes n}$ such that the series in 
\eqref{eq:hida1} converges for some $p\geq0$. 
Then, expression \eqref{eq:generalized chaos} defines a Hida distribution 
$U\in(\S)^*$ and \eqref{eq:hida1} holds.
\end{proposition}

\subsubsection{The $S$-transform}\label{sec:S-transform}

We introduce the $S$-transform and establish
the properties that will be used throughout the paper.  
We follow the standard treatment of white noise analysis 
given in \cite[Section~5.2]{kuo2018white},
but adapted to our Wiener space setting.
We begin by defining the Wick exponentials, 
the infinite-dimensional analogue of 
the exponential function in classical analysis. 
Just as a probability measure on the real line
is determined by its moment generating function,
a Hida distribution is determined by its pairing against 
Wick exponentials, which is precisely the $S$-transform.

\begin{definition}[Wick exponential]\label{def:wick-exp}
For $h \in L^2([0,1])$, define the \emph{Wick exponential}
(or \emph{exponential vector}) by
\begin{equation*}\label{eq:wick-exp}
  \mathcal{E}(h)
  := \exp\!\left(\mathcal W(h) - \tfrac{1}{2}|h|_{L^2}^2\right),
\end{equation*}
where $\mathcal W$ is the isonormal Gaussian process 
introduced in~\eqref{eq:isonormal}.
\end{definition}

In what follows, $H_n$ denotes the $n$-th Hermite polynomial
defined by
\begin{equation}\label{eq:hermite-def}
  H_n(x) := (-1)^n\,e^{x^2/2}\,
  \frac{d^n}{dx^n}\!\left(e^{-x^2/2}\right),
  \qquad n \geq 0.
\end{equation}
In particular, $H_0(x) = 1$, $H_1(x) = x$,
$H_2(x) = x^2 - 1$, and $H_3(x) = x^3 - 3x$.
These polynomials are characterized by the generating
function identity
(see \cite[Eq.~(1.1)]{nualart2006malliavin}):
\begin{equation}\label{eq:hermite-gf}
  \exp\!\left(tx - t^2/2\right)
  = \sum_{n=0}^{\infty} \frac{t^n}{n!}\,H_n(x),
  \qquad t, x \in \mathbb{R},
\end{equation}
where the series converges absolutely.
\begin{remark}  
  Note that the Hermite polynomials used
  in~\cite{nualart2006malliavin} are
  $h_n(x) = H_n(x)/n!$, so
  that~\cite[Eq.~(1.1)]{nualart2006malliavin} reads
  $\exp(tx - t^2/2) = \sum_{n=0}^{\infty} t^n\,h_n(x)$.
\end{remark}

\begin{proposition}\label{prop:wick-chaos}
For every $h \in L^2([0,1])$, the Wick exponential admits the
chaos expansion
\begin{equation}\label{eq:wick-chaos}
  \mathcal{E}(h)
  = \sum_{n=0}^{\infty} \frac{1}{n!}\,I_n(h^{\otimes n})
  \qquad \text{in } (L^2).
\end{equation}
Moreover, $\mathcal{E}(h) \in (\S)$ if and only if $h \in \S$,
and in that case
\begin{equation*}
  \|\mathcal{E}(h)\|_p^2
  = \exp\!\left(|h|_p^2\right),
  \qquad p \geq 0.
\end{equation*}
\end{proposition}

\begin{proof}
If $h = 0$ the identity~\eqref{eq:wick-chaos} reduces to
$1 = I_0(1)$, so we may assume $h \neq 0$.
Set $e := h/|h|_{L^2}$, so that $\mathcal W(e)$ is a standard normal
random variable under $\mu_\infty$.  
Using \eqref{eq:hermite-gf} with 
$x = \mathcal{W}(e)$ and $t = |h|_{L^2}$,
together with $|h|_{L^2}\,\mathcal{W}(e) = \mathcal{W}(h)$, yields 
the almost sure convergence
$$
\mathcal{E}(h)=
\sum_{n=0}^{\infty} \frac{|h|_{L^2}^n}{n!}\,H_n(\mathcal{W}(e))
= \sum_{n=0}^{\infty} \frac{|h|_{L^2}^n}{n!}\,I_n(e^{\otimes n})
= \sum_{n=0}^{\infty} \frac{1}{n!}\,I_n(h^{\otimes n})
\qquad \text{a.s.},
$$
where we used that 
$H_n(\mathcal{W}(e)) = I_n(e^{\otimes n})$
(see \cite[Proposition~1.1.4]{nualart2006malliavin}),
and 
$|h|_{L^2}^n\,e^{\otimes n} = h^{\otimes n}$.
Moreover, using the orthogonality of the chaos spaces 
and It\^o isometry,
one can show that the series on the right-hand 
side of \eqref{eq:wick-chaos} converges in $(L^2)$ to some limit.
Since both almost sure convergence and $(L^2)$-convergence
imply convergence in probability, and limits in probability
are unique, the two limits coincide.
This proves~\eqref{eq:wick-chaos}.

For the norm computation, note that the chaos kernels of
$\mathcal{E}(h)$ are $f_n = h^{\otimes n}/n!$ by \eqref{eq:wick-chaos}.
From the definition of the Hida stochastic norm we have
\[
    \norm{\mathcal{E}(h)}_p^2
  = \sum_{n=0}^{\infty} n!
    \left|(A^p)^{\otimes n} \left(\frac{h^{\otimes n}}{n!}\right)\right|_{L^2([0,1]^n)}^2
  = \sum_{n=0}^{\infty} \frac{|h|_p^{2n}}{n!}
  = \exp\!\left(|h|_p^2\right),
\]
where we used the multiplicativity of the tensor norm:
$|(A^p)^{\otimes n} h^{\otimes n}|_{L^2([0,1]^n)}
= |A^p h|_{L^2}^n = |h|_p^n$.
This is finite for all $p \geq 0$ if and only if
$|h|_p < \infty$ for all $p$, i.e., $h \in \S$.
\end{proof}

\begin{remark}
    The chaos expansion and the norm
    formula in \cref{prop:wick-chaos} are standard in white noise
    analysis; see \cite[Lemma~5.5 and Theorem~5.7]{kuo2018white}
    for analogous statements in the canonical white noise
    setting.
\end{remark}

\begin{definition}[$S$-transform]\label{def:S-transform}
For $U \in (\S)^*$, the \emph{$S$-transform} of $U$ is the
function $SU : \S \to \mathbb{R}$ defined by
\begin{equation*}
  SU(h) := \pp{U}{\mathcal{E}(h)},
  \qquad h \in \S.
\end{equation*}
\end{definition}

\begin{remark}
The $S$-transform is well-defined since 
$U \in (\S)^*$ and 
$\mathcal{E}(h) \in (\S)$ for every $h \in \S$
(\cref{prop:wick-chaos}). 
As a concrete example, consider the case where 
$U = \Phi$ for some square-integrable random variable
$\Phi \in (L^2)$, viewed as a Hida distribution via the 
canonical embedding $(L^2) \subset (\S)^*$ given by the 
$(L^2)$-inner product. Then the dual pairing reduces 
to the expectation under $\mu_\infty$, and the $S$-transform takes 
the familiar form
\[
  S\Phi(h) = \E_{\mu_\infty}[\Phi\,\mathcal{E}(h)],
  \qquad h \in \S.
\]
In this sense, the $S$-transform generalizes 
the notion of testing a random variable against 
exponential vectors to the full space of Hida 
distributions, including those that are not 
representable as elements of $(L^2)$.
\end{remark}

The next proposition provides an explicit series representation
for the $S$-transform in terms of the generalized chaos kernels.
\begin{proposition}\label{prop:S-series}
    Let $U \in (\S)^*$ with generalized chaos expansion
    $U = \sum_{n=0}^{\infty} I_n(u_n)$ as in
    \cref{thm:generalized chaos}.
    Then
    \begin{equation*}
      SU(h) = \sum_{n=0}^{\infty} (u_n,\, h^{\otimes n}),
      \qquad h \in \S,
    \end{equation*}
    where the series converges absolutely.
\end{proposition}
\begin{proof}
    By \cref{thm:generalized chaos} applied to 
    $\Phi = \mathcal{E}(h)$, whose chaos kernels are 
    $f_n = h^{\otimes n}/n!$ (\cref{prop:wick-chaos}),
    \[
      SU(h) 
      = \pp{U}{\mathcal{E}(h)}
      = \sum_{n=0}^{\infty} n!
        \left(u_n,\,\frac{h^{\otimes n}}{n!}\right)
      = \sum_{n=0}^{\infty} (u_n,\, h^{\otimes n}).
    \]
    Absolute convergence also follows by \cref{thm:generalized chaos}.
\end{proof}

\begin{proposition}[Injectivity of $S$-transform]\label{thm:S-injective}
Let $U_1,U_2\in(\S)^*$. 
If $SU_1 = SU_2$, then  $U_1 = U_2$.
\end{proposition}

\begin{proof}
It is enough to show that if $SU\equiv0$ then $U\equiv0$.
Let $U = \sum_{n=0}^{\infty} I_n(u_n) \in (\S)^*$ with
$SU \equiv 0$.  By \cref{prop:S-series},
\begin{equation*}
  SU(th)=\sum_{n=0}^{\infty}t^n (u_n,\, h^{\otimes n}) = 0,
  \qquad \text{for all } 
  t \in \mathbb{R} \text{ and all } h \in \S.
\end{equation*}
Since the series above converges absolutely for every 
$t \in \mathbb{R}$, the map
$t \mapsto SU(th)$ is real-analytic on $\mathbb{R}$.
The condition $SU\equiv 0$ implies that all its Taylor coefficients
are zero:
\[  \frac{d^n}{dt^n}\bigg|_{t=0}\!SU(th) 
    =n!(u_n,\, h^{\otimes n}) =0,
  \qquad \text{for all } n \geq 0
  \text{ and all } h \in \S.
\]
For $n = 0$, the identity reads $u_0 = 0$ directly.
For $n \geq 1$, fix $h_1, \ldots, h_n \in \S$.
By the polarization identity for symmetric $n$-linear forms
(see \cite[proof of Proposition~5.10]{kuo2018white}), we have
\[
  (u_n,\, h_1 \hat\otimes \cdots \hat\otimes h_n)
  = \frac{1}{n!}\sum_{k=1}^{n}(-1)^{n-k}
    \sum_{j_1 < \cdots < j_k}
      \bigl(u_n,\,
        (h_{j_1} + \cdots + h_{j_k})^{\otimes n}
      \bigr).
\]
Each pairing on the right-hand side is of the form
$(u_n, h^{\otimes n})$ with
$h = h_{j_1} + \cdots + h_{j_k} \in \S$, which vanishes by
the previous step.  Hence
$(u_n, h_1 \hat\otimes \cdots \hat\otimes h_n) = 0$
for all $h_1, \ldots, h_n \in \S$.
Since elementary symmetric tensors
$h_1 \hat\otimes \cdots \hat\otimes h_n$ with
$h_i \in \S$ span a dense subspace of
$\S_p^{\hat\otimes n}$, 
and $u_n$ is a continuous linear
functional on $\S_p^{\hat\otimes n}$, it follows that
$u_n = 0$.
As this holds for every $n \geq 0$, we conclude that $U \equiv 0$.
\end{proof}

The density of exponential vectors in $(\S)$ follows from the
injectivity of the $S$-transform established above, together with
the geometric Hahn--Banach separation theorem for locally convex
spaces; see \cite[Theorem~3.5]{rudin1991functional}.

\begin{corollary}\label{cor:exp-dense}
The subspace
$\mathcal{D} \coloneq \operatorname{span}\{\mathcal{E}(h) : h \in \S\}$
is dense in~$(\S)$.
\end{corollary}
\begin{proof}
Suppose $U \in (\S)^*$ satisfies $\pp{U}{\Phi} = 0$ for every
$\Phi \in \mathcal{D}$.  In particular,
$SU(h) = \pp{U}{\mathcal{E}(h)} = 0$ for all $h \in \S$.
By \cref{thm:S-injective}, $U = 0$.  
We claim that this implies that $\mathcal D$ is dense in $(\S)$.
Indeed, if this was not the case, one can use
\cite[Theorem~3.5]{rudin1991functional} 
to find $U\in (\S)^*$ such that $U|_{\mathcal{D}}=0$ but $U\neq 0$,
leading to a contradiction.
\end{proof}

We close this subsection with two general results,
originally due to Potthoff and Streit (PS) \cite{potthoff1991characterization}.
The first is the characterization theorem, 
which provides conditions under which a
given function on $\S$ arises as the $S$-transform of a Hida
distribution. For a proof, we refer the reader to 
\cite[Theorem~8.2]{kuo2018white}.
The second result provides a criterion for
convergence of Hida distributions in terms of their
$S$-transforms, in the same spirit as the Lévy continuity
theorem for characteristic functions of random variables, 
and can be
found in \cite[Theorem~8.6]{kuo2018white}.
In order to state them, we extend the $S$-transform
to the complexification
$\S_{\mathbb{C}} := \S + i\S$ by replacing
$h \in \S$ with $\xi \in \S_{\mathbb{C}}$ in 
\cref{def:S-transform}. 
There is no real obstruction to working with
complex-valued test functions, and all previous
results concerning the $S$-transform carry over directly 
to the complexified setting.

\begin{theorem}[Characterization theorem]%
\label{thm:characterization}
A function $F : \S_{\mathbb{C}} \to \mathbb{C}$ is the
$S$-transform of a (unique) Hida distribution $U \in (\S)^*$
if and only if the following two conditions hold:
\begin{enumerate}
  \item[\textup{\textbf{(PS1)}}]
  \textit{Ray analyticity.}
  For every $\xi, \eta \in \S_{\mathbb{C}}$, the function
  $z \mapsto F(\xi + z\,\eta)$ is entire on $\mathbb{C}$.
  \item[\textup{\textbf{(PS2)}}]
  \textit{Growth bound.}
  There exist constants $K_1, K_2 \geq 0$ and
  $p \in \mathbb{N}_0$ such that
    \[
        |F(\xi)|
        \leq K_1\exp\!\left(K_2\,|\xi|_p^2\right),
        \qquad 
        \xi \in \S_{\mathbb{C}}.
    \]
\end{enumerate}
Uniqueness holds by \cref{thm:S-injective}.
\end{theorem}

\begin{theorem}[Convergence of Hida distributions]\label{thm:S-convergence}
    Let $(U_N)_{N \geq 1} \subset (\S)^*$ and let
    $F_N := SU_N$.  Suppose:
    \begin{enumerate}
      \item[\textup{(i)}]
      $\lim_{N \to \infty} F_N(\xi)$ exists for every
      $\xi \in \S_{\mathbb{C}}$;
      \item[\textup{(ii)}]
      there exist constants $K_1, K_2\geq 0$ and
      $p \in \mathbb{N}_0$, independent of $N$, such that
      \[
        |F_N(\xi)|
        \leq K_1\,\exp\!\left(K_2\,|\xi|_p^2\right),
        \qquad N \geq 1,\;
        \xi \in \S_{\mathbb{C}}.
      \]
    \end{enumerate}
    Then there exists a unique $U \in (\S)^*$ such that
    $SU(\xi) = \lim_{N \to \infty} F_N(\xi)$ for every
    $\xi \in \S_{\mathbb{C}}$, and $U_N \to U$ in the
    weak-$*$ topology of $(\S)^*$, i.e.,
    \[
      \pp{U_N}{\Phi} \to \pp{U}{\Phi}
      \qquad \text{for every } \Phi \in (\S).
    \]
\end{theorem}

\begin{remark}
For a nuclear countably Hilbert space such as $(\S)$,
weak-$*$ convergence of sequences in $(\S)^*$
is equivalent to norm convergence:
$U_N \to U$ in the weak-$*$ topology of $(\S)^*$
if and only if there exists $q \geq 0$ such that
$\|U_N - U\|_{-q} \to 0$ 
(see \cite[pp.~10--11]{kuo2018white}).
In particular, the convergence in
\cref{thm:S-convergence} also holds 
in the sense of norm convergence.
\end{remark}

\section{Renormalization group (RG) operator}\label{sec:RG operator}

In this section we introduce the RG operator together with its
linearization. We use the notation introduced in
\cref{sec:Wiener_space}; in particular, $\Omega$ denotes the space of
continuous functions on $[0,1]$ vanishing at $t=0$, and $\mu_\infty$
is the Wiener measure on $\Omega$.

The starting point is the Lévy--Ciesielski construction of Brownian
motion on the unit interval, which is based on the Faber--Schauder
wavelet system.
The Faber--Schauder system is constructed from scaled and shifted
copies of a single ``tent'' function, playing the role of the mother
wavelet. Specifically, we define the mother wavelet by
\begin{equation}\label{eq:mother-schauder}
    \psi(t) =
    \begin{cases}
      t & \text{if } 0 \leq t < 1/2, \\
      1 - t & \text{if } 1/2 \leq t \leq 1, \\
      0 & \text{otherwise},
    \end{cases}
\end{equation}
and for integers $n \geq 0$ and $0 \leq k \leq 2^n - 1$, we define the
Faber--Schauder wavelet functions
\begin{equation}
    \psi_{n,k}(t) = 2^{-n/2} \, \psi(2^n t - k), \quad 0\leq t\leq 1.
\end{equation}
Each $\psi_{n,k}$ is supported on the dyadic interval
$[k 2^{-n}, (k+1) 2^{-n}]$, attaining its peak at $(k + 1/2) 2^{-n}$.

This construction gives rise to a natural multiscale,
piecewise-linear approximation scheme for Brownian sample paths.
Indeed, let
$\left(\xi, \xi_{n,k}\right)_{n \geq 0, 0 \leq k \leq 2^n - 1}$
be a sequence of independent standard normal random variables
on $(\Omega, \mathscr{B}, \mu_\infty)$.
For each $N\geq0$, define the partial sum
\begin{equation}\label{eq:B_N}
    \beta_N(t) = \sum_{n=0}^N \sum_{k=0}^{2^n - 1} \xi_{n,k} \,
    \psi_{n,k}(t),\quad 0\leq t\leq 1.
\end{equation}
The process $\left(\beta_N(t)\right)_{0 \leq t \leq 1}$ provides an approximation of 
the Brownian bridge at scale $N$. By adding a large-scale Gaussian mode, 
one recovers a standard Brownian motion in the limit of infinite resolution, as stated in the 
following theorem (see, e.g., \cite{evans2012introduction}).

\begin{theorem}[Wavelet approximation of Brownian motion]
\label{thm:bm_convergence}
    For $0\leq t\leq 1$, define
    $$
    W_N(t) = \xi t + \beta_N(t).
    $$
    Then, almost surely,
    $W_N$ converges uniformly to a standard Brownian motion on the
    interval $[0,1]$.
\end{theorem}

Beyond yielding approximations of Brownian motion, the
Faber--Schauder wavelet construction also induces a natural dyadic
cascade relation across scales, which will serve as the starting
point for the RG analysis carried out in this work.

\begin{proposition}[Cascade relation] \label{thm:cascade relation}
    For each $N\geq 0$ we have the following equality in law
    \begin{equation}\label{eq:cascade relation}
        W_{N+1}(t)\law
        \frac{1}{\sqrt 2}
        \left(
        W'_N(2t\wedge 1)
        +W''_N\big((2t-1)\vee 0\big)
        \right),
        \qquad 0\leq t\leq 1,
    \end{equation}
    where $W'_N$ and $W''_N$ are two independent copies of $W_N$.
\end{proposition}
Explicitly, the right-hand side of \eqref{eq:cascade relation} equals
$\tfrac{1}{\sqrt{2}}W'_N(2t)$ for $t\in[0,\tfrac12]$ and
$\tfrac{1}{\sqrt{2}}\left(W'_N(1)+W''_N(2t-1)\right)$ for
$t\in[\tfrac12,1]$. 
In other words, the first copy is 
rescaled to the left half of the unit interval
and then held fixed at its terminal value,
while the second copy is rescaled and shifted to 
the right half of the unit interval;
the two halves join continuously at $t=\tfrac12$.

\begin{proof}
\textit{Wavelet modes.}\quad
We express $W_{N+1}$ in terms of $W_N$ by decomposing 
the approximation $W_{N+1}$ into its constituent wavelet components as follows
\begin{align*}
    W_{N+1}(t) =\,\, &\xi t +  \xi_{0,0}\psi(t) + \sum_{n=1}^{N+1}\sum_{k=0}^{2^n-1}\xi_{n,k}
    \psi_{n,k}(t)
    \nonumber
    \\
    =\,\,& \xi t + \xi_{0,0}\psi(t) + \sum_{n=0}^{N}\sum_{k=0}^{2^{n+1}-1}\xi_{n+1,k}\psi_{n+1,k}(t)
    \nonumber
    \\
    =\,\,& \xi t + \xi_{0,0}\psi(t) + \sum_{n=0}^{N}
    \left[
        \sum_{k=0}^{2^{n}-1}\xi_{n+1,k}\psi_{n+1,k}(t) 
        +
        \sum_{k=2^n}^{2^{n+1}-1}\xi_{n+1,k}\psi_{n+1,k}(t)
    \right]
    \nonumber
    \\
    =\,\,& \xi t + \xi_{0,0}\psi(t) + \sum_{n=0}^{N}
    \left[
        \sum_{k=0}^{2^{n}-1}\xi_{n+1,k}\psi_{n+1,k}(t) 
        +
        \sum_{k=0}^{2^{n}-1}\xi_{n+1,k+2^n}\psi_{n+1,k+2^n}(t) 
    \right].
\end{align*}
By using the scaling relations
\begin{equation*}
    \psi_{n+1,k}(t)=\frac{1}{\sqrt{2}}\psi_{n,k}(2t)
    \quad,\quad
    \psi_{n+1,k+2^n}(t)=\frac{1}{\sqrt{2}}\psi_{n,k}(2t-1),
\end{equation*}
and setting
\begin{equation*}
    \beta_N'(s) \coloneq
    \sqrt{2}\sum_{n=0}^{N}\sum_{k=0}^{2^{n}-1}\xi_{n+1,k}\psi_{n+1,k}(s/2)
    \quad,\quad
    \beta_N''(s)\coloneq
    \sqrt{2}\sum_{n=0}^{N}\sum_{k=0}^{2^{n}-1}\xi_{n+1,k+2^n}\psi_{n+1,k+2^n}\big((s+1)/2\big), 
\end{equation*}
we have that $\beta_N'$ and $\beta_N''$ are two copies of $\beta_N$
independent from each other and from $\xi$ and $\xi_{0,0}$, leading to
\begin{equation}\label{eq:cascade_proof}
    W_{N+1}(t) = \xi t + \xi_{0,0} \psi(t) +\frac{1}{\sqrt{2}} \left(\beta'_N(2t)+\beta''_N(2t-1)\right).
\end{equation}

\medskip
\textit{Large-scale modes.}\quad
It remains to handle the two large-scale modes $\xi t$ and $\xi_{0,0}\psi(t)$.
Setting
\begin{equation*}
    \xi' \coloneq \frac{\xi+\xi_{0,0}}{\sqrt{2}}
    \quad,\quad
    \xi'' \coloneq \frac{\xi-\xi_{0,0}}{\sqrt{2}},
\end{equation*}
the pair $(\xi',\xi'')$ is centered and Gaussian, with
\begin{equation*}
    \mathbb{E}\left[(\xi')^2\right]
    =\mathbb{E}\left[(\xi'')^2\right]=1
    \quad,\quad
    \mathbb{E}\left[\xi'\xi''\right]
    =\frac{1}{2}\,\mathbb{E}\left[\xi^2-\xi_{0,0}^2\right]=0,
\end{equation*}
so that $\xi'$ and $\xi''$ are independent standard normal random
variables, independent of $\beta'_N$ and $\beta''_N$. 
The sum of the first two terms in 
\eqref{eq:cascade_proof} can then be expressed as
\begin{equation}\label{eq:cascade_proof2}
    \begin{split}
        \xi t + \xi_{0,0}\psi(t)
        &= \frac{1}{\sqrt{2}}\,(\xi'+\xi'')t
        + \frac{1}{\sqrt{2}}\,(\xi'-\xi'')\psi(t)\\
        &= \frac{1}{\sqrt{2}}\,\xi'\left(t+\psi(t)\right)
        + \frac{1}{\sqrt{2}}\,\xi''\left(t-\psi(t)\right)\\
        &= \frac{1}{\sqrt{2}}\,\xi'\,(2t\wedge 1)
        + \frac{1}{\sqrt{2}}\,\xi''\,\big((2t-1)\vee 0\big),
    \end{split}
\end{equation}
where the last equality follows from the elementary identities
$t+\psi(t)=2t\wedge 1$ and $t-\psi(t)=(2t-1)\vee 0$.
Moreover, $\beta'_N(2t)=0$ for $t\geq\tfrac12$ and
$\beta''_N(2t-1)=0$ for $t\leq\tfrac12$, since the wavelets defining
$\beta'_N$ and $\beta''_N$ have support inside $[0,\tfrac12]$ and
$[\tfrac12,1]$, respectively. Hence
\begin{equation}\label{eq:cascade_proof3}
    \beta'_N(2t\wedge 1)=\beta'_N(2t)
    \quad,\quad
    \beta''_N\big((2t-1)\vee 0\big)=\beta''_N(2t-1),
    \qquad 0\leq t\leq 1.
\end{equation}

Finally, substituting Eqs.~\eqref{eq:cascade_proof2} and \eqref{eq:cascade_proof3} into \eqref{eq:cascade_proof} 
we obtain
\begin{equation*}
    \begin{split}
        W_{N+1}(t)
        &= \frac{1}{\sqrt{2}}
        \left(
        \xi'\,(2t\wedge 1)+\beta'_N(2t\wedge 1)
        \right)
        + \frac{1}{\sqrt{2}}
        \left(
        \xi''\,\big((2t-1)\vee 0\big)
        +\beta''_N\big((2t-1)\vee 0\big)
        \right)\\
        &= \frac{1}{\sqrt{2}}
        \left(
        W'_N(2t\wedge 1) + W''_N\big((2t-1)\vee 0\big)
        \right),
    \end{split}
\end{equation*}
where
\begin{equation*}
    W'_N(s) \coloneq \xi' s + \beta'_N(s)
    \quad,\quad
    W''_N(s) \coloneq \xi'' s + \beta''_N(s)
\end{equation*}
are two independent copies of $W_N$.
\end{proof}

Next, we express the
cascade relation \eqref{eq:cascade relation}
at the level of probability measures.
To this end, we define the path rescaling transformations
$F,G : \Omega\to \Omega$ by
\begin{equation}\label{eq:FandG}
    \begin{aligned}
        (F\omega)(t) &=
        \begin{cases}
            \dfrac{1}{\sqrt 2}\,\omega(2t),
            \quad\,\,\, & t \in [0,\tfrac{1}{2}], \\[0.3cm]
            \dfrac{1}{\sqrt 2}\,\omega(1), & t \in [\tfrac{1}{2},1],
        \end{cases}
        \\[0.5cm]
        (G\omega)(t) &=
        \begin{cases}
            0, & t \in [0,\tfrac{1}{2}], \\[0.3cm]
            \dfrac{1}{\sqrt 2}\,\omega(2t-1), & t \in [\tfrac{1}{2},1],
        \end{cases}
    \end{aligned}
    \end{equation}
so that, for every $\omega\in\Omega$,
\begin{equation*}
    (F\omega)(t) = \frac{1}{\sqrt 2}\,\omega(2t\wedge 1)
    \quad,\quad
    (G\omega)(t) = \frac{1}{\sqrt 2}\,\omega\big((2t-1)\vee 0\big).
\end{equation*}
Denoting by $\mu_N$ the law of $W_N$ on $\Omega$, the cascade
relation \eqref{eq:cascade relation} becomes
\begin{equation}\label{eq:cascade relation2}
   \mu_{N+1} =
   \left(F_\#\mu_N\right) * \left(G_\#\mu_N\right),
\end{equation}
where $F_\# \mu_N$ and $G_\# \mu_N$ denote the pushforwards\footnote{One can easily check that $F$ and $G$ are continuous linear maps on $\Omega$, so that the pushforwards $F_\#\mu_N$ and $G_\#\mu_N$ are well-defined.}
of $\mu_N$ by the transformations $F$ and $G$, respectively,
and $*$ denotes the convolution of probability measures.

In analogy with the CLT,
expression \eqref{eq:cascade relation2}
can be viewed as a deterministic transformation on
probability laws:
the law at scale $N+1$ is obtained by convolving two
rescaled copies of the law at scale $N$.
This motivates the following definition
of the RG operator.

\begin{definition}
    Let $\M$ be the space of probability measures on $\Omega$.
    The RG operator is the map $\mR: \M\to\M$ defined by
    \begin{equation}\label{eq:RG dynamics}
        \mR\left[\mu\right]=
        \left(F_\#\mu\right) * \left(G_\#\mu\right),
        \quad \mu \in \M.
    \end{equation}
\end{definition}

It follows that relation \eqref{eq:cascade relation2}
holds not only for the finite-resolution approximations of Brownian
motion, but also for the limiting measure $\mu_\infty$.
From the RG perspective, this means that the Wiener measure
is a fixed point of $\mR$, as stated in the next result.

\begin{proposition}
    The Wiener measure $\mu_\infty$ obeys the fixed point
    equation
    \begin{equation}\label{eq:fixedpoint}
        \mR\left[\mu_\infty\right] = \mu_\infty.
    \end{equation}
\end{proposition}
\begin{proof}
    Let $W'$ and $W''$ be independent standard Brownian motions on
    $[0,1]$ and set $X\coloneq FW'+GW''$, so that
    $\mathrm{Law}(X)=\mR[\mu_\infty]$. Explicitly,
    \begin{equation*}
        X(t)=
        \begin{cases}
            \dfrac{1}{\sqrt 2}\,W'(2t),
            & t \in [0,\tfrac{1}{2}], \\[0.3cm]
            \dfrac{1}{\sqrt 2}\,W'(1)
            +\dfrac{1}{\sqrt 2}\,W''(2t-1),
            & t \in [\tfrac{1}{2},1].
        \end{cases}
    \end{equation*}
    Since $F$ and $G$ are linear and $(W',W'')$ is jointly Gaussian,
    the process $X$ is centered and Gaussian, with continuous paths
    and $X(0)=0$. It remains to compute its covariance. 
    Without loss of generality, assume $0\leq s\leq t\leq 1$.

    If $s\leq t\leq\tfrac12$, then
    \begin{equation*}
        \mathbb{E}\left[X(s)X(t)\right]
        =\frac{1}{2}\,\mathbb{E}\left[W'(2s)W'(2t)\right]
        =\frac{1}{2}\,(2s\wedge 2t)=s.
    \end{equation*}
    On the other hand, if $s\leq\tfrac12\leq t$, we have
    \begin{equation*}
            \mathbb{E}\left[X(s)X(t)\right]
            =\frac{1}{2}\,
            \mathbb{E}\left[
            W'(2s)\left(W'(1)+W''(2t-1)\right)
            \right]
            =\frac{1}{2}\,\mathbb{E}\left[W'(2s)\,W'(1)\right]
            =\frac{1}{2}\,(2s\wedge 1)=s,
    \end{equation*}
    where the cross terms vanish by independence and centering.
    Similarly, if $\tfrac12\leq s\leq t$, then
    \begin{equation*}
        \mathbb{E}\left[X(s)X(t)\right]
        =\frac{1}{2}
        \left(
        \mathbb{E}\left[W'(1)^2\right]
        +\mathbb{E}\left[W''(2s-1)\,W''(2t-1)\right]
        \right)
        =\frac{1}{2}\left(1+(2s-1)\right)=s.
    \end{equation*}
    In all cases $\mathbb{E}[X(s)X(t)]=s\wedge t$, so that $X$ is a
    standard Brownian motion and $\mathrm{Law}(X)=\mu_\infty$.
\end{proof}

Having identified $\mu_\infty$ as a fixed point of the
renormalization group dynamics, the next step is to analyze its
stability.
In particular, we wish to understand how small perturbations of
$\mu_\infty$ evolve under successive applications of the RG
operator.
This is achieved by linearizing $\mR$ in a neighborhood of
$\mu_\infty$.
Formally, let
\begin{equation*}
    \mu = \mu_\infty + \varepsilon \nu,
\end{equation*}
where $\nu$ is a signed measure on $\Omega$ with zero total mass,
and $\varepsilon>0$ is a small parameter.
Substituting this expansion into the RG relation
\eqref{eq:RG dynamics} and retaining
terms of order $\varepsilon$ yields the linearized RG operator
\begin{equation}\label{eq:definition_L}
    \L\nu
    = (F_\# \nu) * (G_\# \mu_\infty)
    + (F_\# \mu_\infty) * (G_\# \nu).
\end{equation}
This operator governs the first-order dynamics of perturbations
around the Brownian fixed point and is the central object of the analysis
in the following sections.
\section{The dual operator $\L^*$}\label{sec:dual-operator}

The operator $\L$ is naturally defined on signed measures; 
however, to fully analyze its spectral properties, 
it is necessary to extend its action to a larger space of distributions 
(generalized functionals). 
This extension makes it possible to detect generalized eigenvectors 
of $\L$ that cannot be represented as measures.
The natural framework for such an extension is provided by the theory of 
Hida distributions,
which builds on the analytic machinery of Hida calculus 
(white noise analysis). 

The extension is carried out by following the standard duality principle
for linear operators.
We first introduce the dual operator $\L^*$ acting on Hida test
functionals.
The role of $\L^*$ is precisely to transfer the action of $\L$ from
measures to functionals through the dual pairing between measures and
test functionals.
Once $\L^*$ is shown to be a continuous linear map from the Hida test
space into itself, the extension of $\L$ to Hida distributions is
obtained by duality.

\begin{definition}[Dual operator]\label{def:Lstar}
    Let $\Phi:\Omega\to\mathbb{R}$ be measurable. 
    We define the dual operator $\L^*$ by the expression
    
    \begin{equation}\label{eq:Lstar}
        \L^*\Phi(\omega)
        \coloneq 
        \int_{\Omega}
        \Phi\bigl(F\omega + G\omega_2 \bigr)
        \, d\mu_\infty(\omega_2) 
         +
        \int_{\Omega}
        \Phi\bigl(F\omega_1 + G\omega \bigr)
        \, d\mu_\infty(\omega_1),
    \end{equation}    
    whenever the right-hand side is well-defined and finite.
\end{definition}    

The definition above specifies $\L^*$ as a conditional-expectation
operator. In this sense it is reminiscent of the Mehler representation
of the Ornstein--Uhlenbeck (OU) semigroup on Wiener space, where the
action of the semigroup can be written as a conditional expectation by
combining two independent copies of the underlying Brownian path 
with their amplitudes rescaled. 
In expression \eqref{eq:Lstar}, however, the maps $F$ and $G$ act not
only by rescaling the amplitude, but also by rescaling time:
as \eqref{eq:FandG} shows, they encode path compressions and shifts, restricting
the trajectory to the left and right halves of the interval.
We adopt this formulation because it is precisely the one for which the
natural duality relation with $\L$ holds, as stated in the proposition
below.

\begin{proposition}[Duality]
    Let $\nu$ be a signed measure on $\Omega$
    and let $\Phi:\Omega\to\mathbb{R}$ be measurable.
    Then
    \begin{equation}\label{eq:duality}
    \int_{\Omega} \Phi(\omega)\, d(\L\nu)(\omega)
    \;=\;
    \int_{\Omega} (\L^*\Phi)(\omega)\, d\nu(\omega),
    \end{equation}
    whenever both integrals exist and are finite.
\end{proposition}
\begin{proof}
    By expression \eqref{eq:definition_L}, and the definitions
    of convolution and pushforward of measures,
    it follows that
    \begin{equation*}
        \begin{split}
        \int_{\Omega} \Phi(\omega)\, d(\L\nu)(\omega)
        = &
        \int_{\Omega}\Phi(\omega)
        d\bigl((F_\# \nu) * (G_\# \mu_\infty) \bigr)(\omega)
         +
        \int_{\Omega}
        \Phi(\omega)
        d\bigl((F_\# \mu_\infty) * (G_\# \nu)\bigr)(\omega)\\[0.2cm]
        =&
        \int_{\Omega^2}
        \Phi(F\omega+G\omega_2)
        d\nu(\omega) d\mu_\infty(\omega_2)
        +
        \int_{\Omega^2}
        \Phi(F\omega_1+G\omega)
        d\mu_\infty(\omega_1)d\nu(\omega)\\[0.2cm]
        =& 
        \int_{\Omega} (\L^*\Phi)(\omega)\, d\nu(\omega),
        \end{split}
    \end{equation*}
    where the last equality follows from \cref{def:Lstar} and 
    Fubini's theorem.
\end{proof}

\begin{proposition}[$L^p$-boundedness of $\L^*$]
    \label{thm: Lstar is bounded}
    For $1\leq p\leq \infty$, the dual operator $\L^*$ is a continuous linear map from 
    $L^p(\Omega,\mathscr{B},\mu_\infty)$ into $L^p(\Omega,\mathscr{B},\mu_\infty)$. 
\end{proposition}
\begin{proof}
    The case $p=\infty$ is immediate.
    For $1\leq p < \infty$,
    we use the elementary inequality $(a+b)^p \le 2^{p-1}(a^p+b^p)$, 
    along with Jensen's inequality for the expectations to obtain
    \begin{equation*}
        \begin{split}
        |\L^*\Phi(\omega)|^p
        \leq &
        \; 2^{p-1}\Bigg(
        \int_{\Omega}
        |\Phi\bigl( F\omega + G\omega_2 \bigr)|^p
        \, d\mu_\infty(\omega_2) 
     +
        \int_{\Omega}
        |\Phi\bigl( F\omega_1 + G\omega \bigr)|^p
        \, d\mu_\infty(\omega_1)  \Bigg).
        \end{split}
    \end{equation*}
    Integrating further with respect to $\mu_\infty$, 
    and using the fixed-point relation 
    $\mu_\infty=F_\#\mu_\infty*G_\#\mu_\infty$ we find
    $$
    \norm{\L^*\Phi}_{L^p(\mu_\infty)}^p\leq 2^{p}\norm{\Phi}_{L^p(\mu_\infty)}^p.
    $$
\end{proof}

\subsection{The operators $F^*$ and $G^*$}

The relevant properties of $\L^*$ are most easily expressed in terms of
the auxiliary operators $F^*,G^*: L^2([0,1])\to L^2([0,1])$ given by
\begin{equation}\label{eq:F*_G*}
    (F^* h)(s) = \frac{1}{\sqrt{2}}\,h\Big(\frac{s}{2}\Big)
    \quad,\quad
    (G^* h)(s) = \frac{1}{\sqrt{2}}\,h\Big(\frac{s+1}{2}\Big).
\end{equation} 
They can be viewed as
adjoints of the rescaling maps $F$ and $G$, 
with adjointness understood with respect 
to the Brownian stochastic integral.
\begin{proposition}\label{thm:F*_G*_adjoint}
    For every $h\in L^2([0,1])$ we have,
    \begin{equation}\label{eq:adjoint_relation}
        \begin{split}
            \int_0^1 h(t)\,dW_t(F\omega)
            &= \int_0^1 (F^* h)(s)\,dW_s(\omega),\\[0.2cm]
            \int_0^1 h(t)\,dW_t(G\omega)
            &= \int_0^1 (G^* h)(s)\,dW_s(\omega),
        \end{split}
    \end{equation}
    as identities in $L^2(\Omega,\mathscr{B},\mu_\infty)$.
\end{proposition}
\begin{remark}
    The integral in the left-hand side of the first identity in \eqref{eq:adjoint_relation} 
    is understood as an It\^o integral with respect to 
    the process $W_t(F\omega)=\frac{1}{\sqrt 2}\,\omega(2t\wedge 1)$.
    Indeed, this process is a martingale with respect to the time-changed filtration 
    $\mathcal F_{2t\wedge1}$ with quadratic variation $t\wedge \frac12$,
    where $(\mathcal F_t)_{t\in[0,1]}$ denotes the natural filtration of Brownian motion $W_t$.
    Similarly, the process $W_t(G\omega)=\frac{1}{\sqrt 2}\,\omega((2t-1)\vee 0)$
    is a martingale with respect to the filtration $\mathcal F_{(2t-1)\vee 0}$
    and has quadratic variation $(t-\frac12)\vee 0$. 
    Hence, the integral in the left-hand side of the 
    second identity is also well defined as an It\^o integral.
\end{remark}
\begin{proof}[Proof of \cref{thm:F*_G*_adjoint}]
    We first consider
    $h=\mathbf{1}_{(a,b]}$, with $0\leq a<b\leq1$. 
    In this case, we have
    \[
        \int_0^1 h(t)\,dW_t(F\omega)
        =
        W_b(F\omega)-W_a(F\omega)
        =
        \frac{1}{\sqrt2}
        \bigl(
            W_{2b\wedge1}(\omega)
            -
            W_{2a\wedge1}(\omega)
        \bigr).
    \]
    On the other hand, by \eqref{eq:F*_G*},
    \[
        \int_0^1(F^*h)(s)\,dW_s(\omega)
        =
        \frac1{\sqrt2}
        \int_0^1
        \mathbf 1_{(2a\wedge1,\,2b\wedge1]}(s)\,dW_s
        =
        \frac{1}{\sqrt2}
        \bigl(
            W_{2b\wedge1}(\omega)
            -
            W_{2a\wedge1}(\omega)
        \bigr).
    \]
    This proves the first identity for indicator functions.
    By linearity, it also holds for step functions.
    Since $F^*$ is a bounded linear operator on $L^2([0,1])$, 
    and step functions are dense in $L^2([0,1])$, 
    It\^o isometry yields the first identity for every
    $h\in L^2([0,1])$.
    The proof of the identity involving $G^*$ is analogous.
\end{proof}

In order to obtain analytic control on $\L^*$ in the Hida scale, it is
useful to track how the auxiliary operators $F^*$ and $G^*$ interact
with the deterministic Hilbert scales $\S_p$, which are defined using
the Haar wavelet basis, see \cref{sec:deterministic}. Since the
underlying maps $F$ and $G$ implement a dyadic time compression and
shift, their adjoints inherit a correspondingly simple action in wavelet
coordinates. This is made precise in
\cref{thm:adjoints and wavelets} below. This dyadic compatibility
ultimately implies that $F^*$ and $G^*$ act as contractions on $\S_p$.
Moreover, by multiplicativity of the Hilbert tensor norm, the tensor
powers $(F^*)^{\otimes n}$ and $(G^*)^{\otimes n}$ inherit the same
contraction property on $\S_p^{\otimes n}$, as stated in
\cref{thm: adjoint norms}.

\begin{lemma}\label{thm:adjoints and wavelets}
    Let $f\in L^2([0,1])$ and
    $0\leq k\leq 2^n-1$. 
    Then 
    $$\ps{F^*f,e_{n,k}}_{L^2}=\ps{f,e_{n+1,k}}_{L^2},\quad 
    \ps{G^*f,e_{n,k}}_{L^2}=\ps{f,e_{n+1,\,k+2^n}}_{L^2}.
    $$
\end{lemma}
\begin{proof}
    Using the change of variables $s=2t$ we have
    \[
    \ps{F^*f,e_{n,k}}_{L^2}
    =\frac{\sqrt2}{2}\int_0^1  f(s/2)\, e_{n,k}(s)\,ds
    =\sqrt2\int_0^{1/2} f(t)\, e_{n,k}(2t)\,dt.
    \]
    The Haar scaling relation $e_{n+1,k}(t)=\sqrt2\, e_{n,k}(2t)$
    yields
    \[
    \ps{F^*f,e_{n,k}}_{L^2}=\int_0^{1/2} f(t)e_{n+1,k}(t)\,dt
    =\ps{f,e_{n+1,k}}_{L^2},
    \]
    where the last equality comes from the fact that 
    $\mathrm{supp}(e_{n+1,k})\subset[0,1/2]$ if $k\le 2^n-1$. 
    
    The case of $G^*$ is completely analogous: we use the
    change of variables $s=2t-1$, 
    the scaling relation $e_{n+1,k+2^n}(t)=\sqrt2\, e_{n,k}(2t-1)$,
    and the fact that 
    $\mathrm{supp}(e_{n+1,k+2^n})\subset[1/2,1]$ if $k\le 2^n-1$. 
    We leave the 
    details to the reader.
\end{proof}

\begin{lemma}\label{thm: adjoint norms-combined}
    For every $p\geq 0$ and $f\in \S_p$ we have
    $$|F^*f|_p^2 + |G^*f|_p^2 \leq |f|_p^2.$$
\end{lemma}
\begin{proof}
    We first consider the wavelet components.
    Using \cref{thm:adjoints and wavelets}, we obtain
    \begin{equation*}
    \begin{split}
    \sum_{n=0}^\infty\sum_{k=0}^{2^n-1} 
    \alpha_n^{2p}\left(|\ps{F^*f,e_{n,k}}_{L^2}|^2 + |\ps{G^*f,e_{n,k}}_{L^2}|^2\right)
    &\leq\sum_{n=0}^\infty\sum_{k=0}^{2^n-1} \alpha_{n+1}^{2p}\left(|\ps{f,e_{n+1,k}}_{L^2}|^2 + |\ps{f,e_{n+1,k+2^n}}_{L^2}|^2\right)\\
    &=\sum_{m=1}^\infty\sum_{k=0}^{2^{m-1}-1} \alpha_{m}^{2p}\left(|\ps{f,e_{m,k}}_{L^2}|^2 + |\ps{f,e_{m,k+2^{m-1}}}_{L^2}|^2\right)\\
    &=\sum_{m=1}^\infty\sum_{k=0}^{2^{m}-1} \alpha_{m}^{2p}|\ps{f,e_{m,k}}_{L^2}|^2
    \end{split}
    \end{equation*}
    since the sequence $\{\alpha_n\}_{n\ge -1}$ is increasing and $p\ge 0$.
    The last equality follows from the observation that 
    the $k$-sum in the second line spans all the wavelets at level $m$.
    
    To bound the constant modes, we use the  $L^2$-identities $\mathbf{1}_{[0,1/2)}=\tfrac12(\mathbf{1}+e_{0,0})$ 
    and $\mathbf{1}_{[1/2,1)}=\tfrac12(\mathbf{1}-e_{0,0})$ to obtain
    \begin{equation*}
        \ps{F^*f,\mathbf{1}}_{L^2}=\frac{2}{\sqrt2}\,\ps{f,\mathbf{1}_{[0,1/2)}}_{L^2}
        = \frac{1}{\sqrt2} \left(\ps{f,\mathbf{1}}_{L^2}+\,\ps{f,e_{0,0}}_{L^2}\right),
    \end{equation*}
    \begin{equation*}
        \ps{G^*f,\mathbf{1}}_{L^2}=\frac{2}{\sqrt2}\,\ps{f,\mathbf{1}_{[1/2,1)}}_{L^2}
        = \frac{1}{\sqrt2} \left(\ps{f,\mathbf{1}}_{L^2}-\,\ps{f,e_{0,0}}_{L^2}\right),
    \end{equation*}
    so that
    \begin{equation*}
        \begin{split}
            \alpha_{-1}^{2p}\left(|\ps{F^*f,\mathbf{1}}_{L^2}|^2+|\ps{G^*f,\mathbf{1}}_{L^2}|^2\right)
            &=\alpha_{-1}^{2p} \left(|\ps{f,\mathbf{1}}_{L^2}|^2+\,|\ps{f,e_{0,0}}_{L^2}|^2\right) \\
            &\leq  \alpha_{-1}^{2p}|\ps{f,\mathbf{1}}_{L^2}|^2 + \alpha_{0}^{2p} |\ps{f,e_{0,0}}_{L^2}|^2.
        \end{split}
    \end{equation*}
    Using the definition of $|\cdot|_p$ (see \cref{sec:deterministic}) 
    and combining the above estimates, we get 
    \begin{equation*}
        \begin{split}
            |F^*f|_p^2 &+ |G^*f|_p^2\\
            &=\alpha_{-1}^{2p}\left(|\ps{F^*f,\mathbf{1}}_{L^2}|^2+|\ps{G^*f,\mathbf{1}}_{L^2}|^2\right)
            + \sum_{n=0}^\infty\sum_{k=0}^{2^n-1} 
                \alpha_n^{2p}\left(|\ps{F^*f,e_{n,k}}_{L^2}|^2 + |\ps{G^*f,e_{n,k}}_{L^2}|^2\right)\\
            &\leq\alpha_{-1}^{2p}\,|\ps{f,\mathbf{1}}_{L^2}|^2
            +\sum_{n=0}^\infty\sum_{k=0}^{2^n-1} \alpha_n^{2p}\,|\ps{f,e_{n,k}}_{L^2}|^2
            = |f|_p^2,
        \end{split}
    \end{equation*}
    concluding the proof.
\end{proof}

\begin{corollary}\label{thm: adjoint norms}
    For every $n\geq 1$ and $p\geq 0$, the operators  
    $(F^*)^{\otimes n}$ and $(G^*)^{\otimes n}$
    are bounded linear maps on $\S_p^{\otimes n}$
    with operator norms 
    $$
    \norm{(F^*)^{\otimes n}}_{\S_p^{\otimes n} \to \S_p^{\otimes n}} \leq 1,
    \quad 
    \norm{(G^*)^{\otimes n}}_{\S_p^{\otimes n} \to \S_p^{\otimes n}} \leq 1.
    $$
\end{corollary}
\begin{proof}
    From \cref{thm: adjoint norms-combined}
    we deduce that $\|F^*\|_{\S_p\to \S_p}\le 1$
    and 
    $\|G^*\|_{\S_p\to \S_p}\le 1$.
    Using the multiplicativity 
    of the Hilbert tensor norm 
    we deduce that the tensorized versions 
    $(F^*)^{\otimes n}$ and $(G^*)^{\otimes n}$ 
    are bounded linear operators on the algebraic tensor product  
    with operator norms bounded by 1.
    Density of the algebraic tensor product in $\S_p^{\otimes n}$  
    yields the claim.
\end{proof}

\subsection{Second quantization formula}

The goal of this subsection is to prove \cref{thm:sec quanti}, which
provides a second quantization representation of $\L^*$, allowing us to
write symbolically $\L^*=\Gamma(F^*)+\Gamma(G^*)$.
Here $\Gamma$ denotes the second quantization operator introduced in
\cref{sec:stochastic}.
This second quantization characterization is important because it yields
an explicit, chaoswise description of $\L^*$. Indeed, since Hida norms
are defined in terms of the kernels in the Wiener chaos expansion, any
estimate in the Hida scale reduces to understanding how $\L^*$ transforms
the chaos kernels of a functional $\Phi$.
\cref{thm:sec quanti} provides exactly this
kernel-level transformation, and therefore is a natural tool for
analysis in the Hida framework.
Before proceeding to the proof of \cref{thm:sec quanti}, we need 
the following lemma.

\begin{lemma}[Adjoint change of variables] 
    \label{thm:adj change}
    Let $W$ be the coordinate process on $\Omega$ and 
    let $\mathcal W$ be the isonormal Gaussian process defined by 
    \eqref{eq:isonormal}.
    On the product space 
    $\left(\Omega^2, \mathscr{B}^2,\mu_\infty\otimes\mu_\infty\right)$
    define the projections $\pi_i(\w_1,\w_2)=\w_i$, $i=1,2$.
    Let $W^{(1)}_t\coloneqq W_t\circ\pi_1$ and 
    $W^{(2)}_t\coloneqq W_t\circ\pi_2$
    be two independent copies of $W$.
    Define the independent isonormal families
    \begin{equation*}
       \mathcal W^{(1)}(h)
      \coloneqq\int_0^1 h\,dW^{(1)},
        \qquad
        \mathcal W^{(2)}(h)
        \coloneqq\int_0^1 h\,dW^{(2)}.
    \end{equation*}
    Then, 
    for every $h\in L^2([0,1])$,
    \begin{equation}\label{eq:isonormal adjoint}
        \mathcal{W}(h)\circ S_{FG}=\mathcal{W}^{(1)}(F^*h)+ \mathcal{W}^{(2)}(G^*h)
        \quad \text{in}\quad  L^2\left(\mu_\infty\otimes\mu_\infty\right),
    \end{equation}
    where $S_{FG}(\w_1,\w_2)= F\w_1+G\w_2$.
\end{lemma}
\begin{remark}
    We note that 
    the composition with $S_{FG}$ is well defined on 
    $L^2$-equivalence classes since 
    $(S_{FG})_\#(\mu_\infty\otimes\mu_\infty)=\mu_\infty$
    by the fixed-point relation \eqref{eq:fixedpoint},
    so that $\mu_\infty$-null sets pull back to 
    $(\mu_\infty\otimes\mu_\infty)$-null sets.
\end{remark}
\begin{proof}[Proof of \cref{thm:adj change}]
    This follows from the linearity of the Wiener integral with respect to the integrator 
    together with \cref{thm:F*_G*_adjoint} applied to the two independent Brownian motions.
\end{proof}

\begin{theorem}[Second quantization formula] \label{thm:sec quanti}
    Let $\Phi\in L^2(\Omega,\mathcal{F}_\mathcal{W},\mu_\infty)$
    have chaos expansion 
    $$\Phi=\sum_{n=0}^{\infty} I_n(f_n)\quad \text{in}\quad L^2(\Omega,\mathcal{F}_\mathcal{W},\mu_\infty).$$
    Then,  $\L^*\Phi$ has chaos expansion 
    \begin{equation*}
        \L^*\Phi
        =\sum_{n=0}^{\infty}\Big(I_n\big((F^*)^{\otimes n} f_n\big)
        + I_n\big((G^*)^{\otimes n} f_n\big)\Big)
        \quad \text{in}\quad L^2(\Omega,\mathcal{F}_\mathcal{W},\mu_\infty).
    \end{equation*}
\end{theorem}
\begin{proof}
    Assume first that $n\ge1$ and that
    $\Phi=I_n(h^{\otimes n})$ for $h\in L^2([0,1])$,
    with $F^*h\neq0$ and $G^*h\neq0$.
    Then, using Proposition 1.1.4 in 
    \cite{nualart2006malliavin} we may write 
    \begin{equation*}
            \Phi\circ S_{FG}
            =|h|_{L^2}^n 
            H_n(\mathcal{W}(h/|h|_{L^2}))
            \circ S_{FG}
            =|h|_{L^2}^n 
            H_n(\mathcal{W}(h/|h|_{L^2})
            \circ S_{FG})\quad \mu_\infty\otimes\mu_\infty\text{-a.s.},
    \end{equation*}
    where $H_n$ is the $n$-th Hermite polynomial 
    defined in \eqref{eq:hermite-def}.

    Using \cref{thm:adj change} and 
    setting 
    $X_1\coloneqq\mathcal{W}^{(1)}(F^*h)$, and
    $X_2\coloneqq\mathcal{W}^{(2)}(G^*h)$
    yields
    \begin{equation*}
        \begin{split}
            \mathcal{W}(h/|h|_{L^2})\circ S_{FG} 
            &=
            \frac{|F^*h|_{L^2}}
            {|h|_{L^2}} 
            \frac{X_1}{|F^*h|_{L^2}} 
            + \frac{|G^*h|_{L^2}}
            {|h|_{L^2}} 
            \frac{X_2}{|G^*h|_{L^2}} 
            \\[0.2cm]
            & \eqcolon b_1 \widetilde X_1 + b_2 \widetilde X_2,
        \end{split}
    \end{equation*}
    with $b_1^2+b_2^2=1$, 
    which follows from the identity
    $|h|_{L^2}^2 =  |F^*h|_{L^2}^2 + |G^*h|_{L^2}^2$.
    Using the Hermite addition formula in 
    \cite[Corollary~5.10]{hairer2026advanced},
    leads to
    \begin{equation*}
        \Phi\circ S_{FG}
        =|h|_{L^2}^n H_n(b_1 \widetilde X_1 + b_2 \widetilde X_2)
        =|h|_{L^2}^n \sum_{k=0}^n 
        \binom{n}{k}
        b_1^k b_2^{n-k}\,
        H_k(\widetilde X_1)
        H_{n-k}(\widetilde X_2).
    \end{equation*}

    Since $\widetilde X_1$, 
    and $\widetilde X_2$ are 
    independent standard normals,
    and $\E_{\mu_\infty\otimes\mu_\infty}
    [H_n(\widetilde X_i)]=0$ 
    for $n\ge1$ and $i=1,2$,
    taking the conditional expectation 
    with respect to $W^{(1)}$
    only the index $k=n$ survives:
    $$
    \E_{\mu_\infty\otimes\mu_\infty}
    \left[\Phi\circ S_{FG}\, 
    \middle|\, W^{(1)}\right]
    =|h|_{L^2}^n\,b_1^n\,
    H_n\big(\widetilde X_1\big)
    =|F^*h|_{L^2}^{\,n}\,
    H_n\Big(\frac{\mathcal W^{(1)}(F^*h)}
    {|F^*h|_{L^2}}\Big)
    =I_n\big((F^*h)^{\otimes n}\big).
    $$  
    Here the last equality is due to
    Proposition 1.1.4 in \cite{nualart2006malliavin}. 
    Similarly, taking the conditional 
    expectation with respect to $W^{(2)}$
    only the index $k=0$ survives:
    $$
    \E_{\mu_\infty\otimes\mu_\infty}
    \left[\Phi\circ S_{FG}\,
    \middle|\, W^{(2)}\right]
    =|h|_{L^2}^n\,b_2^n\,
    H_n\big(\widetilde X_2\big)
    =|G^*h|_{L^2}^{\,n}\,
    H_n\Big(\frac{\mathcal W^{(2)}(G^*h)}
    {|G^*h|_{L^2}}\Big)
    =I_n\big((G^*h)^{\otimes n}\big).
    $$
    The cases where $F^*h=0$ or $G^*h=0$ follow directly, and the case $n=0$ is immediate.
    Since $\L^*\Phi$ is the sum of these two conditional expectations, we have
    $$
    \L^*\Phi
    =I_n\big((F^*h)^{\otimes n}\big)+I_n\big((G^*h)^{\otimes n}\big).
    $$
    
    By linearity and continuity, the identity 
    for $\Phi=I_n(h^{\otimes n})$ 
    extends by polarization to finite linear 
    combinations of symmetrized kernels 
    $h_1\otimes\cdots\otimes h_n$.
    By density of these kernels in 
    $L^2([0,1])^{\hat\otimes n}$ 
    together with It\^o isometry,
    continuity of 
    $(F^*)^{\otimes n}$ and 
    $(G^*)^{\otimes n}$ on 
    $L^2([0,1])^{\hat\otimes n}$ 
    (\cref{thm: adjoint norms} with $p=0$),
    and continuity of $\L^*$ on 
    $L^2(\Omega,\mu_\infty)$ 
    (\cref{thm: Lstar is bounded} with $p=2$),
    it holds for every $\Phi=I_n(f_n)$ with 
    $f_n\in 
    L^2([0,1])^{\hat\otimes n}$. 
    Finally, 
    applying $\L^*$ to the partial sums 
    $\sum_{n\le N} I_n(f_n)$ 
    and letting $N\to\infty$ yields the result 
    for any $\Phi\in 
    L^2(\Omega,\mathcal F_{\mathcal W},
    \mu_\infty)$.
\end{proof}

\subsection{Extension of $\L$ to Hida distributions}

Combining the second quantization representation of $\L^*$ with the
contraction estimates for $(F^*)^{\otimes n}$ and $(G^*)^{\otimes n}$
yields continuity of $\L^*$ in the Hida test topology.

\begin{proposition} 
    The dual operator $\L^*$ is a continuous linear map from 
    $(\S)$ into $(\S)$. 
\end{proposition}
\begin{proof}
    Linearity is clear from the definition of $\L^*$.
    Let $p\ge 0$ and $\Phi\in(\S)$ have chaos expansion
    \[
    \Phi=\sum_{n=0}^\infty I_n(f_n).
    \]
    Using \cref{thm:sec quanti} and \cref{thm: adjoint norms} we obtain
    \begin{equation*}
        \begin{split}
            \|\L^*\Phi\|_p^2
            =\sum_{n=0}^\infty n!\,\big|(A^p)^{\otimes n}
            \bigl[(F^*)^{\otimes n}+(G^*)^{\otimes n}\bigr] f_n\big|_{L^2([0,1]^n)}^2
            &\le 
            2\sum_{n=0}^\infty n! \left( \big|(F^*)^{\otimes n}f_n\big|_{S_p^{\otimes n}}^2 +\big|(G^*)^{\otimes n}f_n\big|_{S_p^{\otimes n}}^2\right)\\
            &\leq 4\sum_{n=0}^\infty n!|f_n|_{S_p^{\otimes n}}^2 \\[0.2cm]
            &=4\norm{\Phi}_p^2
        \end{split}
    \end{equation*}
    Hence $\norm{\L^*\Phi}_p\le 2\|\Phi\|_p$ 
    for every $p\ge 0$, so $\L^*:(\S)\to(\S)$ is continuous with
    respect to the 
    topology of the Hida test space $(\S)$ (see \cref{sec:stochastic}).
\end{proof}

We can now extend the action of $\L$ from measures to
general Hida distributions by extending the duality relation \eqref{eq:duality}.
\begin{definition}\label{def:ext}
    Let $U\in (\S)^*$ be a Hida distribution. 
    We define $\L U \in (\S)^*$ to be the Hida distribution 
    given by
    \begin{equation}\label{eq:duality2}
        \pp{\L U}{\Phi}= \pp{U}{\L^*\Phi},
        \quad \Phi \in (\S).
    \end{equation}
\end{definition}
Observe that expression \eqref{eq:duality2} indeed defines 
a Hida distribution, since the composition of continuous linear maps 
is a continuous linear map.
Moreover, when $U$ is a signed measure on $\Omega$,
expression \eqref{eq:duality2} reduces to the duality relation \eqref{eq:duality},
so we recover \eqref{eq:definition_L}.
In this sense, we extend $\L$ 
from measures to general Hida distributions through the generalized 
duality relation \eqref{eq:duality2}.

\section{Generalized spectral theory of $\L$}\label{sec:gen-spec}

In this section, we develop the generalized spectral theory of the
linearized RG operator \(\L\) and give the precise statements 
and proofs of parts (i) and (ii) of \cref{thm:intro-spec};
see \cref{thm:spectral_gap} and \cref{thm:diag_support_eigenkernel} 
below.

Let us briefly describe our approach.
The operator $\L$ is initially defined on signed
measures, but the eigenvalue problem in that setting
is too restrictive: the relevant eigenvectors turn
out to be distributional objects that cannot be
represented as measures.
This calls for an extension of $\L$ to a larger
space, along with a weak formulation of the
eigenvalue problem.
The space $(\S)^*$ of Hida distributions is
particularly well-suited for this purpose, 
allowing us to develop a rich spectral
theory --- establishing general spectral 
properties of $\L$ in the present section --- and
to construct concrete eigenvectors in
\cref{sec:wick-power}.

Having extended $\L$ to $(\S)^*$ in \cref{sec:dual-operator} 
(see \cref{def:ext}), we now formulate the eigenvalue 
problem in the generalized sense.
The approach then rests on reducing this problem
to a family of eigenvalue problems for the
operators $T_n = F^{\otimes n} + G^{\otimes n}$,
one at each level of the Wiener chaos --- a
consequence of the second quantization formula
(\cref{thm:sec quanti}).
Each eigenvalue problem for $T_n$ is also
formulated in the generalized sense, on the
deterministic distribution space
$\S'^{\hat{\otimes} n}$.
The analysis of the adjoints $F^*$ and $G^*$
yields both the characterization of the
generalized spectrum of $\L$
(\cref{thm:spectral_gap}) and the diagonal
concentration of the kernels of generalized
eigenvectors
(\cref{thm:diag_support_eigenkernel}).

\subsection{Generalized eigenvalue problem}

We begin by recalling \cref{def:gen-spec-intro}
here for the reader's convenience, now fully
justified by the analysis of
\cref{sec:dual-operator}, with $\L U$ understood
in the sense of \cref{def:ext}.

\begin{definition}[Generalized eigenvectors and spectrum of $\L$]\label{def:gen-spec-sec5}
    We say that $U\in (\S)^*\setminus \{0\}$ is a \emph{generalized 
    eigenvector} of $\L$ with eigenvalue $\lambda \in \mathbb{C}$ if
    \begin{equation}\label{eq:eigenvector-sec5}
        \pp{\L U}{\Phi} = \lambda \pp{U}{\Phi},
        \quad \text{for all }  \Phi \in (\S).
    \end{equation}
    The \emph{generalized spectrum} of $\L$, denoted by $\sigma(\L)$, 
    is the set of all $\lambda \in \mathbb{C}$ for which 
    \eqref{eq:eigenvector-sec5} holds for some $U \in (\S)^* \setminus \{0\}$.
\end{definition}

At the level of chaos kernels, we adopt the analogous weak
formulation. 
We define $T_n$ on $\S'^{\hat{\otimes} n}$ through
its adjoint action on test kernels, and define the generalized
eigenvalue problem in the same manner.

\begin{definition}[Generalized eigenvectors and spectrum of $T_n$]\label{def:spec-Tn}
    For each $n \geq 0$, define
    \begin{equation*}
        T_n \coloneqq F^{\otimes n} + G^{\otimes n},
        \qquad
        T_n^* \coloneqq (F^*)^{\otimes n}
        + (G^*)^{\otimes n}.
    \end{equation*}
    We say that $u\in\S'^{\hat\otimes n}
    \setminus \{0\}$ is a generalized eigenvector
    of $T_n$ with eigenvalue
    $\lambda \in \mathbb{C}$ if
    \begin{equation}\label{eq:Tn-eigenvector}
      (T_n u, f)\coloneqq(u, T_n^* f) = \lambda\, (u,f),
      \quad \text{for all} \quad
      f\in \S^{\hat\otimes n}.
    \end{equation}
    The \emph{generalized spectrum} of $T_n$,
    denoted by $\sigma(T_n)$, is the set of all
    $\lambda \in \mathbb{C}$ for which
    \eqref{eq:Tn-eigenvector} holds for some
    $u \in \S'^{\hat\otimes n}\setminus\{0\}$.
\end{definition}
    
\begin{remark}
Note that the definition of $T_n u$ via duality is
well-posed. By \cref{thm: adjoint norms},
the operators $(F^*)^{\otimes n}$ and
$(G^*)^{\otimes n}$ are contractions on
$\S_p^{\hat{\otimes} n}$ for every $p \geq 0$,
so $T_n^*$ maps $\S^{\hat{\otimes} n}$ into itself
continuously. For any
$u \in \S'^{\hat{\otimes} n}$, the map
$f \mapsto (u, T_n^* f)$ is therefore a continuous
linear functional on $\S^{\hat{\otimes} n}$,
defining an element
$T_n u \in \S'^{\hat{\otimes} n}$.
\end{remark}

\subsection{Proof of \cref{thm:intro-spec} \textnormal{\emph{(i)}}}
Before proceeding to the proof of part \emph{(i)} of \cref{thm:intro-spec} 
we need the following upper bound on the operator norm of $T_n^*$.

\begin{proposition}\label{prop:Tn-star-bound}
For every \(n\geq 1\) and \(p\geq 0\), $T_n^*$
is a bounded linear operator on \(\S_p^{\otimes n}\) with
\[
    \|T_n^*\|_{\S_p^{\otimes n}\to \S_p^{\otimes n}}
    \leq \sqrt{2}.
\]
\end{proposition}
\begin{proof}
    Define
    \[
        J:\S_p\to \S_p\oplus\S_p,
        \qquad
        Jf=(F^*f,G^*f).
    \]
    By \cref{thm: adjoint norms-combined},
    \[
        |Jf|_{\S_p\oplus\S_p}^2
        =|F^*f|_p^2+|G^*f|_p^2
        \leq |f|_p^2,
    \]
    so $\norm{J}_{\S_p\to\S_p\oplus\S_p}\leq1$. 
    By multiplicativity of the Hilbert tensor norm,
    $\|J^{\otimes n}\|\leq1$. 
    Now $(\S_p\oplus\S_p)^{\otimes n}
    \simeq\bigoplus_{\varepsilon\in\{F,G\}^n}\S_p^{\otimes n}$
    isometrically, and $(F^*)^{\otimes n}f$, $(G^*)^{\otimes n}f$
    are the components of $J^{\otimes n}f$ corresponding to 
    the pure words $\varepsilon=(F,\dots,F)$ and $(G,\dots,G)$.
    Hence
    \[
        |(F^*)^{\otimes n}f|_{\S_p^{\otimes n}}^2
        +
        |(G^*)^{\otimes n}f|_{\S_p^{\otimes n}}^2
        \leq
        |J^{\otimes n}f|^2
        \leq
        |f|_{\S_p^{\otimes n}}^2.
    \]
    Consequently,
    \[
        |T_n^*f|_{\S_p^{\otimes n}}^2
        \leq
        2\left(
        |(F^*)^{\otimes n}f|_{\S_p^{\otimes n}}^2
        +
        |(G^*)^{\otimes n}f|_{\S_p^{\otimes n}}^2
        \right)
        \leq 2|f|_{\S_p^{\otimes n}}^2,
    \]
    and the claim follows.
\end{proof}

The bound above translates, by duality, into a bound on the
generalized spectral radius of $T_n$.

\begin{proposition}[Spectral radius of $T_n$]\label{thm:spec radius} 
    Let $n\geq 1$. If $\lambda \in \sigma(T_n)$ then
    $|\lambda|\leq \sqrt{2}$.
\end{proposition}
\begin{proof}
Suppose $u\in\S'^{\hat\otimes n}\setminus \{0\}$ 
is a generalized eigenvector of $T_n$ with eigenvalue 
$\lambda \in \mathbb{C}$, in the sense of \cref{def:spec-Tn}. 
Then, there exists \(p\geq 0\) such that
\(u\in \S_{-p}^{\hat\otimes n}\), and
\[
|\lambda| |(u,f)| = |(T_n u,f)|=|(u,T_n^* f)|
\leq |u|_{\S_{-p}^{\otimes n}}
\|T_n^*\|_{\S_p^{\otimes n}\to \S_p^{\otimes n}},
\]
for all \(f\in \S^{\hat\otimes n}\) such that
\(|f|_{\S_p^{\otimes n}}=1\).
Taking the supremum over such \(f\) yields
\[
|\lambda| |u|_{\S_{-p}^{\otimes n}}
\leq
|u|_{\S_{-p}^{\otimes n}}
\|T_n^*\|_{\S_p^{\otimes n}\to \S_p^{\otimes n}},
\]
which by \cref{prop:Tn-star-bound} implies
\[
|\lambda|
\leq
\|T_n^*\|_{\S_p^{\otimes n}\to \S_p^{\otimes n}}
\leq \sqrt{2},
\]
concluding the proof.
\end{proof}

We now give the proof of \cref{thm:intro-spec} \emph{(i)}.

\begin{theorem}[Spectral gap]\label{thm:spectral_gap}
We have the inclusion    
$\sigma(\L)\subseteq \{z\in\mathbb{C} : |z|\leq \sqrt{2}\}\cup \{2\}$.
\end{theorem}
\begin{proof}
Let $\lambda\in\mathbb{C}$ be such that $\lambda\neq 2$ and
$$
\pp{\L U}{\Phi}=\pp{U}{\L^*\Phi}=\lambda \pp{U}{\Phi},  
$$    
for some $U\in(\S)^*\setminus \{0\}$ and for all $\Phi\in(\S)$.
Using the second quantization formula (\cref{thm:sec quanti})
and expanding the relation above in Wiener chaos, we find
$$
\sum_{n=0}^\infty n! (u_n,T_n^*f_n)=\lambda \sum_{n=0}^\infty n!(u_n,f_n),
$$
where $\{u_n\}_{n\geq0}\subseteq\S'^{\hat\otimes n}$ are the kernels appearing 
in the expansion of $U$, and 
$\{f_n\}_{n\geq0}\subseteq \S^{\hat\otimes n}$ are the kernels appearing 
in the expansion of $\Phi$; 
see \cref{prop:hida-test-char,thm:generalized chaos}.
By choosing localized test functionals $\Phi=I_n(f_n)$,
it follows that 
$$
(u_n,T_n^*f_n)=\lambda (u_n,f_n),
$$
for every $n\geq0$ and arbitrary $f_n\in\S^{\hat\otimes n}$.

Since $\lambda\neq 2$, we must have $u_0=0$.
Indeed, $T_0^*$ is two times the identity map over constants, 
$u_0\in \R$, and the dual pairing $(u_0,f_0)$, with $f_0\in \R$, 
is the multiplication in $\R$. Thus, the zeroth order 
equation implies that $2 u_0 f_0=\lambda u_0 f_0$, for every 
$f_0\in \R$. If $u_0\ne 0$, then we must have $(2-\lambda)f_0=0$.
Choosing $f_0=1$ leads to a contradiction.
Consequently, 
it follows that  $u_{n_*}\neq 0$ for some $n_*\geq 1$
since $U\ne 0$. 
As a result, $u_{n_*}$ is a generalized eigenvector of $T_{n_*}$
and $\lambda \in \sigma(T_{n_*})$,
in the sense of \cref{def:spec-Tn}.
\cref{thm:spec radius} then implies $|\lambda|\leq \sqrt{2}$,
concluding the proof.
\end{proof}

\begin{remark}
    The eigenvalue $\lambda_0=2$ is simple. Indeed, if
    $\L U = 2U$ with $U \neq 0$, the same chaos-level
    reduction gives $T_n u_n = 2 u_n$ for every
    $n \geq 0$. For $n \geq 1$,
    \cref{thm:spec radius} gives
    $\sigma(T_n) \subseteq \{|z| \leq \sqrt{2}\}$,
    and since $2 > \sqrt{2}$, it follows that
    $u_n = 0$ for all $n \geq 1$.
    Hence $U = I_0(u_0)$ with
    $u_0 \in \mathbb{R} \setminus \{0\}$,
    and for every $\Phi = \sum_n I_n(f_n) \in (\S)$ we have
    $$\pp{U}{\Phi}=u_0f_0=u_0\,\mathbb{E}_{\mu_\infty}[\Phi]
    =u_0\pp{\mu_\infty}{\Phi},
    $$
    from which we conclude that $U = u_0\,\mu_\infty$.
\end{remark}

\subsection{Proof of \cref{thm:intro-spec} \textnormal{\emph{(ii)}}}

We now turn to part \emph{(ii)} of
\cref{thm:intro-spec}, namely the diagonal
concentration of eigenkernels.
We derive a closed-form expression for
the iterates of $T_n^*$
from which its
geometric content becomes transparent:
$(T_n^*)^m$ acts as a dyadic
averaging operator over a diagonal band of
width $2^{-m}$, ultimately concentrating
mass on the diagonal as $m\to\infty$.
This concentration property allows us to deduce that if
$f$ is supported away from the diagonal,
then $(T_n^*)^m f = 0$ for $m$ large enough.
Part \emph{(ii)} then follows by iterating the
chaos-level eigenvalue relation.

\begin{lemma}\label{thm:closed_form_Tm}
    Let $n,m\ge1$ and $f\in L^2([0,1]^n)$. Then,
    for a.e.  $\mathbf s\in[0,1]^n$
    \begin{equation}\label{eq:closed_form}
    (T_n^*)^m f(\mathbf s)
    =
    2^{-mn/2}\sum_{r=0}^{2^m-1}
    f\!\left(\frac{\mathbf s+r\mathbf 1}{2^m}\right),
    \end{equation}
    where $\mathbf{1}=(1,\dots,1)$.
\end{lemma}
    
\begin{proof}
    We proceed by induction on $m$. The base case follows directly from 
    the definition:
    \[
    (T_n^* f)(\mathbf s)=2^{-n/2}\Big(f(\mathbf s/2)+f((\mathbf s+\mathbf 1)/2)\Big)
    =
    2^{-n/2}\sum_{r=0}^{1} f\!\left(\frac{\mathbf s+r\mathbf 1}{2}\right),
    \]
    which is exactly \eqref{eq:closed_form} for $m=1$.
    Now assume \eqref{eq:closed_form} holds for some $m\ge 1$. Set
    \[
    g(\mathbf s):=( T_n^*)^m f(\mathbf s)
    =
    2^{-mn/2}\sum_{r=0}^{2^m-1} f\!\left(\frac{\mathbf s+r\mathbf 1}{2^m}\right).
    \]
    Then
    \[
    ( T_n^*)^{m+1}f(\mathbf s)
    =( T_n^* g)(\mathbf s)
    =
    2^{-n/2}\Big(g(\mathbf s/2)+g((\mathbf s+\mathbf 1)/2)\Big).
    \]
    We compute the term $g(\mathbf s/2)$ as
    \[
    g(\mathbf s/2)
    =
    2^{-mn/2}\sum_{r=0}^{2^m-1}
    f\!\left(\frac{\mathbf s/2+r\mathbf 1}{2^m}\right)
    =
    2^{-mn/2}\sum_{r=0}^{2^m-1}
    f\!\left(\frac{\mathbf s+2r\,\mathbf 1}{2^{m+1}}\right).
    \]
    Similarly, the term $g((\mathbf s+\mathbf 1)/2)$ can be written as
    \[
    g((\mathbf s+\mathbf 1)/2)
    =
    2^{-mn/2}\sum_{r=0}^{2^m-1}
    f\!\left(\frac{(\mathbf s+\mathbf 1)/2+r\mathbf 1}{2^m}\right)
    =
    2^{-mn/2}\sum_{r=0}^{2^m-1}
    f\!\left(\frac{\mathbf s+(2r+1)\mathbf 1}{2^{m+1}}\right).
    \]
    Plugging into $( T_n^*)^{m+1}f$ we obtain
    \[
    ( T_n^*)^{m+1}f(\mathbf s)
    =
    2^{-n/2}\cdot 2^{-mn/2}
    \left[
    \sum_{r=0}^{2^m-1} f\!\left(\frac{\mathbf s+2r\,\mathbf 1}{2^{m+1}}\right)
    +
    \sum_{r=0}^{2^m-1} f\!\left(\frac{\mathbf s+(2r+1)\mathbf 1}{2^{m+1}}\right)
    \right].
    \]
    Note that, by reindexing $r'=2r$, the first sum 
    covers all even $r'$ up to $2^{m+1}-2$. Similarly, by reindexing 
    $r'=2r+1$, the second sum covers all odd $r'$ up to $2^{m+1}-1$. 
    As a result, both of them combined cover all integers
    $r'=0,1,\dots,2^{m+1}-1$ exactly once. 
    Moreover, since the prefactor in the expression above is
    $2^{-n/2}\cdot 2^{-mn/2}=2^{-(m+1)n/2}$,
    we get
    \[
        ( T_n^*)^{m+1}f(\mathbf s)
        =
        2^{-(m+1)n/2}\sum_{r'=0}^{2^{m+1}-1}
        f\!\left(\frac{\mathbf s+r'\mathbf 1}{2^{m+1}}\right),
        \]
        which is precisely \eqref{eq:closed_form} with $m+1$.
        This completes the induction and proves the lemma.
\end{proof}

\begin{proposition}\label{thm:diag_supp_f}
    Fix $n\ge2$ and let 
    $$\Delta\coloneqq\{(t_1,\dots,t_n)\in[0,1]^n:\ t_1=\cdots=t_n\}.$$
    Assume $f\in L^2([0,1])^{\otimes n}$ satisfies $\mathrm{supp}(f)\cap \Delta=\varnothing$.
    Then there exists $m\ge1$ such that
    \[
    (T_n^*)^{m} f = 0 \qquad \text{in } L^2([0,1])^{\otimes n}.
    \]
\end{proposition}
    
    \begin{proof}
        Since $\mathrm{supp}(f)$ is closed and $\Delta$ is compact, 
        we have
        \[
            \delta\coloneqq \mathrm{dist}(\mathrm{supp}(f),\Delta)>0.
        \]    
        Choose $m\ge1$ 
        such that $2^{-m}<2\delta/\sqrt{n}$ and 
        define the dyadic diagonal band
        \[
            \Delta_m\coloneqq \bigcup_{r=0}^{2^m-1} Q_{m,r},
            \qquad 
            Q_{m,r}\coloneqq [r2^{-m},(r+1)2^{-m})^n.
        \]  
    Then we have $\mathrm{dist}(\mathrm{supp}(f),\Delta_m)>0$
    and therefore
    $f=0$ a.e. on $\Delta_m$.

    By \cref{thm:closed_form_Tm}, for a.e.\ 
    $\mathbf s\in[0,1]^n$,
    \[
        (T_n^*)^m f(\mathbf s)
        = 2^{-mn/2}\sum_{r=0}^{2^m-1}
        f\!\left(\frac{\mathbf s+r\mathbf 1}{2^m}\right).
    \]
    For each $r$ and each $\mathbf s\in[0,1]^n$, 
    the point $(\mathbf s+r\mathbf 1)/2^m$ belongs to 
    $Q_{m,r}\subset\Delta_m$. Since $f=0$ a.e.\ on 
    $\Delta_m$, every summand vanishes for a.e.\ 
    $\mathbf s$, and hence $(T_n^*)^m f=0$ in 
    $L^2([0,1])^{\otimes n}$.
\end{proof}
  
We recall the definition of support of a deterministic distribution.

\begin{definition}[Support of distributions]\label{def:support_Sminus}
    Let $n\ge 1$ and
    $u\in (\S^{\hat\otimes n})^\prime$.
    We say that $u$ \emph{vanishes on an open set} $U\subset[0,1]^n$ if
    \[
    (u,f)=0 \quad\text{for every } f\in \mathcal S^{\hat\otimes n}\ \text{with }\mathrm{supp}(f)\subset U.
    \]
    The \emph{support} of $u$, denoted $\mathrm{supp}(u)$, is the complement of the largest open set on which
    $u$ vanishes.
\end{definition}

With the above definition in place, we 
are now ready to give the proof 
of \cref{thm:intro-spec} \emph{(ii)}.

\begin{theorem}[Diagonal support of eigenkernels]\label{thm:diag_support_eigenkernel}
    Let $U=\sum_{n=0}^\infty  I_n(u_n) \in (\S)^*$ 
    be a generalized eigenvector of $\L$ with eigenvalue $\lambda \neq0$.
    Then, for $n\geq 2$,
    \[
    \mathrm{supp}(u_n)\subset \Delta .
    \]
\end{theorem}
\begin{proof}
    The eigenvector condition means
    $$
    \pp{U}{\L^*\Phi}=\lambda \pp{U}{\Phi}.
    $$    
    By using the generalized chaos expansion 
    (\cref{thm:generalized chaos}) and 
    choosing localized test functionals $\Phi=I_n(f)$, it follows that 
    \begin{equation}\label{p3}
        (u_n,T_n^*f)=\lambda (u_n,f),
    \end{equation}
    for every $n\geq0$ and $f\in \S^{\hat\otimes n}$ arbitrary. 
    
    Now assume $n\ge2$ and let $f$ be such that $\mathrm{supp}(f)\cap \Delta=\varnothing$.
    Then we can apply \cref{thm:diag_supp_f}
    to find $m\geq1$ such that 
    \[
    (T_n^*)^{m} f = 0.
    \]
    By iterating \eqref{p3} $m$ times we conclude 
    \begin{equation}
        0=(u_n,(T_n^*)^mf)=\lambda^m (u_n,f).
    \end{equation}
    Since $\lambda \neq 0$, it follows that  $(u_n,f)=0$ 
    and the claim follows.
\end{proof}

\begin{remark}
In the usual setting of tempered distributions, the diagonal 
concentration result in \cref{thm:diag_support_eigenkernel} would have further implications. 
By the structure theorem for distributions supported on a 
submanifold \cite[Theorem~2.3.5]{hormander2003analysis}, 
any Schwartz 
distribution on $[0,1]^n$ with support contained in the diagonal 
$\Delta$ must be a finite linear combination of transverse 
derivatives of Dirac masses along $\Delta$. 
In our setting, this would imply that every eigenvector kernel 
$u_n$ possesses a highly constrained structure, essentially 
composed of derivatives of delta functions, with only one degree 
of freedom left. 
Observe that the family constructed in 
\cref{def:eigenvector} has 
this form, with the remaining degree of freedom realized by an 
integral along the diagonal. 
We do not attempt to verify such a structural 
constraint in the particular Gelfand triple considered 
in this work, and leave it as a conjecture.
\end{remark}

\section{Wick powers of white noise}
\label{sec:wick-power}

Here we construct the family of Wick-power
eigenvectors $\{\mathfrak{U}_n\}_{n \geq 0}$ and prove
\cref{thm:intro-spec}\,(iii); see 
\cref{def:eigenvector} and \cref{thm:eigenvector} below.
In the preceding section we established general
structural properties of $\L$; we now turn to the
complementary task of constructing explicit eigenvectors
and computing their eigenvalues.
Before proceeding to the proof, we first provide intuition
behind the construction of the eigenvectors
$\mathfrak{U}_n$.
The heuristic argument relies on
applying the $S$-transform to the eigenvalue problem
and using the second quantization formula.
This reduces the problem to a functional equation on
$\S$, whose simplest polynomial solutions turn out to
be precisely the $S$-transforms of the $\mathfrak{U}_n$.

\subsection{Intuition and motivation}
\label{subsec:S-derivation}

To gain further intuition about the structure of the
eigenmodes of $\L$, beyond the structural results of
\cref{sec:gen-spec}, it is convenient to express
the eigenvalue problem in terms of $S$-transforms.
This complementary point of view makes more transparent how
the constraints imposed by the rescaling maps $F$ and $G$
lead to explicit generalized eigenvectors.

Assume that $U \in (\S)^*$ is a generalized
eigenvector of $\L$ with eigenvalue
$\lambda \neq 0$, so that
\begin{equation}\label{eq:eigen-relation-heuristic}
  \pp{\L U}{\Phi}=\pp{U}{\L^*\Phi} = \lambda\,\pp{U}{\Phi}
\end{equation}
for all $\Phi \in (\S)$.
Let $h \in \S$ and $\mathcal{E}(h)$ be the Wick exponential 
as in \cref{def:wick-exp}.
The chaos kernels
of $\mathcal{E}(h)$ are
$f_n = h^{\otimes n}/n!$ (\cref{prop:wick-chaos}) and 
the second quantization formula (\cref{thm:sec quanti}) gives 
\begin{equation}\label{eq:Lstar-wick}
    \L^*\mathcal{E}(h)
    = \sum_{n=0}^{\infty} \frac{1}{n!}
    \Big[
    I_n\!\big((F^*h)^{\otimes n}\big)
    +
    I_n\!\big((G^*h)^{\otimes n}\big)
    \Big]
    = \mathcal{E}(F^*\!h) + \mathcal{E}(G^*\!h).
\end{equation}
Using $\Phi=\mathcal E(h)$ in \eqref{eq:eigen-relation-heuristic} and 
applying \eqref{eq:Lstar-wick}, 
the eigenvalue equation takes the form
\begin{equation}\label{eq:S-eigenvalue}
    SU(F^*\!h) + SU(G^*\!h)
    = \lambda\, SU(h),
    \qquad h \in \S.
\end{equation}
This is a functional
equation for $SU$: we seek a scalar $\lambda$ and
a functional $SU \colon \S \to \mathbb{R}$ such
that the relation above holds for every $h \in \S$.
A particular family of solutions to
\eqref{eq:S-eigenvalue} is given by the monomials
\begin{equation}\label{eq:monomial-ansatz}
    S \frak U_n(h)
    \coloneqq \int_0^1 h(t)^n\,dt,
    \qquad n \geq 0,
\end{equation}
with corresponding eigenvalues $\lambda_n = 2^{1-n/2}$.

It remains to identify the functionals \(S \frak U_n\) as the
\(S\)-transforms of Hida distributions. Let \(\dot W\) denote
white noise and write its pointwise Wick powers as
\[
    {:}\dot W(t)^n{:}
    =
    I_n(\delta_t^{\hat\otimes n}).
\]
A standard identity in white noise analysis gives (see \cite[p.~40]{kuo2018white})
\[
    S\bigl({:}\dot W(t)^n{:}\bigr)(h)
    =
    h(t)^n .
\]
Consequently, \eqref{eq:monomial-ansatz} is precisely the
\(S\)-transform of the Hida distribution
\begin{equation}\label{eq:un-sec6}
    \mathfrak U_n
    =
    \int_0^1 {:}\dot W(t)^n{:}\,dt,
\end{equation}
namely, the time-integrated \(n\)-th Wick power of white noise 
introduced in \eqref{eq:un-intro}.
The associated chaos decomposition is given by
\begin{equation}\label{eq:un-chaos}
  \mathfrak U_n = I_n(\mathfrak u_n),
  \qquad
  \mathfrak u_n =
  \int_0^1 \delta_t^{\hat\otimes n}\,dt .
\end{equation}

We conclude this heuristic subsection 
with a side remark concerning Malliavin calculus. Although
our construction of \(\mathfrak U_n\) is carried out in the Hida
distribution space, its action admits a formal representation in terms
of Malliavin derivatives. More precisely, for sufficiently regular test
functionals \(\Phi\), one has
\begin{equation}\label{eq:Un-Malliavin}
    \pp{\mathfrak U_n}{\Phi}
    =
    \mathbb E_{\mu_\infty}
    \left[
        \int_0^1 D_t^n \Phi\,dt
    \right],
\end{equation}
where \(D_t^n\Phi\) denotes the \(n\)-th Malliavin derivative of
\(\Phi\) 
evaluated at the diagonal point \((t,\ldots,t)\).
This identity provides a natural point of contact between Hida white
noise analysis and Malliavin calculus. It suggests that tools from
Malliavin calculus, such as integration-by-parts formulae, could be
useful for giving alternative interpretations of the eigenvectors
\(\mathfrak U_n\) and their role in the linearized RG dynamics.
We do not, however, pursue this direction here; the analysis below
relies only on Hida white noise analysis and distribution theory.

\subsection{Proof of \cref{thm:intro-spec} \textnormal{\emph{(iii)}}}
\label{sec:proofs6.2}

We now make the heuristic construction of 
\cref{subsec:S-derivation} rigorous.
We begin by showing that the delta function $\delta_t$
belongs to the deterministic distribution 
space constructed in \cref{sec:deterministic},
which allows us to realize the kernel 
$\frak u_n$ in \eqref{eq:un-chaos}
as a Bochner integral of tensor powers 
of $\delta_t$ in the corresponding 
tensorized distribution space.

\begin{proposition}\label{thm:delta}
    Let $p>1/2$ and
    $\delta_t$ be the point-evaluation map defined by $(\delta_t,f)=f(t)$ 
    for $f \in \S_p$.
    Then, $\delta_t\in \S_{-p}$ and
    \[
    \sup_{t\in[0,1]}|\delta_t|_{-p}\leq C_p,
    \]
    where $C_p$ is the constant appearing in 
    \cref{prop:Sp-Linfty}.
\end{proposition}
\begin{proof}
    By \cref{prop:Sp-Linfty}, 
    for every $f\in\S_p$ 
    and every $t\in[0,1]$, we have
    $$|(\delta_t,f)|=|f(t)|
    \leq C_p\,|f|_p.$$
    Consequently,
    \[
        |\delta_t|_{-p} 
        = \sup_{f\in\S_p\setminus\{0\}} 
        \frac{|(\delta_t,f)|}{|f|_p} 
        \leq C_p,
    \]
for every $t\in[0,1]$. 
Taking the supremum over $t$ yields the result.
\end{proof}

\begin{proposition}
    Fix $n\ge 1$ and $p>1/2$.  
    Then, $$\frak u_n \coloneqq\int_0^1 \delta_t^{\hat\otimes n}\,dt\,\,\in\,\, \S_{-p}^{\hat\otimes n} ,$$
    where the integral is understood as a Bochner integral 
    of the $\S_{-p}^{\hat\otimes n}$-valued map
    $t\mapsto \delta_t^{\hat\otimes n}$.
\end{proposition}
\begin{remark}
  By an argument similar to the proof of \cref{prop:Sp-Linfty},
  the integrand $\delta_t^{\hat{\otimes} n}$ is approximated
  uniformly by the simple functions
  \begin{equation*}
    \Delta_M
      := \sum_{r=0}^{2^M-1}
         \mathbf{1}_{[r2^{-M},\,(r+1)2^{-M})}\,
         \bigl(\delta_{r2^{-M}}\bigr)^{\hat{\otimes} n}
       + \mathbf{1}_{\{1\}}\,\delta_1^{\hat{\otimes} n},
    \qquad M \geq 1 .
  \end{equation*}
  As a result, the map $t \mapsto \delta_t^{\hat{\otimes} n}$ is
  strongly measurable and the Bochner integral above is well defined;
  see \cite[Chapter 1]{hytonen2016analysis}.
\end{remark}

\begin{proof}    
    Using the norm inequality for Bochner integrals, 
    multiplicativity of the Hilbert tensor norm, and \cref{thm:delta} yields
    \begin{equation*}
        |\frak u_n|_{\S_{-p}^{\otimes n}}
        \le \int_0^1 |\delta_t^{\hat\otimes n}|_{\S_{-p}^{\otimes n}}\,dt
        = \int_0^1 |\delta_t|_{-p}^{\,n}\,dt
        \le \Big(\sup_{t\in[0,1]}|\delta_t|_{-p}\Big)^{n}
        <\infty.
    \end{equation*}
\end{proof}

Now that we have the kernel $\mathfrak{u}_n$ in 
the correct distribution space, 
we may define the corresponding 
Hida distribution via the generalized 
chaos expansion, 
justifying \eqref{eq:un-chaos}.

\begin{definition}\label{def:eigenvector}
    For $n\geq0$, we define the time-integrated $n$-th 
    Wick power of white noise as the Hida distribution
    $$\frak U_n = I_n(\frak u_n),$$
    where we set $\frak u_0 \coloneq 1\in \R$.
\end{definition}

\begin{remark}
    The definition above provides the 
    rigorous interpretation of the formal 
    expression appearing in \eqref{eq:un-intro} and \eqref{eq:un-sec6}.
    Indeed, since 
    ${:}\dot W(t)^n{:}=I_n(\delta_t^{\hat\otimes n})$
    and $I_n$ is a continuous linear map 
    on $\S_{-p}^{\hat\otimes n}$, 
    the Bochner integral commutes with $I_n$,
    giving
    \[
        \frak U_n 
        = I_n\!\left(\int_0^1 
        \delta_t^{\hat\otimes n}\,dt\right)
        = \int_0^1 
        I_n(\delta_t^{\hat\otimes n})\,dt
        = \int_0^1 
        {:}\dot W(t)^n{:}\,dt.
    \]
\end{remark}

\begin{remark}\label{rmk:Un}
Observe that \cref{def:eigenvector} indeed defines a Hida distribution 
by the generalized chaos expansion (\cref{thm:generalized chaos}).
Moreover, for 
a test functional $\Phi=\sum_{m=0}^\infty I_m(f_m)\in(\S)$,
the distribution $\frak U_n$ acts as 
\begin{equation}\label{eq:Un_action}
    \pp{\frak U_n}{\Phi} = n!(\frak u_n,f_n)
    =n!\int_0^1 (\delta_t^{\hat \otimes n},f_n)\,dt,
\end{equation}
where the last equality follows from
the commutation between Bochner integral and bounded linear maps, namely 
$\Bigl(\int_0^1 \delta_t^{\hat\otimes n}\,dt,f_n \Bigr)=\int_0^1 (\delta_t^{\hat\otimes n},f_n)\,dt$.
For $n=0$, we have $\frak u_0=1$ and hence 
$$\pp{\frak U_0}{\Phi} = (1,f_0)=\E_{\mu_\infty}\![\Phi].$$
Thus $\frak U_0$ acts as integration against the Wiener
measure; in other words, under the canonical embedding of measures into
$(\S)^*$, $\frak U_0$ corresponds to $\mu_{\infty}$
and we may write $\frak U_0=\mu_\infty$ in $(\S)^*$. 
\end{remark}

To prove part \emph{(iii)} of \cref{thm:intro-spec}, 
we need first to understand how the adjoints $F^*$ and $G^*$ 
interact with the delta function.

\begin{lemma}\label{thm:delta_adjoints}
    Let $p>1/2$ and $n\ge1$.
    For every $f\in \S_p^{\hat\otimes n}$ we have
    \begin{equation*}
        \begin{split}
            (\delta_t^{\hat\otimes n},(F^*)^{\otimes n}f)
            &=2^{-n/2}(\delta_{t/2}^{\hat\otimes n},f),
            \qquad\\[0.2cm]
            (\delta_t^{\hat\otimes n},(G^*)^{\otimes n}f)
            &=2^{-n/2}(\delta_{(t+1)/2}^{\hat\otimes n},f).
        \end{split}
    \end{equation*}
    for a.e. $t\in[0,1]$.
\end{lemma}
\begin{proof}
    From the definition of $F^*$ in \eqref{eq:F*_G*}
    it follows that for every $f\in \S_p$ and a.e. $t\in [0,1]$, we have
    \begin{equation*}
            (\delta_t,F^*f) = (F^*f)(t)
            = \frac{1}{\sqrt 2}f(t/2) 
            =\frac{1}{\sqrt 2}\,(\delta_{t/2},f).
    \end{equation*}
    Similarly, by the definition of $G^*$ in \eqref{eq:F*_G*}
    we obtain 
    $$
    (\delta_t,G^*f) =\frac{1}{\sqrt 2}\,(\delta_{(t+1)/2},f).
    $$
    The tensorized identities follow by multiplicativity of
    the tensor pairing on elementary tensors and density.
\end{proof}

We now state and prove \cref{thm:intro-spec} \emph{(iii)}.

\begin{theorem}\label{thm:eigenvector}
    Let $n\geq0$. Then, 
    the Hida distributions $\frak U_n$ are generalized eigenvectors of 
    $\L$ with eigenvalues 
    $$
    \lambda_n =2^{1-n/2}.
    $$
\end{theorem}
\begin{proof}
    From the expression for $\L$ in \eqref{eq:definition_L} along with 
    the fixed point relation \eqref{eq:fixedpoint},
    we have $\L\mu_{\infty}=2\mu_\infty$.
    This proves the case $n=0$ since $\frak U_0=\mu_\infty$, see \cref{rmk:Un}.
    
    Now let $n\geq1$ and recall 
    that the action of $\frak{U}_n$ is given by \eqref{eq:Un_action}.
    We must show the equality $\pp{\frak U_n}{\L^*\Phi}=\lambda_n\pp{\frak U_n}{\Phi}$ for any $\Phi\in(\S)$.
    Using \cref{thm:sec quanti} and \cref{thm:delta_adjoints}, 
    the left hand side can be written as 
    \begin{equation*}
        \begin{split}
            \pp{\frak U_n}{\L^*\Phi} &= n!\int_0^1 (\delta^{\hat\otimes n}_t,  (F^*)^{\otimes n} f_n)\, dt + n! \int_0^1 (\delta^{\hat\otimes n}_t,  (G^*)^{\otimes n} f_n)\, dt\\
            &= 2^{-n/2}n!\int_0^1 (\delta^{\hat\otimes n}_{t/2}, f_n)\, dt 
            + 
            2^{-n/2}n!\int_0^1 (\delta^{\hat\otimes n}_{(t+1)/2}, f_n)\, dt\\
            &= 2^{1-n/2} n! 
            \left(\int_0^{1/2} (\delta^{\hat\otimes n}_{s}, f_n)\, ds + \int_{1/2}^{1} (\delta^{\hat\otimes n}_{s}, f_n)\, ds\right)\\
            &= 2^{1-n/2} n!\int_0^{1} (\delta^{\hat\otimes n}_{s}, f_n)\, ds\\
            &= \lambda_n \pp{\frak U_n}{\Phi},
        \end{split}
    \end{equation*}
    completing the proof.
\end{proof}

\section{Cumulants and Wick perturbations}\label{sec:cumulant-interpretation}

The eigenmodes constructed above have a clear 
probabilistic interpretation.
Here we show that a perturbation of 
the Wiener measure in the direction 
of $\mathfrak{U}_n$ modifies precisely the
$n$-th cumulant, leaving all others
unchanged at leading order.
This perspective motivates the 
cumulant analysis of random walks 
in \cref{sec:donsker} and 
gives the intuition underlying 
\cref{thm:intro-universality}.
We start by recalling the notion 
of cumulants for probability measures 
on Banach spaces and illustrate it in the
Gaussian case. 
For background on Gaussian measures 
on Banach spaces, we refer the reader 
to the classical monograph 
of Bogachev~\cite{bogachev1998gaussian}.

\begin{definition}[Cumulants]
    Let $B$ be a separable real Banach space and $\mu$ a 
    Radon probability measure on $B$.
    Define the \emph{cumulant generating functional}
    $$
      K(\ell) \coloneqq \log \E_\mu(e^{\ell(\cdot)}), \qquad \ell \in B^*.
    $$
    For $n \ge 1$, the $n$-th \emph{cumulant}
    of $\mu$ is the symmetric $n$-linear
    form $\kappa_n : (B^*)^{n} \longrightarrow \mathbb{R}$, given by
    $$
      \kappa_n(\ell_1,\dots,\ell_n)
      \coloneqq
      \left.
      \frac{\partial^n}{\partial t_1 \cdots \partial t_n}
      \right|_{t_1=\cdots=t_n=0}
      K\Big(\sum_{j=1}^n t_j \ell_j\Big),
    $$
    whenever the expression above is finite. 
\end{definition}

\begin{example}[Gaussian measure]
    Let $\mu$ be a Radon Gaussian probability measure on $B$ with mean
    $m:B^* \to \R$ and covariance $C:B^* \times B^* \to \mathbb{R}$.
    The cumulant generating functional is
    $$
      K(\ell)
      \;=\; \log \E_\mu(e^{\ell(\cdot)})
      \;=\; m(\ell) + \tfrac12\,C(\ell,\ell).
    $$
    Then, for $t_1,\dots,t_n \in \mathbb{R}$ and 
    $\ell_1,\dots,\ell_n \in B^*$, we have
    $$
    K\Big(\sum_{j=1}^n t_j \ell_j\Big)=\sum_{i=1}^{n}t_i m(\ell_i)
    +\frac{1}{2}\sum_{i,j=1}^{n}t_it_jC(\ell_i,\ell_j).
    $$
    A direct calculation yields
    $$
      \kappa_1(\ell_i)=
      \left.
      \frac{\partial}{\partial t_i}
      \right|_{t_1=\cdots=t_n=0}
      K\Big(\sum_{j=1}^n t_j \ell_j\Big)= m(\ell_i),
    $$
    and
    \[
      \kappa_2(\ell_i,\ell_j)
      = \left.\frac{\partial^2}{\partial t_i \partial t_j}\right|_{t_1=\cdots=t_n=0}
        K\Big(\sum_{j=1}^n t_j \ell_j\Big)
      = C(\ell_i,\ell_j).
    \]
    Since $K$ is a polynomial of degree at most $2$ in $t_1,\dots,t_n$,
    all higher derivatives at $0$ vanish, hence
    $$
      \kappa_n(\ell_1,\dots,\ell_n) = 0
      \qquad \text{for all } n \ge 3.
    $$
    Thus, for a Gaussian measure, the only nonzero cumulants are the first and
    second, given by the mean $m$ and the covariance $C$.
\end{example}
    
In the Wiener-space setting, it is natural to 
extend the defintion of the cumulant generating functional to 
measurable linear functionals that are not necessarily continuous.
Recall that the isonormal Gaussian process
$\mathcal W(h)$, $h\in L^2([0,1])$, provides a family of measurable
linear functionals on Wiener space.  These extend the Gaussian
linear observables arising from $\Omega^*$ and are the natural
coordinates for the Wiener chaos and Hida calculus developed above.
Accordingly, in the following heuristic discussion
we work with the cumulant generating functional
$$
    K(h)
    \coloneqq
    \log \mathbb E_{\mu}[e^{\mathcal W(h)}],
    \qquad h\in\S.
$$

Consider a probability measure whose first-order perturbation is, formally, of the form
$$\mu_\eps\approx\mu_\infty + \eps \frak U_n,$$ 
with $\eps>0$ small. Pairing with the functional $\Phi=e^{\mathcal W(h)}$, $h\in \S$, leads to 
\begin{equation*}
    \begin{split}
        \E_{\mu_\eps}(e^{\mathcal W(h)})&= \E_{\mu_\infty}(e^{\mathcal W(h)}) + \eps \pp{\frak U_n}{e^{\mathcal W(h)}}\\
        &=e^{\frac{1}{2}|h|^2_{L^2}} +
        \eps e^{\frac{1}{2}|h|^2_{L^2}}  \int_{0}^{1}h(t)^n\,dt 
    \end{split}
\end{equation*}
Taking the logarithm we obtain the cumulant generating function 
\begin{equation*}
    \begin{split}
        K^{(\eps)}(h) &= \frac{1}{2}|h|^2_{L^2} + \log\Bigl(1+\eps\int_{0}^{1}h(t)^n\,dt \Bigr)\\
        &=\frac{1}{2}|h|^2_{L^2} +\eps\int_{0}^{1}h(t)^n\,dt + O(\eps^2).
    \end{split}
\end{equation*}
Substituting $h=\sum_{i=1}^{d}t_i h_i$ in the expression above we find 
$$
K^{(\eps)}\Big(\sum_{i=1}^{d}t_i h_i\Big)=
\frac{1}{2}\sum_{i,j=1}^{d}t_it_j\ps{h_i,h_j}_{L^2} 
+
\eps\sum_{i_1,\dots i_n=1}^{d}  t_{i_1} \cdots t_{i_n} \int_{0}^{1} h_{i_1}(t) \cdots h_{i_n}(t)\,dt + O(\eps^2).
$$

For the sake of simplicity, assume $n\geq 3$. From the expression above, we read off the cumulants as
follows. Since $K^{(\varepsilon)}$ has no linear term in the variables
$t_1,\dots,t_d$, the first cumulant vanishes at first order:
\[
  \kappa^{(\varepsilon)}_1(h_i) = O(\varepsilon^2), \quad i=1,\dots,d.
\]
The quadratic part of $K^{(\varepsilon)}$ coincides with that of the Wiener 
measure $\mu_\infty$, hence
\[
  \kappa^{(\varepsilon)}_2(h_i,h_j)
  = \langle h_i,h_j\rangle_{L^2} + O(\varepsilon^2),
\]
so the covariance is unchanged to first order in $\varepsilon$.
The first-order correction appears only in the cumulant of order $n$:
\[
  \kappa_n^{(\varepsilon)}(h_1,\dots,h_n)
  = \varepsilon\,n!
    \int_0^1 h_1(t)\cdots h_n(t)\,dt
    + O(\varepsilon^2),
\]
while all other higher-order cumulants remain zero at first order:
\[
  \kappa_m^{(\varepsilon)}(h_1,\dots,h_m)
  = O(\varepsilon^2),
  \qquad m \ge 3,\ m \neq n.
\]
In other words, for $n\geq3$, perturbing by $\frak U_n$ 
produces, to leading order, a pure $n$-th-order cumulant 
in the otherwise Gaussian process: the mean and covariance 
are unchanged at order $\varepsilon$, and a nontrivial 
$n$-th cumulant is created. This yields a genuinely 
non-Gaussian perturbation. By contrast, the cases $n=1,2$ 
correspond to a first-order modification of the mean and 
covariance, respectively, without introducing any 
higher-order cumulants. In these cases, $\mu_\eps$ is a 
Gaussian perturbation of the underlying measure $\mu_\infty$.

\section{Donsker's invariance principle and universality of corrections}\label{sec:donsker}

We prove 
\cref{thm:intro-universality}\,(i) and (ii);
the precise statements can be found in \cref{thm:cumulant-expansion} and 
\cref{thm:leading-correction} below.
The classical Donsker
invariance principle asserts that appropriately rescaled random
walks converge weakly to Brownian motion. The results in this
section quantify this convergence by showing that the rate
of convergence and the structure of correction terms are governed
by the eigenmodes $\mathfrak{U}_n$ of $\L$, 
thereby revealing a second layer of universality 
in the invariance principle.

\subsection{Donsker approximations}
Let $(\Omega_0,\mathcal F_0,\mathbb P)$ be a 
probability space carrying an i.i.d.\ sequence 
$\{\zeta_j\}_{j\geq 1}$ of real-valued random 
variables with $\mathbb{E}[\zeta_1] = 0$ 
and $\mathbb{E}[\zeta_1^2] = 1$.
We write $\zeta$ for a random variable with 
the common distribution.
Fix an integer $N \geq 1$ and define the 
partial sums
\begin{equation}\label{eq:partial_sums}
S_k^{(N)} := 2^{-N/2} \sum_{j=1}^k \zeta_j, \quad
k = 0, 1, \ldots, 2^N,
\end{equation}
where by convention $S_0^{(N)} = 0$. The Donsker approximation of
Brownian motion at resolution $N$ is the continuous piecewise-linear
path $W_N$ defined on the unit interval by
\begin{equation}\label{eq:donsker_path}
W_N(t) := S_k^{(N)} + 2^{N/2}(t - k/2^N) \zeta_{k+1},
\end{equation}
for $t \in [k/2^N, (k+1)/2^N]$, $k = 0, \ldots, 2^N - 1$.
In other words, $W_N$ is obtained by linearly interpolating the
values $S_k^{(N)}$ at the dyadic grid points $k/2^N$. By
construction, $W_N$ is continuous, $W_N(0) = 0$, and
$W_N \in \Omega$ almost surely.
Thus $W_N:\Omega_0\to \Omega$ is a random path, and we denote 
its law on $\Omega$ by
\[
\widetilde{\mu}_N := (W_N)_\#\mathbb P,
\]
where $(W_N)_\#$ denotes the pushforward.
By Donsker's invariance principle,
$\widetilde{\mu}_N$ converges weakly to $\mu_\infty$ as
$N \to \infty$.

\subsection{Donsker approximations as Hida distributions}
The goal of this section is to embed the Donsker measures 
$\widetilde{\mu}_N$ into the Hida distribution space $(\S)^*$. 
We first realize $\widetilde{\mu}_N$ as
a linear functional 
$\mu_N^{\mathrm{core}}(\Phi)$ 
for $\Phi$ in a core $\mathcal{D}$ of $(\S)$. On this core, 
we obtain a tentative expression for the $S$-transform of 
$\mu_N^{\mathrm{core}}$. We then show that the expression 
satisfies the analyticity and growth conditions and use 
the characterization theorem (\cref{thm:characterization}) to define a Hida distribution 
$\mu_N^{\mathrm{ext}}$, now defined on the whole stochastic test
space $(\S)$. \emph{A posteriori}, we show that 
$\mu_N^{\mathrm{ext}}$ coincides with $\mu_N^{\mathrm{core}}$ 
on the core $\mathcal D$. By uniqueness of the continuous extension, 
we conclude that $\mu_N^{\mathrm{ext}}$ is the unique continuous 
extension of $\mu_N^{\mathrm{core}}$. 
This allows us to realize the probability laws $\widetilde{\mu}_N$
as Hida distributions, giving access to the toolkit of
white noise analysis, which will be instrumental in
what follows.

Let $(W_t)_{t\in[0,1]}$ denote the coordinate process on
$\Omega$.
For $h\in L^2([0,1])$, we write
\[
\mathcal W(h)=\int_0^1 h(t)\,dW_t,
\]
for the integral of $h$ against the coordinate process.
Under the Wiener measure $\mu_\infty$, this is the usual
Wiener integral with respect to Brownian motion.
Under the measure $\widetilde{\mu}_N$, 
the coordinate process is piecewise affine and 
absolutely continuous with a piecewise constant, square-integrable derivative. 
Specifically, $\widetilde{\mu}_N$ is supported on the space
of continuous functions that are affine on each dyadic interval
\[
I_{N,k}\coloneq\bigl[k2^{-N},(k+1)2^{-N}\bigr),\qquad k=0,\dots,2^N-1.
\]
Hence, for $\widetilde{\mu}_N$-almost every $\omega\in\Omega$, the random
variable $\mathcal W(h)$ admits the pathwise representation
\[
\mathcal W(h)(\omega)=\int_0^1 h(t)\,d\omega(t),
\]
where the right-hand side is understood as a Lebesgue--Stieltjes 
integral with respect to an absolutely continuous (hence of bounded variation) path. 
By introducing
\begin{equation}\label{eq:step}
    \phi_{N,k}:=2^{N/2}\mathbf{1}_{I_{N,k}},\qquad k=0,\dots,2^N-1,
\end{equation}
and using the piecewise representation \eqref{eq:donsker_path},
we have, $\P$-a.s.,
\[
\dot W_N(t)=\sum_{k=0}^{2^N-1}\zeta_{k+1}\phi_{N,k}(t)
\qquad\text{for a.e. }t.
\]
Note that $\phi_{0,0}=\mathbf 1$ in $L^2([0,1])$.
Consequently, we obtain under $\widetilde{\mu}_N$
\begin{equation}\label{eq:donsker_w}
\mathcal W(h)
\;\law\;
\sum_{k=0}^{2^N-1}\langle h,\phi_{N,k}\rangle_{L^2}\,\zeta_{k+1}.
\end{equation}

For $h \in \S$, define the Wick exponential
\[
  \mathcal{E}(h) := \exp\!\left(\mathcal{W}(h) 
    - \tfrac{1}{2}|h|_{L^2}^2\right).
\]
By \cref{prop:wick-chaos}, $\mathcal{E}(h) \in (\S)$, and as
in \cref{sec:S-transform} we set
\[
  \mathcal{D} = \operatorname{span}\{\mathcal{E}(h) : h \in \S\}
  \subset (\S),
\]
which is dense in $(\S)$ by \cref{cor:exp-dense}.
In order to compute expectations of the 
random variables $\mathcal{E}(h)$ under $\widetilde{\mu}_N$ we make use of 
\eqref{eq:donsker_w} to evaluate
\begin{equation}\label{eq:S-tentative}
        \mathbb{E}_{\widetilde{\mu}_N}[\mathcal{E}(h)]
        =
        e^{-\tfrac12|h|_{L^2}^2}
        \prod_{k=0}^{2^N-1}
        \mathbb{E}\!\left[e^{\zeta\langle h,\phi_{N,k}\rangle_{L^2} }\right],
\end{equation}
which is finite under suitable assumptions on $\zeta$, e.g., 
the sub-Gaussian condition \eqref{eq:subgaussian} below.
This suggests defining a linear functional
$\mu_N^{\mathrm{core}}$ on $\mathcal{D}$ by setting
\[
  \mu_N^{\mathrm{core}}(\mathcal{E}(h))
  \coloneq \mathbb{E}_{\widetilde{\mu}_N}[\mathcal{E}(h)],
  \qquad h \in \S,
\]
and extending to the span by linearity.  Since the family of exponential
vectors is linearly independent, this defines
$\mu_N^{\mathrm{core}}$ unambiguously on~$\mathcal{D}$.

Intuitively, one expects \eqref{eq:S-tentative} 
to represent the $S$-transform
of the Hida distribution $\widetilde{\mu}_N$, although at this stage 
we have not yet established that 
$\widetilde{\mu}_N$ belongs to $(\S)^*$. In principle 
one can arrive at this conclusion by two different routes.
The first is to show that $\mu_N^{\mathrm{core}}$ 
is continuous on $\mathcal D$ with respect to the $(\S)$-topology.
\cref{cor:exp-dense} would then imply that 
$\mu_N^{\mathrm{core}}$ can be extended continuously to all 
of $(\S)$.
A second approach is to show that the right-hand side 
of \eqref{eq:S-tentative} satisfies the analyticity and growth 
conditions required by the characterization theorem,
ensuring the
existence of a unique Hida distribution 
$\mu_N^{\mathrm{ext}}\in(\S)^*$
whose $S$-transform equals the right-hand side of 
\eqref{eq:S-tentative}. 
Since 
$\mu_N^{\mathrm{ext}}$ coincides with $\mu_N^{\mathrm{core}}$ 
on $\mathcal{D}$, we then conclude that 
$\mu_N^{\mathrm{ext}}$ is the unique continuous extension
of $\mu_N^{\mathrm{core}}$, yielding an
embedding of the probability laws $\widetilde{\mu}_N$ into $(\S)^*$.
In what follows, we settle for the second approach 
due to its technical simplicity.

For the next results, recall that 
$\S_{\mathbb{C}} = \S + i\S$ denotes the 
complexification of $\S$, and the $S$-transform of 
a Hida distribution $U\in(\S)^*$ is denoted by $SU$.
See \cref{sec:S-transform} for more details.
We also introduce the bilinear form 
$Q:\S_\mathbb{C}\times \S_\mathbb{C}\to \mathbb C$,
$$
Q(\xi,\eta)\coloneq\int_0^1\xi(t)\,\eta(t)\,dt,
$$
with $Q(\xi)\coloneq Q(\xi,\xi)$. 
Obviously, when $\xi$ and $\eta$ are real-valued, 
we recover the usual $L^2$-inner product and norm.
The reason for introducing $Q$ is that for the condition 
(PS1) below to hold, we must work with the bilinear
extension of the $L^2$-inner product rather than the usual 
sesquilinear one.

\begin{proposition}\label{prop:FN-characterization}
    Assume that $\zeta$ satisfies the sub-Gaussian bound
    \begin{equation}\label{eq:subgaussian}
      \mathbb{E}\left[e^{t\zeta}\right]
      \leq e^{\sigma^2 t^2/2},
      \qquad t \in \mathbb{R},
    \end{equation}
    for some constant $\sigma > 0$.
    Then, for each $N \geq 1$, the function
    $F_N : \S_{\mathbb{C}} \to \mathbb{C}$ defined by
    \begin{equation}\label{eq:FN-def}
      F_N(\xi)
      \coloneq e^{-Q(\xi)/2}\,
        \prod_{k=0}^{2^N-1}
          \mathbb{E}\left[
            e^{\zeta\, Q(\xi, \phi_{N,k})}
          \right],
      \qquad \xi \in \S_{\mathbb{C}},
    \end{equation}
    satisfies conditions \textup{(PS1)} and \textup{(PS2)} of
    \cref{thm:characterization}, namely:
    \begin{enumerate}
        \item[\textup{(PS1)}]
        for every $\xi, \eta \in \S_{\mathbb{C}}$, the function
        $z \mapsto F_N(\xi + z\,\eta)$ is entire on $\mathbb{C}$;
        \item[\textup{(PS2)}]
        for every
        $\xi \in \S_{\mathbb{C}}$,
        \[
          |F_N(\xi)|
          \leq \exp\!\left(
            \frac{1+\sigma^2}{2}\,|\xi|_{L^2}^2
          \right).
        \]
    \end{enumerate}
\end{proposition}

\begin{proof}
    \textit{Verification of \textup{(PS1)}.}
      First note that
      $M(t):=\mathbb{E}\left[e^{t\zeta}\right]$ 
      extends to an entire function on $\mathbb{C}$.
      Indeed, 
      under the sub-Gaussian assumption~\eqref{eq:subgaussian}, we have
    \[
      \mathbb{E}[e^{|t|\,|\zeta|}]
      \leq \mathbb{E}[e^{|t|\zeta} + e^{-|t|\zeta}]
      \leq 2\,e^{\sigma^2 t^2/2},
      \qquad t \in \mathbb{R},
    \]
    where we used the basic 
    inequality $e^{|x|}\leq e^{x} + e^{-x}$ with $x=|t|\zeta$.
    For any $w\in \mathbb{C}$, applying this bound 
    with $t=|w|$ and using monotone convergence,
    \[
      \sum_{n=0}^{\infty} \frac{|w|^n}{n!}\,\mathbb{E}[|\zeta|^n]
      = \mathbb{E}\!\left[\sum_{n=0}^{\infty}
          \frac{(|w|\,|\zeta|)^n}{n!}\right]
      = \mathbb{E}[e^{|w|\,|\zeta|}]
      \leq 2\,e^{\sigma^2|w|^2/2}
      < \infty.
    \]
    Hence, by Fubini,
    \[
      \mathbb{E}[e^{w\zeta}]
      = \mathbb{E}\!\left[
          \sum_{n=0}^{\infty} \frac{(w\zeta)^n}{n!}
        \right]
      = \sum_{n=0}^{\infty} \frac{\mathbb{E}[\zeta^n]}{n!}\,w^n,
    \]
    with absolute convergence for every $w \in \mathbb{C}$.
    Therefore $M(w)=\mathbb{E}[e^{w\zeta}]$ is given by a
    power series with infinite radius of convergence, hence
    entire on~$\mathbb{C}$.
    
    Now fix $\xi, \eta \in \S_{\mathbb{C}}$.
    For each $N$ and $k$, the map
    $z \mapsto Q(\xi + z\,\eta,\, \phi_{N,k})
    = Q(\xi,\phi_{N,k}) + z\,Q(\eta,\phi_{N,k})$
    is a polynomial of degree one in $z$, hence entire.
    Since $M(w)$ is entire and
    the composition of entire functions is entire,
    $z \mapsto M(Q(\xi + z\,\eta,\,\phi_{N,k}))$
    is entire for each $N$ and $k$.
    Similarly, $z \mapsto Q(\xi + z\,\eta)
    = Q(\xi) + 2z\,Q(\xi,\eta) + z^2\,Q(\eta)$
    is a polynomial of degree two in $z$, so the Gaussian
    prefactor $z \mapsto e^{-Q(\xi+z\eta)/2}$ is entire.
    Since finite products of entire functions are entire,
    it follows that $z \mapsto F_N(\xi + z\,\eta)$ is entire 
    for each $N$.
  
  \medskip
  \textit{Verification of \textup{(PS2)}.}
  Let $\xi \in \S_{\mathbb{C}}$.
  For the Gaussian prefactor, we have
  \begin{equation}\label{eq:gauss-prefactor}
    \left|e^{-Q(\xi)/2}\right|
    = e^{-\mathrm{Re}(Q(\xi))/2}
    \leq 
    e^{|Q(\xi)|/2}
    \leq 
    e^{|\xi|_{L^2}^2/2}.
  \end{equation}
  For the product, we use the sub-Gaussian
  bound~\eqref{eq:subgaussian} to obtain, for every $N$ and $k$,
    \[
    \left|\mathbb{E}\!\left[
      e^{\zeta\,Q(\xi,\phi_{N,k})}
    \right]\right|
    \leq \mathbb{E}\!\left[
      e^{\mathrm{Re}(Q(\xi,\phi_{N,k}))\,\zeta}
    \right]
    \leq e^{\sigma^2\,(\mathrm{Re}\,Q(\xi,\phi_{N,k}))^2/2}
    \leq e^{\sigma^2\,|Q(\xi,\phi_{N,k})|^2/2}.
    \]
    Taking the product yields
    \begin{equation}\label{eq:product-bound}
        \prod_{k=0}^{2^N-1}
          \left|\mathbb{E}\!\left[
            e^{\zeta\,Q(\xi,\phi_{N,k})}
          \right]\right|
        \leq \prod_{k=0}^{2^N-1}
          e^{\sigma^2\,|Q(\xi,\phi_{N,k})|^2/2}
        = e^{
          \frac{\sigma^2}{2}
          \sum_{k=0}^{2^N-1}|Q(\xi,\phi_{N,k})|^2}.
      \end{equation}
  Now observe that, since $\phi_{N,k}$ is real, 
  $Q(\xi,\phi_{N,k}) = \langle\xi,\phi_{N,k}\rangle_{L^2}$,
  and thus
  \begin{equation}\label{eq:ps2}
    \sum_{k=0}^{2^N-1}|Q(\xi,\phi_{N,k})|^2=
    \sum_{k=0}^{2^N-1}|\langle\xi,\phi_{N,k}\rangle_{L^2}|^2
    \leq |\xi|_{L^2}^2,
    \end{equation}
  where in the final inequality we used Bessel's inequality
  applied to the orthonormal family
  $\{\phi_{N,k}\}_{k=0}^{2^N-1}$ in $L^2([0,1])$.
  Combining~\eqref{eq:gauss-prefactor},
  \eqref{eq:product-bound}, and \eqref{eq:ps2} yields
  \[
    |F_N(\xi)|
    \leq e^{|\xi|_{L^2}^2/2}\,
       e^{\sigma^2\,|\xi|_{L^2}^2/2}
    = e^{
        \bigl(\frac{1+\sigma^2}{2}\bigr)\,|\xi|_{L^2}^2
    }.
  \]
  This is condition (PS2) of \cref{thm:characterization}
  with $K_1 = 1$, $K_2 = (1+\sigma^2)/2$, and $p = 0$.
\end{proof}
  
\begin{corollary}[Donsker measures as Hida distributions]\label{prop:donsker-hida}
  Under the conditions of \cref{prop:FN-characterization},
  for each $N \geq 1$ there exists a unique Hida distribution
  $\mu_N^{\mathrm{ext}} \in (\S)^*$ whose $S$-transform is
  $F_N$.  Moreover, $\mu_N^{\mathrm{ext}}$ is the unique
  continuous extension of $\mu_N^{\mathrm{core}}$ from
  $\mathcal{D}$ to $(\S)$.
\end{corollary}  
\begin{proof}
  By \cref{prop:FN-characterization}, the function $F_N$
  satisfies conditions (PS1) and (PS2) of
  \cref{thm:characterization}.
  The characterization theorem then yields a unique Hida
  distribution $\mu_N^{\mathrm{ext}} \in (\S)^*$ with
  $S\mu_N^{\mathrm{ext}} = F_N$.
  By definition of the $S$-transform (\cref{def:S-transform}),
  for every $h \in \S$,
  \[
    \pp{\mu_N^{\mathrm{ext}}}{\mathcal{E}(h)}
    = S\mu_N^{\mathrm{ext}}(h)
    = F_N(h)
    = \mu_N^{\mathrm{core}}(\mathcal{E}(h)).
  \]
  By linearity, $\mu_N^{\mathrm{ext}}$ and
  $\mu_N^{\mathrm{core}}$ agree on all of $\mathcal{D}$.
  Since $\mu_N^{\mathrm{ext}}$ is continuous on
  $(\S)$, and $\mathcal{D}$ is dense in $(\S)$
  (\cref{cor:exp-dense}), $\mu_N^{\mathrm{ext}}$ is the unique
  continuous extension of $\mu_N^{\mathrm{core}}$.
\end{proof}

\begin{remark}\label{rem:subgaussian}
    The sub-Gaussian condition~\eqref{eq:subgaussian} is a
    technical assumption used to verify the growth
    bound~(PS2), and is not required for the Donsker invariance
    principle itself, which holds under the weaker conditions
    $\mathbb{E}[\zeta] = 0$ and $\mathbb{E}[\zeta^2] = 1$.
    In practice, \eqref{eq:subgaussian} is not restrictive,
    essentially covering all the standard 
    cases of interest for the Donsker approximation,
    e.g., the Rademacher distribution ($\zeta = \pm 1$ with equal
    probability), 
    the symmetric uniform distribution
    (normalized to unit variance),
    and more generally any
    distribution supported on a bounded interval.
  \end{remark}

\begin{remark}
  From this point onward, we write $\widetilde{\mu}_N$ for the Hida
  distribution $\mu_N^{\mathrm{ext}} \in (\S)^*$ and
  $\pp{\widetilde{\mu}_N}{\Phi}$ for the dual pairing with any test
  functional $\Phi \in (\S)$.  On the generating set of Wick
  exponentials, this pairing recovers the literal integral:
  $\pp{\widetilde{\mu}_N}{\mathcal{E}(h)}
  = \mathbb{E}_{\widetilde{\mu}_N}[\mathcal{E}(h)]$,
  computed via~\eqref{eq:S-tentative}.
\end{remark}
  
\subsection{Proof of \cref{thm:intro-universality} \textnormal{\emph{(i)}}}
\label{sec:cumulant-expansion}

Let $\lambda_n = 2^{1-n/2}$, $n \geq 0$, be the eigenvalues 
associated with the eigenvectors $\frak{U}_n$ of
the linearized RG operator $\L$
(see \cref{def:eigenvector} and \cref{thm:eigenvector}).
Using \cref{prop:S-series}, we obtain the following expression 
for the $S$-transform of $\mathfrak{U}_n$, 
\[
  S\mathfrak{U}_n(\xi)
  = (\mathfrak{u}_n,\, \xi^{\otimes n})
  = \int_0^1 (\delta_t^{\hat\otimes n},\, \xi^{\otimes n})\,dt
  = \int_0^1 \xi(t)^n\,dt,
  \qquad \xi \in \S_\mathbb{C}.
\]
Let
\begin{equation}\label{eq:cumulants-def}
  \kappa_n(\zeta)
  \coloneq \frac{d^n}{dt^n}
    \log \mathbb{E}[e^{t\zeta}]\bigg|_{t=0},
  \qquad n \geq 1,
\end{equation}
denote the cumulants of the increment distribution $\zeta$.
Observe that $\kappa_n(\zeta)$ 
is well-defined and finite under the sub-Gaussian
assumption~\eqref{eq:subgaussian},
with $\kappa_1(\zeta) = 0$ and $\kappa_2(\zeta) = 1$.
The aim of this subsection is to prove the following result.

\begin{theorem}[Cumulant expansion for $\widetilde{\mu}_N$]\label{thm:cumulant-expansion}
  Assume that $\zeta$ satisfies the sub-Gaussian 
  condition~\eqref{eq:subgaussian} and  
  let $\xi\in \mathcal S_{\mathbb C}$.
    Then, for every integer $M\ge 3$,
    \begin{equation}\label{eq:cumulant-expansion}
      \log S\widetilde \mu_N(\xi)
      = \sum_{m=3}^{M}
          \frac{\kappa_m(\zeta)}{m!}\,\lambda_m^N\,
          S\mathfrak{U}_m(\xi)
        + o(\lambda_M^N),
      \qquad N \to \infty,
    \end{equation}
    where $\log$ denotes the principal branch of the logarithm.
\end{theorem}

The expansion~\eqref{eq:cumulant-expansion} reveals a direct
connection between the Donsker approximation and the spectral
structure of the linearized RG operator.  The terms on the
right-hand side are organized by the eigenvalues $\lambda_m^N$
of $\L$, with each contribution weighted by the cumulant
$\kappa_m(\zeta)$ and paired with the $S$-transform of the
corresponding eigenvector $\mathfrak{U}_m$.
In particular, the non-Gaussianity of the increment
distribution $\zeta$, encoded in its cumulants, is
decomposed along the eigendirections of $\L$:
the cumulant $\kappa_m(\zeta)$ controls the amplitude
of the projection onto the $m$-th eigenmode, while the
eigenvalue $\lambda_m^N$ dictates its decay rate as
$N \to \infty$.
Since $|\lambda_m| < 1$ for $m \geq 3$ and the eigenvalues
are strictly decreasing in $m$, the dominant contribution
is generically determined by the first nonvanishing
cumulant $\kappa_m(\zeta)$ with $m \geq 3$.

Before proceeding with the proof of
\cref{thm:cumulant-expansion}, we collect two auxiliary
lemmas.  Recall the dyadic step functions
$\phi_{N,k} = 2^{N/2}\,\mathbf{1}_{I_{N,k}}$,
$k = 0, \ldots, 2^N-1$, introduced in \eqref{eq:step}.
We denote by $P_N$ the orthogonal projection
of $L^2([0,1])$ onto their span, that is,
\begin{equation}\label{eq:PN-def}
  P_N \xi
  = \sum_{k=0}^{2^N-1}
    \langle \xi, \phi_{N,k}\rangle_{L^2}\,\phi_{N,k}.
\end{equation}
The first lemma shows that $P_N$ is a contraction on
$L^\infty$, which is used to control the nonlinear terms
in the cumulant expansion.
The second expresses $P_N$ in terms of the Haar basis,
which provides quantitative estimates on the approximation
error $|P_N \xi - \xi|_{L^2}$ via the Hida norms $|\cdot|_p$.

\begin{lemma}\label{lem:projection_linfty}
    Let $\xi \in \S_{\mathbb{C}}$. Then,
    \[
    |P_N \xi|_{L^\infty} \leq |\xi|_{L^\infty}.
    \]
\end{lemma}
\begin{proof}
    For a.e.\ $t \in [0,1]$, there is a unique 
    $k$ such that $t \in I_{N,k}$, and
    \[
    |P_N \xi(t)| 
    = \left|2^N \int_{I_{N,k}} \xi(s)\,ds\right| 
    \leq 2^N \int_{I_{N,k}} |\xi(s)|\,ds 
    \leq 2^N\,|\xi|_{L^\infty}|I_{N,k}| 
    = |\xi|_{L^\infty}.
    \]
    Taking the essential supremum over $t$ gives 
    the result.
\end{proof}

\begin{lemma}\label{thm:projection}
    Let $\{e_{n,k}\}_{n \geq 0,\, 0 \leq k \leq 2^n - 1}$ 
    be the Haar system and  $\xi \in L^2([0,1])$.
    Then,
    \[
    P_N \xi = \ps{\xi,\phi_{0,0}}_{L^2} \phi_{0,0}
    + \sum_{n=0}^{N-1} \sum_{k=0}^{2^n-1} 
    \ps{\xi, e_{n,k}}_{L^2} \, e_{n,k},
    \quad \text{in } L^2([0,1]).
    \]
\end{lemma}
\begin{proof}
    Let $V_N := \operatorname{span}\{\phi_{N,k}\}_{k=0}^{2^N - 1}$ 
    and define
    \[
    \mathcal{B}_N := \{\phi_{0,0}\} \cup \{e_{n,k} : 0 \leq n \leq N-1,\; 0 \leq k \leq 2^n - 1\}.
    \]
    It follows that $\mathcal{B}_N \subset V_N$. Indeed, for $n \leq N-1$, 
    the wavelet $e_{n,k}$ lives on a coarser grid than the partition defining $V_N$,
    hence it is constant on each dyadic interval $I_{N,k'}$ and may be written 
    as a linear combination of the corresponding indicators $\phi_{N,k'}$.
    Since $\mathcal{B}_N$ is a subset of the Haar system, 
    it is orthonormal, 
    and its cardinality satisfies
    $
    |\mathcal{B}_N| = 1 + \sum_{n=0}^{N-1} 2^n = 2^N = \dim V_N,
    $
    so $\mathcal{B}_N$ is an orthonormal basis for $V_N$
    with $\operatorname{span}(\mathcal{B}_N)=V_N$.

    We now show that the right-hand side of the claimed formula defines 
    the orthogonal projection onto $V_N$. Define $P_N'\xi$ by the right-hand side. 
    Since $P_N'\xi$ is a linear combination of elements of $\mathcal{B}_N \subset V_N$, 
    we have $P_N'\xi \in V_N$. It remains to show $\xi - P_N'\xi \perp V_N$, 
    since uniqueness of orthogonal projections then gives $P_N' = P_N$. 
    Since $\mathcal{B}_N$ spans $V_N$, it suffices to check orthogonality 
    against each basis element. 
    For any $e_{n',k'} \in \mathcal{B}_N$, 
    linearity and orthonormality of the Haar system give
    \[
    \langle P_N'\xi,\, e_{n',k'}\rangle_{L^2} 
    = \langle \xi, \phi_{0,0}\rangle_{L^2}\langle \phi_{0,0}, e_{n',k'}\rangle_{L^2}
    + \sum_{n,k} \langle \xi, e_{n,k}\rangle_{L^2}\, 
    \langle e_{n,k}, e_{n',k'}\rangle_{L^2} 
    = \langle \xi, e_{n',k'}\rangle_{L^2},
    \]
    and similarly $\langle P_N'\xi,\, \phi_{0,0}\rangle_{L^2} = \langle \xi, \phi_{0,0}\rangle_{L^2}$.
    Hence $\langle \xi - P_N'\xi,\, v\rangle_{L^2} = 0$ for all $v \in V_N$,
    and therefore $\xi - P_N'\xi \in V_N^{\perp}$, as required.
\end{proof}

We now turn to the proof of \cref{thm:cumulant-expansion}.

\begin{proof}[Proof of \cref{thm:cumulant-expansion}]
  By the sub-Gaussian assumption, the function 
  $w \mapsto \mathbb{E}[e^{w\zeta}]$ is entire 
  and equals~$1$ at $w = 0$, so its logarithm 
  is holomorphic in a neighborhood of the origin, 
  with the branch fixed by 
  $\log \mathbb{E}[e^{w\zeta}]\big|_{w=0} = 0$;
  see \cite[Theorem~6.2]{stein2003complex}.
  By \cref{prop:Sp-Linfty}, $\xi \in \S_{\mathbb{C}}$ 
  is bounded on $[0,1]$ and we have
  \begin{equation}\label{eq:linfty_bound}
    \bigl|Q(\xi,\phi_{N,k})\bigr|
    = \left|
        2^{N/2}\int_{I_{N,k}} \xi(t)\,dt
      \right|
    \leq 2^{-N/2}\,|\xi|_{L^\infty},
  \end{equation}
  so that $\log \mathbb{E}[e^{Q(\xi,\phi_{N,k})\zeta}]$ 
  is well defined for all $k$ when $N$ is 
  sufficiently large.
  By \cref{prop:donsker-hida}, 
  $S\widetilde{\mu}_N(\xi) = F_N(\xi)$, 
  and exponentiating one verifies that (see \eqref{eq:FN-def})
  \[
    \log S\widetilde{\mu}_N(\xi)
    = -\frac{Q(\xi)}{2}
      + \sum_{k=0}^{2^N-1}
        \log \mathbb{E}\!\left[
          e^{Q(\xi, \phi_{N,k})\,\zeta}
        \right],
  \]
  where $\log S\widetilde{\mu}_N(\xi)$ denotes 
  a branch of the logarithm, identified as 
  the principal branch at the end of the proof.
  
  Expanding $\log \mathbb{E}[e^{w\zeta}]$ 
  up to order $M$ yields
  \[
    \log\mathbb{E}\!\left[
      e^{Q(\xi,\phi_{N,k})\,\zeta}
    \right]
    = \frac{1}{2}Q(\xi,\phi_{N,k})^2
      + \sum_{m=3}^{M}
        \frac{\kappa_m(\zeta)}{m!}\,
        Q(\xi,\phi_{N,k})^m
      + R_{N,k}^{(M)},
  \]
  where we used that $\kappa_1(\zeta) = 0$ and
  $\kappa_2(\zeta) = 1$, and the remainder satisfies
  \begin{equation}\label{eq:remainder}
    |R_{N,k}^{(M)}|
    \leq C_M\,|Q(\xi,\phi_{N,k})|^{M+1}.
  \end{equation}
  Summing over $k$ yields
  \begin{equation}\label{eq:log_S_sum}
    \log S\widetilde{\mu}_N(\xi)
    = -\frac{Q(\xi)}{2}
      + \frac{1}{2}\sum_{k=0}^{2^N-1}
        Q(\xi,\phi_{N,k})^2
      + \sum_{m=3}^{M}
        \frac{\kappa_m(\zeta)}{m!}
        \sum_{k=0}^{2^N-1}Q(\xi,\phi_{N,k})^m
      + \sum_{k=0}^{2^N-1}R_{N,k}^{(M)}.
  \end{equation}
Because $P_N$ is the orthogonal projection onto
  $\operatorname{span}\{\phi_{N,k}\}_{k=0}^{2^N-1}$, 
  we have
  \begin{equation}\label{eq:mid-eq}
    \sum_{k=0}^{2^N-1}Q(\xi,\phi_{N,k})^2
    = Q(P_N \xi),
    \qquad
    \sum_{k=0}^{2^N-1}Q(\xi,\phi_{N,k})^m
    = \lambda_m^N\int_0^1 (P_N \xi(t))^m\,dt.
  \end{equation}
  Substituting these identities 
  into~\eqref{eq:log_S_sum}
  and noting that 
  $Q(P_N\xi,(I-P_N)\xi) = 0$
  (since $\phi_{N,k}$ are real-valued),
  we obtain
  \begin{equation}\label{eq:log_S_truncated}
    \log S\widetilde{\mu}_N(\xi)
    = -\frac{Q((I-P_N)\xi)}{2}
      + \sum_{m=3}^{M}
        \frac{\kappa_m(\zeta)}{m!}\,\lambda_m^N
        \int_0^1 (P_N \xi(t))^m\,dt
      + \sum_{k=0}^{2^N-1}R_{N,k}^{(M)}.
  \end{equation}
  We now estimate each term in the expression 
  above.

\smallskip
  \textit{Covariance error.}
  Since $\xi \in \S_{\mathbb C}$, 
  by \cref{thm:projection},
  \begin{equation*}
    \begin{split}
      |Q((I-P_N)\xi)|\leq|(I-P_N)\xi|_{L^2}^2
      &= \sum_{n=N}^{\infty}\sum_{k=0}^{2^n-1}
         |\langle \xi, e_{n,k}\rangle_{L^2}|^2 \\
      &= \sum_{n=N}^{\infty}\sum_{k=0}^{2^n-1}
         \alpha_n^{-2p}\,\alpha_n^{2p}\,
         |\langle \xi, e_{n,k}\rangle_{L^2}|^2 \\
      &\leq \alpha_N^{-2p}
         \sum_{n=N}^{\infty}\sum_{k=0}^{2^n-1}
         \alpha_n^{2p}\,|\langle \xi, e_{n,k}\rangle_{L^2}|^2
      \leq \alpha_N^{-2p}\,|\xi|_p^2,
    \end{split}
  \end{equation*}
  for every $p \geq 0$.  Since $\alpha_N \sim 2^N$,
  choosing $p > (M-2)/4$ gives
  $\alpha_N^{-2p} = o(\lambda_M^N)$,
  and hence $|Q((I-P_N)\xi)| = o(\lambda_M^N)$.

  \smallskip
  \textit{Taylor remainder.}
  Using the bound~\eqref{eq:remainder} on the 
  remainder,
  along with the basic 
  estimate~\eqref{eq:linfty_bound}, yields
  \[
    \left|\sum_{k=0}^{2^N-1} R_{N,k}^{(M)}\right|
    \leq C_M\sum_{k=0}^{2^N-1}
      |Q(\xi, \phi_{N,k})|^{M+1}
    \leq C_M\,|\xi|_{L^\infty}^{M+1}\,
      2^N\,2^{-N(M+1)/2}
    = C_M\,|\xi|_{L^\infty}^{M+1}\,\lambda_{M+1}^N
    = o(\lambda_M^N).
  \]

  \smallskip
  \textit{Riemann sum approximation.}
  To bound the second term in the right-hand side
  of~\eqref{eq:log_S_truncated}, 
  let $3 \leq m \leq M$.
  For $a, b \in \mathbb{C}$, the identity
  $a^m - b^m = (a-b)\sum_{j=0}^{m-1}a^{m-1-j}b^j$
  and the bound
  $\sum_{j=0}^{m-1}|a|^{m-1-j}|b|^j
  \leq m\,\max(|a|,|b|)^{m-1}$ yield
  \begin{equation*}
    \begin{split}
      \left|\int_0^1 (P_N \xi)^m\,dt
            - \int_0^1 \xi^m\,dt\right|
      &\leq \int_0^1 
         |(P_N \xi)^m - \xi^m|\,dt \\
      &\leq m\int_0^1
         \max(|P_N \xi(t)|,|\xi(t)|)^{m-1}\,
         |P_N \xi(t) - \xi(t)|\,dt \\
      &\leq m\,|\xi|_{L^\infty}^{m-1}
         \int_0^1 
         |P_N \xi(t) - \xi(t)|\,dt \\
      &\leq m\,|\xi|_{L^\infty}^{m-1}\,
         |P_N \xi - \xi|_{L^2},
    \end{split}
  \end{equation*}
  where we used \cref{lem:projection_linfty} 
  to bound
  $|P_N \xi(t)| \leq |\xi|_{L^\infty}$ 
  for a.e.\ $t \in [0,1]$, 
  and Cauchy--Schwarz to pass from $L^1$ to
  $L^2$.  Arguing as in the estimate of the 
  covariance error,
  we can choose $p$ sufficiently large 
  ($p > (M-3)/2$) so that
  \[
    \lambda_m^N\,|P_N \xi - \xi|_{L^2}
    \leq \lambda_m^N\,\alpha_N^{-p}\,|\xi|_p
    \sim 2^{N(1-m/2-p)}
    = o(\lambda_M^N).
  \]
  Consequently,
  \[
    \lambda_m^N\int_0^1 (P_N \xi)^m\,dt
    = \lambda_m^N\int_0^1 \xi^m\,dt 
    + o(\lambda_M^N).
  \]

  \smallskip
  Substituting the estimates for all three terms
  into~\eqref{eq:log_S_truncated} and summing 
  over $m$ yields
  \[
    \log S\widetilde{\mu}_N(\xi)
    = \sum_{m=3}^{M}
      \frac{\kappa_m(\zeta)}{m!}\,\lambda_m^N\,
      S\mathfrak{U}_m(\xi) + o(\lambda_M^N),
  \]
  establishing \eqref{eq:cumulant-expansion}.
  The identification of the principal branch 
  is due to the fact that the right-hand side 
  goes to zero as $N\to \infty$.
\end{proof}

\subsection{Proof of \cref{thm:intro-universality} \textnormal{\emph{(ii)}}}\label{sec:leading-correction}

The cumulant expansion (\cref{thm:cumulant-expansion})
describes the deviation of $\widetilde{\mu}_N$ from the Gaussian fixed
point $\mu_\infty$ at the level of $S$-transforms.
We now promote this to a statement about the dual pairing
with arbitrary test functionals $\Phi \in (\S)$.

\begin{theorem}[Leading correction]\label{thm:leading-correction}
  Assume that $\zeta$ satisfies the 
  sub-Gaussian condition~\eqref{eq:subgaussian}.
  Fix $m \geq 3$, and suppose that the moments of $\zeta$
  match those of a standard Gaussian up to order $m-1$,
  that is,
  \begin{equation}\label{eq:moment-match}
    \mathbb{E}[\zeta^r] = \mathbb{E}[Z^r],
    \qquad r = 1, \ldots, m-1,
  \end{equation}
  where $Z \sim \mathcal{N}(0,1)$.
  Then, for every $\Phi \in (\S)$,
  \begin{equation}\label{eq:leading-correction}
    \pp{\widetilde{\mu}_N}{\Phi}-\pp{\mu_\infty}{\Phi}
    = \frac{\kappa_m(\zeta)}{m!}\,\lambda_m^N\,
      \pp{\mathfrak{U}_m}{\Phi}
    + o(\lambda_m^N),
    \qquad N \to \infty.
  \end{equation}
\end{theorem}

The theorem reveals a two-level universality structure.
First, the convergence rate $\lambda_m^N$ is universal:
it depends neither on the distribution of $\zeta$ nor
on the test functional $\Phi$.
Second, the functional form of the correction,
given by the pairing $\pp{\mathfrak{U}_m}{\Phi}$,
is likewise universal.
The only distribution-dependent quantity is the scalar
prefactor $\kappa_m(\zeta)/m!$, which encodes the leading
non-Gaussianity of $\zeta$ through its $m$-th cumulant.

The standing assumptions $\mathbb{E}[\zeta] = 0$ and
$\mathbb{E}[\zeta^2] = 1$ already imply the moment-matching
condition~\eqref{eq:moment-match} for $m = 3$
(since the first two Gaussian moments are
$\mathbb{E}[Z] = 0$ and $\mathbb{E}[Z^2] = 1$).
In this generic case, the leading correction is governed
by $\lambda_3 = 2^{-1/2}$, provided $\kappa_3(\zeta) \neq 0$.
More generally, by matching higher moments one can
systematically accelerate the convergence rate from
$\lambda_3^N$ to $\lambda_m^N$.

\begin{remark}[Observable coupling]\label{rem:observable-coupling}
  If $\Phi$ does not have a chaos component of order $m$,
  then it does not excite $\mathfrak{U}_m$, that is,
  the pairing $\pp{\mathfrak{U}_m}{\Phi}$ vanishes
  since $\mathfrak{U}_m$ lives entirely in the $m$-th
  Wiener chaos.  In this case, the leading term
  in~\eqref{eq:leading-correction} is zero, and the
  actual correction is governed by the next nonvanishing
  higher-order eigenmode.
  Thus, the eigenmode that controls the convergence rate
  for a given $\Phi$ depends not only on the increment
  distribution $\zeta$ (through its cumulants
  $\kappa_m(\zeta)$), but also on the particular way that
  the chaos structure of $\Phi$ couples with the spectral
  data of $\L$ through the eigenvectors $\mathfrak{U}_m$.
  Nevertheless, the convergence remains universal in the
  sense that only the eigenmodes
  $(\mathfrak{U}_m, \lambda_m)$ of $\L$ appear as
  possible corrections --- which mode is selected depends
  on $\zeta$ and $\Phi$, but the set of admissible
  corrections is determined entirely by the spectral data
  of $\L$.
\end{remark}

The proof of \cref{thm:leading-correction} is based on an
application of \cref{thm:S-convergence} to the
renormalized difference between $\widetilde{\mu}_N$ and $\mu_\infty$.
Accordingly, we must verify two properties of its
$S$-transform: pointwise convergence and a uniform growth
bound. The first follows from the cumulant expansion obtained in
\cref{thm:cumulant-expansion}, whereas the second 
is established using the auxiliary lemmas below.

\begin{lemma}\label{lem:g-minus-one}
  Under the assumptions of \cref{thm:leading-correction},
  define
  \[
    g(z)\coloneq e^{-z^2/2}\E[e^{z\zeta}],
    \qquad z\in\mathbb C.
  \]
  Then $g$ is entire, and there exists a constant $C>0$
  such that
  \[
    |g(z)-1|
    \le C\,|z|^m 
    e^{\frac{1+\sigma^2}{2}|z|^2},
    \qquad z\in\mathbb C.
  \]
\end{lemma}
\begin{proof}
    As in the proof of \cref{prop:FN-characterization}, the
    sub-Gaussian assumption implies that
    $z\mapsto \E[e^{z\zeta}]$ is entire, hence so is $g$.
    By noting that 
    $e^{z^2/2}$ is the moment generating function of the 
    standard Gaussian law, and that
    $\zeta$ matches its moments up to order
    $m-1$, we conclude that the Taylor coefficients of $\E[e^{z\zeta}]$ 
    and $e^{z^2/2}$ agree up to order $m-1$.  
    As a result,
    \[
      g(z)-1
      =
      e^{-z^2/2}\bigl(\E[e^{z\zeta}]-e^{z^2/2}\bigr)
      =
      O(z^m),
      \qquad z\to0,
    \]
    since $e^{-z^2/2}=O(1)$ as $z \to 0$.
    Therefore there exists $C_1>0$ such that
    \[
      |g(z)-1|\le C_1|z|^m
      \le C_1|z|^m e^{\frac{1+\sigma^2}{2}|z|^2},
      \qquad |z|\le 1.
    \]
    On the other hand, by \eqref{eq:subgaussian},
    \begin{equation*}
      |g(z)|= |e^{-z^2/2}\E[e^{z\zeta}]| 
      \le e^{|z|^2/2} e^{\sigma^2|z|^2/2}
      =e^{\frac{1+\sigma^2}{2}|z|^2}.
    \end{equation*}
    Hence, for $|z|\ge 1$,
    \[
      |g(z)-1|
      \le |g(z)|+1
      \le e^{\frac{1+\sigma^2}{2}|z|^2}+1
      \le 2|z|^m e^{\frac{1+\sigma^2}{2}|z|^2}.
    \]
    Combining the cases $|z|\le 1$ and $|z|\ge 1$, we obtain
    \[
      |g(z)-1|
      \le C\,|z|^m e^{\frac{1+\sigma^2}{2}|z|^2},
      \qquad z\in\mathbb C,
    \]
    where $C=\max\{C_1,2\}$.
\end{proof}

\begin{lemma}\label{lem:uniform-bound}
  Under the assumptions of \cref{thm:leading-correction},
  for every
  \[
    p>\max\Bigl\{\frac{1}{2},\frac{m-2}{4}\Bigr\},
  \]
  there exist constants $K_1,K_2>0$, independent of $N$,
  such that

  \begin{equation}\label{eq:FN-minus-one-bound}
    |F_N(\xi)-1|
    \le K_1\,\lambda_m^N\,
      \exp\bigl(K_2|\xi|_p^2\bigr),
    \qquad \xi\in \S_{\mathbb C},\ \ N\ge1.
  \end{equation}
\end{lemma}
\begin{proof}
  Fix $N\ge1$, $\xi\in \S_{\mathbb C}$, and set
  $g(z)\coloneq e^{-z^2/2}\E[e^{z\zeta}]$.
  By \eqref{eq:FN-def}, we have  
    \begin{equation}
    F_N(\xi)
    =
    e^{-Q(\xi)/2}
    \prod_{k=0}^{2^N-1}
    e^{Q(\xi,\phi_{N,k})^2/2}
    g\bigl(Q(\xi,\phi_{N,k})\bigr). 
    \end{equation}
  The bilinear form $Q$ behaves in the same way as the $L^2$ inner
  product with respect to the orthogonal projection $P_N$ 
  defined in \eqref{eq:PN-def}. 
  In particular, we have the following identities,
  \[
    Q(P_N\xi)
    =
    \sum_{k=0}^{2^N-1} Q(\xi,\phi_{N,k})^2,
    \qquad
    Q(\xi)=Q(P_N\xi)+Q((I-P_N)\xi).
  \]
  Therefore the Gaussian factor can be rewritten as
    \begin{equation*}
    e^{-Q(\xi)/2}\prod_{k=0}^{2^N-1} e^{Q(\xi,\phi_{N,k})^2/2}
    =
    e^{
      -\frac12 Q(\xi)
      + \frac12 \sum_{k=0}^{2^N-1} Q(\xi,\phi_{N,k})^2
    } 
    =
    e^{
      -\frac12\bigl(Q(\xi)-Q(P_N\xi)\bigr)
    } 
    =
    e^{
      -\frac12 Q((I-P_N)\xi)
    }.
    \end{equation*}
  Hence $F_N(\xi)-1$ splits into
  \begin{equation}\label{eq:FN-split}
  F_N(\xi)-1
  =
  \left(
    e^{-\frac12 Q((I-P_N)\xi)}-1
  \right)
  \prod_{k=0}^{2^N-1}
  g\bigl(Q(\xi,\phi_{N,k})\bigr)
  +
  \left(
    \prod_{k=0}^{2^N-1}
    g\bigl(Q(\xi,\phi_{N,k})\bigr)-1
  \right).
  \end{equation}
  We estimate each of the contributing terms separately.

  \smallskip
  \emph{First term.}
  Using the basic inequality $|e^z - 1| \leq |z|\,e^{|z|}$ 
  we obtain
  \begin{equation}\label{eq:7-1}
    \left|e^{-Q((I-P_N)\xi)/2} - 1\right|
    \leq \frac{1}{2}|Q((I-P_N)\xi)|\,
    e^{\frac{1}{2}|Q((I-P_N)\xi)|}.
  \end{equation}
  By observing that $|Q(\eta)|\le |\eta|_{L^2}^2$
  and using \cref{thm:projection}, we argue as in the
  covariance error estimate in the proof of
  \cref{thm:cumulant-expansion} to obtain
  \[
    |Q((I-P_N)\xi)|
    \le |(I-P_N)\xi|_{L^2}^2
    \le \alpha_N^{-2p}|\xi|_p^2
    \leq \lambda_m^N\,|\xi|_p^2,
  \]
  for $p>(m-2)/4$.
  Substituting this bound in \eqref{eq:7-1} and using $\lambda_m^N \leq 1$
  yields
  \begin{equation}\label{eq:first-piece}
    \left|e^{-Q((I-P_N)\xi)/2} - 1\right|
    \leq \frac{1}{2}\lambda_m^N|\xi|_p^2\,
      e^{|\xi|_p^2/2}.
  \end{equation}
  
  \smallskip
  \emph{Second term.}
  By the sub-Gaussian condition \eqref{eq:subgaussian} we have
  \begin{equation}\label{eq:g-bound}
    |g(z)|= |e^{-z^2/2}\E[e^{z\zeta}]| 
    \le e^{|z|^2/2} e^{\sigma^2|z|^2/2}
    =e^{\frac{1+\sigma^2}{2}|z|^2}.
  \end{equation}
  Hence
  \begin{equation}\label{eq:7-2}
    \prod_{k=0}^{2^N-1}|g(Q(\xi,\phi_{N,k}))|
    \leq \prod_{k=0}^{2^N-1}
    e^{\frac{1+\sigma^2}{2}|Q(\xi,\phi_{N,k})|^2}
    = e^{\frac{1+\sigma^2}{2}
    \sum_{k=0}^{2^N-1}|Q(\xi,\phi_{N,k})|^2}.
  \end{equation}
  By Bessel's inequality applied to the orthonormal family
  $\{\phi_{N,k}\}$ in $L^2([0,1])$ (using
  $Q(\xi,\phi_{N,k}) = \langle\xi,\phi_{N,k}\rangle_{L^2}$
  since $\phi_{N,k}$ is real),
  \[
    \sum_{k=0}^{2^N-1}|Q(\xi,\phi_{N,k})|^2
    \leq |\xi|_{L^2}^2
    \leq |\xi|_p^2.
  \]
  Substituting into \eqref{eq:7-2} yields
  \begin{equation}\label{eq:some-piece}
    \prod_{k=0}^{2^N-1}|g(Q(\xi,\phi_{N,k}))|
    \leq e^{\frac{1+\sigma^2}{2}|\xi|_p^2}.
  \end{equation}
  
  \smallskip
  \emph{Third term.}
  Adding and subtracting partial products, we obtain the
  telescoping identity
\begin{align*}
  \prod_{k=0}^{2^N-1} g(Q(\xi,\phi_{N,k})) - 1
  &= \sum_{j=0}^{2^N-1}
    \left(
      \prod_{k=0}^{j}g(Q(\xi,\phi_{N,k}))
      - \prod_{k=0}^{j-1}g(Q(\xi,\phi_{N,k}))
    \right) \\
  &= \sum_{j=0}^{2^N-1}
    \bigl(g(Q(\xi,\phi_{N,j})) - 1\bigr)
    \prod_{k=0}^{j-1}g(Q(\xi,\phi_{N,k})).
\end{align*}
Reasoning as before, we use \eqref{eq:g-bound} and Bessel's
inequality to obtain
\begin{equation}\label{eq:lemma-1}
  \left|\prod_{k=0}^{2^N-1} g(Q(\xi,\phi_{N,k})) - 1\right|
  \leq \left(\sum_{k=0}^{2^N-1}
  |g(Q(\xi,\phi_{N,k})) - 1|\right)
  e^{\frac{1+\sigma^2}{2}|\xi|_p^2}.
\end{equation}  
We estimate the sum by using \cref{lem:g-minus-one},
\begin{equation}\label{eq:lemma-2}
  \begin{split}
    \sum_{k=0}^{2^N-1}|g(Q(\xi,\phi_{N,k}))-1|
    &\leq
    C\sum_{k=0}^{2^N-1}
    |Q(\xi,\phi_{N,k})|^m
    e^{\frac{1+\sigma^2}{2}|Q(\xi,\phi_{N,k})|^2} \\
    &\leq
    C e^{\frac{1+\sigma^2}{2}|\xi|_p^2}
    \sum_{k=0}^{2^N-1}|Q(\xi,\phi_{N,k})|^m \\
    &=
    C\,
    e^{\frac{1+\sigma^2}{2}|\xi|_p^2}
    \lambda_m^N\int_0^1|P_N\xi(t)|^m\,dt,
  \end{split}
\end{equation}
where the last equality follows by the same 
computation as in \eqref{eq:mid-eq}.
\cref{lem:projection_linfty} gives
$|P_N\xi|_{L^\infty}\le |\xi|_{L^\infty}$, while
$|P_N\xi|_{L^2}\le |\xi|_{L^2}$ since $P_N$ is an orthogonal projection.
Since $\S_p$ embeds into $L^\infty([0,1])$ for $p>1/2$,
we conclude that
\begin{equation}\label{eq:lemma-3}
  \int_0^1|P_N\xi(t)|^m\,dt
  \leq
  |P_N\xi|_{L^\infty}^{m-2}|P_N\xi|_{L^2}^2
  \leq
  C|\xi|_p^m.
\end{equation}
Combining \eqref{eq:lemma-1}, 
\eqref{eq:lemma-2}, and \eqref{eq:lemma-3} yields
\begin{equation}\label{eq:second-piece}
  \left|\prod_{k=0}^{2^N-1} g(Q(\xi,\phi_{N,k})) - 1\right|
  \leq
  C\,e^{(1+\sigma^2)|\xi|_p^2}\lambda_m^N|\xi|_p^m .
\end{equation}

\smallskip
\emph{Final step.}
Returning to \eqref{eq:FN-split}, we combine the bounds
\eqref{eq:first-piece},
\eqref{eq:some-piece}, \eqref{eq:second-piece} to obtain
\[
  \begin{aligned}
  |F_N(\xi)-1|
  &\leq
  \frac{1}{2}\lambda_m^N|\xi|_p^2\,
  e^{|\xi|_p^2/2}
  e^{\frac{1+\sigma^2}{2}|\xi|_p^2}
  +
  C\,e^{(1+\sigma^2)|\xi|_p^2}\lambda_m^N|\xi|_p^m \\
  &\leq
  C_1\,\lambda_m^N
  \bigl(|\xi|_p^2+|\xi|_p^m\bigr)
  e^{C_2|\xi|_p^2}.
  \end{aligned}
\]
Since the polynomial factor can be absorbed into the exponential,
this yields \eqref{eq:FN-minus-one-bound}.  
\end{proof}

We now give the proof of the main result.

\begin{proof}[Proof of \cref{thm:leading-correction}]
  Let
\begin{equation}
  \mathfrak V^{(m)}_N \coloneq 
  \lambda_m^{-N}\bigl(\widetilde\mu_N-\mu_\infty\bigr)\in(\mathcal S)^\ast,
  \qquad
  \mathfrak V^{(m)} \coloneq \frac{\kappa_m(\zeta)}{m!}\,\mathfrak U_m \in(\mathcal S)^\ast.
\end{equation}
Then \eqref{eq:leading-correction} is equivalent to 
$\pp{\mathfrak V^{(m)}_N}{\Phi}\to\pp{\mathfrak V^{(m)}}{\Phi}$ for every $\Phi\in(\mathcal S)$,
so it is enough to show that $\mathfrak V^{(m)}_N\to\mathfrak V^{(m)}$ strongly in $(\mathcal S)^\ast$.
To this end, we verify conditions (i) and (ii) of 
\cref{thm:S-convergence}, which will then imply 
the desired result.

(i) 
Since the $r$-th cumulant is a polynomial in 
the first $r$ moments, the moment-matching 
condition \eqref{eq:moment-match} 
implies the vanishing of the cumulants:
$\kappa_r(\zeta)=0$ for $3\le r\le m-1$. 
As a result, the cumulant expansion 
(\cref{thm:cumulant-expansion})
at order $M=m$ reduces to
\begin{equation*}
  \log S\widetilde\mu_N(\xi)=\frac{\kappa_m(\zeta)}{m!}\,\lambda_m^N\,S\mathfrak U_m(\xi)+o(\lambda_m^N).
\end{equation*}
Since $S\mu_\infty\equiv 1$, taking the exponential and Taylor expanding to first order yields
\begin{equation*}
  \begin{split}
    S\mathfrak V^{(m)}_N(\xi)=\lambda_m^{-N}\bigl(S\widetilde\mu_N(\xi)-1\bigr)
    &=\frac{\kappa_m(\zeta)}{m!}\,S\mathfrak U_m(\xi)+o(1)\\[0.2cm]
    &=S\mathfrak V^{(m)}(\xi) +o(1).
  \end{split}
\end{equation*}
Hence $S\mathfrak V^{(m)}_N(\xi)\to S\mathfrak V^{(m)}(\xi)$ for every $\xi\in\mathcal S_{\mathbb C}$,
and condition (i) holds.
 
(ii) By \cref{lem:uniform-bound} we have
$$
|S\mathfrak V^{(m)}_N(\xi)|=\lambda_m^{-N}|F_N(\xi)-1|\le K_1\exp(K_2|\xi|_p^2),
$$
for large enough $p$ and uniformly in $N$.
This verifies condition (ii) and completes the proof.
\end{proof}

\section{Formal application to Quantum Field Theory (QFT)}\label{sec:qft}
As mentioned in the introduction, the eigenmodes $\mathfrak{U}_n$ have a 
canonical interpretation in QFT, where Wick-renormalized 
polynomials appear as interaction terms in the action 
functional. It is therefore natural to ask whether the 
present framework extends to such field-theoretic models, 
beyond the random-walk class. 
In this final discussion, we sketch formally how our program 
applies to Gibbs-type tilts of $\mu_\infty$
and show that the spectral 
data of $\mathscr{L}$ recovers the classical 
relevant/marginal/irrelevant classification 
of QFT power counting.

To make the connection with QFT more explicit, we 
adopt the standard field-theoretic notation for the 
remainder of this section.
We identify paths $W$ under $\mu_\infty$ 
with a one-dimensional Gaussian free field (GFF) $\phi$ on 
$[0,1]$ with Dirichlet--Neumann boundary conditions 
$\phi(0) = \partial_x\phi(1) = 0$; the Wiener 
measure $\mu_\infty$ thus coincides with the law 
of the Dirichlet--Neumann GFF in $d = 1$.
Under this identification, the distributional 
time derivative $\dot{W}$ becomes the 
spatial gradient $\partial_x \phi$, and the 
eigenmodes take the form 
$$\mathfrak{U}_n 
= \int_0^1 {:}(\partial_x \phi(x))^n{:}\,dx.$$

We consider the following family of Gibbs-type 
perturbations of the Brownian fixed point
\begin{equation}\label{eq:tilt}
    \mu^{(\varepsilon, n)}(d\phi) \propto
    e^{-\varepsilon \int_0^1 
      {:}(\partial_x\phi(x))^n{:}\,dx}\,
    \mu_\infty(d\phi),
\end{equation}
with $\varepsilon > 0$ a small coupling parameter. 
In the context of QFT, the family 
$\mu^{(\varepsilon, n)}$ describes a 
one-dimensional self-interacting field theory 
with derivative-coupling interaction of 
order $n$. 
The dependence on the slope 
$\partial_x\phi$ rather than the field $\phi$ 
alone is a common occurrence in interface 
models. In fact, $\mu^{(\varepsilon, n)}$ can 
be interpreted as the continuum analogue of 
the $\nabla\phi$ interface model, sometimes 
referred to as the Ginzburg--Landau effective 
interface model or the anharmonic 
crystal \cite{gawedzki1980rigorous,gawedzki1981renormalization,gawedzki1982renormalization,gawedzki1983block,funaki1997motion,miller2011fluctuations}.
It should be noted that the expression above 
is purely formal, since the tentative 
exponential density is too singular to define 
an honest probability measure on path space. 
In what follows, we do not attempt to give a 
rigorous construction of such measures, and 
thus we keep the discussion at a formal level.
We also note that, for an honest Gibbs-measure interpretation, one typically takes
\(n\) even, so that the interaction term is bounded from below.
Here we keep \(n\geq 1\) arbitrary since the discussion is purely
formal.

The measure $\mu^{(\varepsilon, n)}$ can 
be characterized dynamically from the perspective of 
stochastic partial differential equations (SPDEs),
as the formal invariant measure of the SPDE
\begin{equation*}
    \partial_t \phi = \partial_x^2 \phi
    + \varepsilon n\,\partial_x\left({:}(\partial_x\phi)^{n-1}{:}\right) 
    + \sqrt{2}\,\xi,
\end{equation*}
where $\xi$ is a space-time white noise, describing 
the Langevin dynamics associated to the potential
$\eps\int_0^1 {:}(\partial_x\phi(x))^n{:}\,dx$;
see \cite{da2003strong} and \cite[Section 13.7]{da2014stochastic}. 
Several values of $n$ are of particular interest:
when $n = 1$, 
the interaction is linear 
and we recover the 
classical stochastic heat equation (SHE). 
For 
$n = 2$, the tilt remains Gaussian, amounting to a modification of 
the GFF covariance. The corresponding SPDE then 
reduces to the SHE with an effective diffusion 
coefficient $1 + 2\varepsilon$. 
The case \(n = 3\) corresponds to a cubic derivative-coupling 
model. The associated SPDE contains a slope-quadratic
nonlinearity, reminiscent of Burgers-type structures.
Finally, $n = 4$ is the natural 
derivative-coupled analogue of the $\phi^4$ model, with 
interaction $\int_0^1 {:}(\partial_x \phi(x))^4{:}\,dx$.

We now study the behavior of the perturbed measures 
$\mu^{(\varepsilon, n)}$ under the 
RG operator $\mathcal{R}$.
Expanding the density in
\eqref{eq:tilt} in powers of $\varepsilon$, we obtain
\[
    e^{
        -\varepsilon
        \int_0^1 {:}(\partial_x\phi(x))^n{:}\,dx
    }
    =1 -\varepsilon\int_0^1 {:}(\partial_x\phi(x))^n{:}\,dx
    +
    O(\varepsilon^2),
\]
so that the perturbed measure admits the 
formal expansion
\[
    \mu^{(\varepsilon,n)}
    =
    \mu_\infty
    -
    \varepsilon \mathfrak U_n
    +
    O(\varepsilon^2),
\]
in the sense of distributions. 
Thus, at leading order, the Gibbs measures $\mu^{(\varepsilon,n)}$ 
can be written as a perturbation of the Gaussian fixed point
$\mu_\infty$ in the direction of the eigenmode $\mathfrak U_n$.
Whether or not this perturbation relaxes back to the fixed point
is determined, at the linearized level, by the associated eigenvalue
$\lambda_n = 2^{1-n/2}$.
Indeed, since $\mu_\infty$ is a fixed point of $\mathcal R$, 
and $\L \mathfrak U_n = \lambda_n \mathfrak U_n$, 
we may neglect higher order terms and iterate the RG operator to obtain 
\[
    \mathcal R^N[\mu^{(\varepsilon,n)}]
    \approx
    \mu_\infty
    -
    \varepsilon \L^N \mathfrak U_n
    =
    \mu_\infty
    -
    \varepsilon \lambda_n^N \mathfrak U_n.
\]

This immediately gives the formal stability 
classification of the perturbed measures, which 
can be described in light of the cumulant interpretation of 
\cref{sec:cumulant-interpretation}.
For $n=0$, the eigenvalue is $\lambda_0=2$.
This is the mass direction --- it changes the total mass 
of the Gaussian fixed point and is unstable.
For $n=1$, the perturbation excites the first-order cumulant,
whose eigenvalue is $\lambda_1=\sqrt{2}$. Hence this direction is
also unstable. 
The $n=2$ mode represents a perturbation of the covariance and is
marginal, since $\lambda_2=1$. That is, under the linearized
dynamics, the perturbation $\mathfrak{U}_2$ neither grows nor shrinks under
iterations of $\L$. 
For $n\geq 3$, one has $\lambda_n<1$, and thus the corresponding
modes $\{\mathfrak U_n\}_{n\geq3}$ are stable perturbations of the
fixed point. They describe deviations in the higher-order
cumulants of the GFF, and their amplitudes shrink under the
linearized RG flow.
In this sense, for $n \geq 3$, the associated Gibbs measures
$\mu^{(\varepsilon,n)}$ formally belong to the basin of attraction
of $\mu_\infty$, or equivalently, they lie in the same Gaussian
universality class governed by $\mu_\infty$.
This extends the scope of the theory 
developed in this paper beyond the 
random-walk setting of \cref{sec:donsker}: 
the basin of attraction of $\mu_\infty$ 
contains not only random walks, but 
also, at a formal level, the field theories 
defined by $\mu^{(\varepsilon,n)}$, with the 
eigenpairs $(\lambda_n, \mathfrak{U}_n)$ 
characterizing a second layer of 
universality; they govern both the rate and 
the direction of convergence to the Gaussian 
fixed point.
\bigskip

The stability classification derived above 
from the spectral theory of $\L$ is 
traditionally obtained in QFT through 
dimensional analysis, also known as 
power counting, which we now sketch.
Since $\phi$ is a one-dimensional GFF, the 
scaling relation 
$\phi(cx) \law c^{1/2}\phi(x)$ 
holds under $\mu_\infty$, for $c>0$. 
The derivative thus obeys 
$(\partial_x\phi)(cx) \law 
c^{-1/2}\,(\partial_x\phi)(x)$, 
recovering the standard scaling of white 
noise in one dimension. Taking the $n$-th 
Wick power and integrating, we arrive at
\begin{equation}\label{eq:heuristic-scaling}
  \eps\int_0^c {:}(\partial_x\phi(x))^n{:}\,dx
  \;\law\;
  \eps c^{1-n/2}
  \int_0^1 {:}(\partial_x\phi(x))^n{:}\,dx.
\end{equation}
Observe that the Wick renormalization does 
not alter the scaling: each lower-order 
correction in ${:}(\partial_x\phi)^n{:}$ is 
a monomial in the field paired with a 
compensating power of the variance, so that 
every term carries the same total degree $n$.

The calculation above allows us to interpret 
the rescaling as a flow on the coupling.
Indeed, the scaling relation 
\eqref{eq:heuristic-scaling} induces a 
natural map on the coupling constant,
$$
\eps\mapsto \eps_{\mathrm{eff}}(c)
=\varepsilon c^{1-n/2},
$$
sending the initial coupling to the 
effective coupling at scale $c$, which 
measures the strength of the interaction 
as seen at that scale.
The large-scale effective behavior can be probed by 
setting $c > 1$, leading to the same 
classification as before: interactions with 
$n < 2$ persist at large scales and are 
\emph{relevant} 
(unstable), 
those with $n = 2$ are \emph{marginal}, 
and those with $n > 2$ become negligible at large scales
and are 
\emph{irrelevant} 
(stable).
This is in one-to-one correspondence with 
the stability classification derived above 
solely from the spectral theory of $\L$.
In other words, our approach can recover,
at a formal level, the classical QFT 
classification purely from 
a linear stability analysis in path space.

\paragraph{Acknowledgements}
The authors thank J.\ Bec, E.\ Simonnet, 
and S.\ Thalabard for stimulating discussions 
and valuable feedback.
AC thanks the financial support from the 
CAPES/MATH-AmSud project CHA2MAN.

\paragraph{Data availability}
Not applicable.

\paragraph{Conflict of interest}
The authors declare that they have no competing interests.


\bibliographystyle{unsrt}
\bibliography{biblio}

\end{document}